\newif\ifready\readytrue

\documentclass[11pt,letterpaper]{article}

\usepackage[utf8]{inputenc}
\usepackage[margin=1in]{geometry}
\usepackage{times}
\usepackage{amsmath,amssymb,amsfonts,amsthm}
\usepackage{mathtools}
\usepackage{graphicx}
\usepackage[colorlinks=true,allcolors=blue]{hyperref}
\usepackage[table,dvipsnames]{xcolor}
\usepackage{wrapfig}
\usepackage{tikz}
\usetikzlibrary{shapes, arrows, tikzmark, calc}
\usepackage{setspace}
\usepackage{nicefrac}
\usepackage[noend]{algpseudocode}
\usepackage[ruled,vlined,linesnumbered,algosection]{algorithm2e}
\usepackage[framemethod=tikz]{mdframed}
\usepackage{xspace, xparse}
\usepackage{pgfplots}
\usepackage{framed}
\usepackage{caption}
\usepackage{subcaption}
\usepackage{changepage}
\usepackage[capitalise,noabbrev,nameinlink]{cleveref}
\usepackage[shortlabels,inline]{enumitem}
    \pgfplotsset{compat=1.5}
\usepackage{thmtools, thm-restate}
\usepackage{typed-checklist}
\usepackage{multirow}
\usepackage{booktabs}
\usepackage{makecell}
\usepackage{comment}
\usepackage{float}
\usepackage{url}
\usepackage{nicematrix}

\ifready
\else 
\usepackage{color-edits}
\fi

\usepackage[backref=true, backend=biber, style=alphabetic, sortcites=true, sorting=alphabeticlabel,
minbibnames=3, maxbibnames=999, mincitenames=3, maxcitenames=3, minalphanames=3, maxalphanames=3]{biblatex}
\DeclareSortingTemplate{alphabeticlabel}{
  \sort[final]{\field{labelalpha}} %
  \sort{
    \field{year}   %
  }
  \sort{
    \field{title}  %
  }
}
\AtBeginRefsection{\GenRefcontextData{sorting=ynt}}
\AtEveryCite{\localrefcontext[sorting=ynt]}
\newcommand{\storelines}{\xdef\rememberedlines{\number\value{AlgoLine}}}
\newcommand{\restorelines}{\setcounter{AlgoLine}{\rememberedlines}}

\crefname{assumption}{assumption}{assumptions}
\Crefname{assumption}{Assumption}{Assumptions}
\crefname{step}{step}{steps}
\Crefname{step}{Step}{Steps}

\definecolor{mydarkblue}{rgb}{0,0.08,0.45}
\hypersetup{ %
    colorlinks=true,
    linkcolor=mydarkblue,
    citecolor=mydarkblue,
    filecolor=mydarkblue,
    urlcolor=mydarkblue,
}

\newcommand{\Lap}{\ensuremath{{\textsf{Lap}}}\xspace}

\DeclareMathOperator*{\poly}{poly}

  \newcommand{\eps}[0]{\ensuremath{\varepsilon}}
  \let\epsilon\eps

  \newcommand{\cD}{\ensuremath{{\mathcal D}}\xspace}

  \newcommand{\cI}{\ensuremath{{\mathcal I}}\xspace}

  \newcommand{\cP}{\ensuremath{{\mathcal P}}\xspace}
  \newcommand{\cQ}{\ensuremath{{\mathcal Q}}\xspace}
  
  \newcommand{\cS}{\ensuremath{{\mathcal S}}\xspace}

\newtheorem{theorem}{Theorem}[section]

\newtheorem{definition}[theorem]{Definition}
\newtheorem{lemma}[theorem]{Lemma}
\newtheorem{proposition}[theorem]{Proposition}
\newtheorem{claim}[theorem]{Claim}

\theoremstyle{definition}
\newtheorem*{remark}{Remark}

\newcommand{\algoref}[1]{\namedref{Algorithm}{alg:#1}}

\newcommand{\namedref}[2]{\hyperref[#2]{#1~\ref*{#2}}\xspace}

\usepackage{soul}

\NewDocumentEnvironment{pf}{o}
  {\IfNoValueTF{#1}{\begin{proof}}{\begin{proof}[Proof (#1)]}}
  {\IfNoValueTF{#1}{\end{proof}}{\end{proof}}}

\newcommand{\core}{k}

\newcommand{\kest}{\hat{\core}}

\mathchardef\hyph="2D

\newcommand{\alg}{\mathcal{A}}
\newcommand{\defn}[1]{\emph{\textbf{#1}}}

\newcommand{\sset}{\subseteq}
\DeclarePairedDelimiter{\card}{\lvert}{\rvert}
\DeclarePairedDelimiter{\norm}{\lVert}{\rVert}
\DeclarePairedDelimiter{\set}{\lbrace}{\rbrace}
\DeclarePairedDelimiter{\iprod}{\langle}{\rangle}
\DeclarePairedDelimiter\floor{\lfloor}{\rfloor}
\DeclarePairedDelimiter{\ceil}{\lceil}{\rceil}

\ifready
\else

\fi

\newcommand{\opt}{\mathrm{OPT}}
\newcommand{\prob}{\mathbf{P}}
\newcommand{\expect}{\mathbb{E}}

\DeclareMathOperator{\den}{den}
\DeclareMathOperator{\Range}{Range}
\DeclareMathOperator{\OPT}{OPT}

\DeclareMathOperator{\Count}{count}

\newcommand{\R}{\mathbb{R}}
\newcommand{\Z}{\mathbb{Z}}
\newcommand{\N}{\mathbb{N}}

\newcommand{\SVT}{\ensuremath{\textsc{SVT}}\xspace}

\newcommand{\mcal}{\mathcal}

\newcommand{\tO}{\tilde{O}}

\newcommand{\lf}{\eta}
\newcommand{\upexpold}{{(2+\eta)(1+\lf)^2}}
\newcommand{\probfactorminusone}{\frac{1}{n^{c-1}}}
\newcommand{\downexp}{(1+\eta)}
\newcommand{\lbexp}{\left(\downexp^{g'-2} - \frac{a\log^3 n}{\eps}\right) \left(1 -\frac{1}{(2+\eta)(1+\eta)^2}\right)}
\newcommand{\lbabrv}{J}

\newcommand{\myparagraph}[1]{\smallskip\noindent\textbf{#1.}}

\title{Sublinear Space Graph Algorithms\\in the Continual Release Model}
\author{Alessandro Epasto, Quanquan C. Liu, Tamalika Mukherjee, Felix Zhou}
\author{%
    \begin{tabular}{cc}	
        \begin{tabular}{c}
            Alessandro Epasto\\
            Google Research\\
            \texttt{aepasto@google.com}
        \end{tabular}
        & 	
        \begin{tabular}{c}
            Quanquan C. Liu\\
            Yale University\\
            \texttt{quanquan.liu@yale.edu}
        \end{tabular}
        \\
        \\
        \begin{tabular}{c}
            Tamalika Mukherjee\\
            Columbia University\\
            \texttt{tm3391@columbia.edu}
        \end{tabular}
         & 
         \begin{tabular}{c}
            Felix Zhou\\
            Yale University\\
            \texttt{felix.zhou@yale.edu}
        \end{tabular}
    \end{tabular}
}
\date{}

\begin{document}

\maketitle

\pagenumbering{gobble}

\begin{abstract}
    The graph continual release model of differential privacy seeks to produce differentially private solutions to graph problems 
    under a stream of edge updates where new private solutions are released after each update. 
    Thus far, previously known \emph{edge-differentially private} algorithms for most graph problems including densest subgraph and matchings in the continual release
    setting only output real-value estimates (not vertex subset solutions) and
    do not use sublinear space. 
    Instead, they rely on computing
    exact graph statistics on the input \cite{FichtenbergerHO21,Song18}.
    In this paper, 
    we leverage \emph{sparsification} to address the above shortcomings for edge-insertion streams.
    Our edge-differentially private algorithms use sublinear space with respect to the number of edges in the graph
    while some also achieve sublinear space in the number of vertices in the graph.
    In addition, for the densest subgraph problem, we also output edge-differentially private vertex subset solutions; no previous graph algorithms in the continual release model
    output such subsets.
    
    We make novel use of assorted sparsification techniques 
    from the non-private streaming and static graph algorithms literature
    {to achieve new results in the sublinear space, continual release setting.
    This includes algorithms for densest subgraph,
    maximum matching, as well as the first continual release $k$-core decomposition algorithm.}
    {We also develop a novel sparse level data structure for $k$-core decomposition
    that may be of independent interest.}
    To complement our insertion-only algorithms,
    we conclude with polynomial additive error 
    lower bounds for edge-privacy in the fully dynamic setting,
    where only logarithmic lower bounds were previously known.
\end{abstract}

\newpage
\setcounter{tocdepth}{2}
\tableofcontents
\newpage
\cleardoublepage
\pagenumbering{arabic}

\section{Introduction}

Today's data is not only massive in size but also rapidly growing.
As the world becomes increasingly digitized, the sensitive associations between individuals 
are becoming a key focus for data analysis. 
Such a focus also comes with privacy risks and the question of how to safeguard privacy.
Differential privacy~\cite{DMNS06} is the gold standard for privacy protection, with 
a rich body of work focusing on \emph{static} graphs~\cite{AU19,blocki2022privately,dinitz2024tight,DLRSSY22,ELRS22,kasiviswanathan2013analyzing,kalemaj2023node,
LUZ24,mueller2022sok,NRS07,RS16,raskhodnikova2016differentially,Upadhyay13}.
However, today's graphs are not static; they are rapidly expanding, with new connections between individuals forming every minute.
Streaming graph algorithms, a well-studied area in the non-private setting~\cite{assadi2022semi,AG11,EHW16,Esfandiari2018,feigenbaum2005graph,FMU22,HuangP19,henzinger1998computing,
muthukrishnan2005data,mcgregor2014graph,mcgregor2015densest,mcgregor2018simple},
are designed to handle such massively evolving datasets. In this model,
updates to the dataset are received in a stream. 
Online streaming algorithms continuously release accurate
graph statistics after each update~\cite{chen2023sublinear,halldorsson2016streaming,cormode2018approximating,GS24},
but these graph statistics may reveal private information.

Our paper focuses on studying graph problems in the well-known \emph{continual release} model~\cite{DworkNPR10,CSS11} of differential privacy 
that promises a strong guarantee for online streaming outputs: 
the entire vector of \emph{all outputs} satisfies differential privacy. 
The notion of graph privacy we consider is edge-differential privacy,
where we require the output distributions of the algorithm on edge-neighboring stream inputs to be ``close'',
where edge-neighboring streams differ in one edge update.
Despite the large body of differential privacy works on static graph algorithms, 
there are few works on graph continual release~\cite{FichtenbergerHO21,jain2024time,Song18} 
and related graph models~\cite{upadhyay2021differentially}.
While there has been substantial progress in graph continual release algorithms in recent years,
there are some drawbacks to existing algorithms, 
including: 
1) not releasing private (vertex subset) solutions in addition to their real-valued function outputs and 
2) requiring exact algorithms for computing graph statistics, thus necessitating linear space usage in the number of \emph{edges}.

In this paper, we take the first steps to address these shortcomings for insertion-only streams in the continual release model via the following contributions:

\begin{itemize}
    \item We give the first $k$-core decomposition algorithm in the continual release model.
    \item We formulate the first continual release graph algorithms that return edge-differentially private vertex subset solutions 
    {for densest subgraph}
    (in addition to the value of the solutions). 
    \item We give the first \emph{sublinear} space algorithms for general graphs to solve densest subgraph, $k$-core decomposition, and matching problems. 
\end{itemize}
The approximation guarantees of our algorithms nearly match the guarantees for their private static or non-private streaming counterparts.
Moreover, as is standard in streaming, all of our results hold for a single pass over the stream.

Resolving these issues 
{closes the gap between continual release graph algorithms and}
their non-private streaming counterparts for many applications. 
While the private streaming
literature has explored sublinear space outside the graph algorithms setting~\cite{ChanLSX12,epasto2023differentially,upadhyay2019sublinear, BlockiGMZ22}, 
it remains unaddressed in general graphs for most problems in continual release.
Furthermore, while streaming algorithms typically operate on datasets too large for memory,  practical applications often require information about each vertex, not just aggregate values.
For example, advertisers may want the members of the current densest community, not merely its density. Hence, achieving both goals--sublinear space combined with useful vertex information--is essential to the development of private graph algorithms.

\subsection{Sparsification}\label{subsec:sparsification-intro} 
The main 
algorithmic technique we use in this paper
is \emph{sparsification}, traditionally used in the streaming, static,
and dynamic graph algorithms literature. 
Sparsification aims to simplify a large graph by strategically selecting a smaller set of its edges. This smaller ``sketch'' of the graph still captures the essential information needed for a specific problem.
Such sparsifiers exist for 
a large number of problems in the non-private literature. 
In general, graph sparsifiers are \emph{problem-specific},
where each sparsifier provides guarantees for only one or two closely related problems.
We design novel methods
for creating differentially private problem-specific graph sparsifiers in continual release, 
which we describe in more detail in the technical overview (\cref{sec:tech-overview}).

Using sparsification techniques,
we design a number of {new} continual release algorithms 
that achieve {the same} utility guarantees {as previous algorithms} for the problem in the static DP and non-private streaming settings
{up to logarithmic factors}. 
Our sparsification techniques directly lead to space savings in the edge-privacy setting
for a variety of problems 
including densest subgraph, $k$-core decomposition, and maximum matching. 
We return vertex subsets for 
the densest subgraph problem, 
approximate $k$-core numbers for every vertex in core decomposition, 
and estimates for the size of the maximum matching,
all using sublinear space in the \emph{number of edges}. For maximum matching, we even use space that is sublinear in the number of vertices.\footnote{While all non-private
streaming algorithms use space that is sublinear in the number of total edges, very few use space sublinear in the number of vertices, although results sublinear in the number of vertices are most desirable.}

Finally, to complement our results on insertion-only streams,
we demonstrate the difficulty of obtaining algorithms in fully dynamic streams 
via improved polynomial error lower bounds for maximum cardinality 
matching,
counting connected components,
and triangle counting 
based on reductions from inner product queries. 
Similar to the recent continual release lower bounds for counting distinct elements given in~\cite{jain2023counting},
our lower bounds support the observation from previous works that the fully dynamic setting results in significantly
more error~\cite{Song18,FichtenbergerHO21} than the insertion-only setting within the continual release model.

\subsection{Our Contributions}
Our paper gives the following set of results, stated informally here and given in more precise terms in their respective 
sections. We group our results by graph problems 
for $n$-vertex $m$-edge graph streams over $T$ time steps
and our multiplicative guarantees
are given in terms of a constant factor $\eta > 0$.

For our algorithmic results,
we consider insertion-only streams where {(possibly empty)} edge insertions
are given in the stream.
{We may assume without loss of generality that each non-empty edge is inserted at most once,
as we can treat an insertion operation for an existing edge as an empty insertion.}
{Without requiring privacy,
this closely resembles the well-studied online streaming setting (see related works in \Cref{sec:related-works}).
Our algorithms also incur logarithmic additive noise,
which is necessary,
as demonstrated by lower bounds that hold even in the static model (see \Cref{table:results}).}

Our lower bounds are given for fully dynamic streams
where edges can be inserted and deleted multiple times,
also known as the \emph{turnstile} model.
{In this setting,
we show that polynomial additive error is necessary,
demonstrating a fundamental difficulty in maintaining small errors in the presence of both insertions and deletions.}

All of our results in the continual release model 
are given for edge-neighboring streams where the two streams differ
by one edge update; 
results in the static DP setting are given for edge-neighbors
(graphs that differ by one edge). 
Throughout the paper, we will use the phrase 
\emph{$\eps$-edge DP} to refer to these settings.
Our algorithms produce $(\beta, \zeta)$-bicriteria approximations where $\beta$ is the multiplicative
factor in the approximation and $\zeta$ is the additive error. Our approximation bounds hold, with high probability,
for the released solution after \emph{every} update in the stream. Throughout the paper, we assume $n \le T = O(\poly(n))$.
This is reasonable as $O(n^2)$ (non-empty) updates are sufficient to have a complete graph.\footnote{Our results 
can be modified to handle streams of length $T = \omega(\poly(n)).$} We 
compare our results with static edge-DP lower bounds 
as such lower bounds also hold in the insertion-only continual release setting.\footnote{Otherwise, if we obtain better error in the continual
release setting, we can solve the static problem with better error by inserting all of the edges of the static graph and taking the 
solution released at the end of the stream.}
A more in-depth discussion of related works can be found in~\cref{sec:related-works}. Our entire set of results can be found in~\cref{table:results},
including baseline algorithms and lower bounds in the static edge-DP and non-private streaming settings. %

\begin{table}[htbp!]
\centering
\caption{Our results are highlighted in green. Non-private bounds are shown in gray. 
All results are for edge-privacy. Bicriteria approximations $(\beta, \zeta)$ are given, where $\beta$ is the multiplicative factor and $\zeta$ is the additive error. $\tilde{\alpha}$ is a public upper bound on the maximum arboricity of the stream; our approximation guarantees for maximum matching hold when $\tilde{\alpha}$ upper 
bounds the private arboricity and our privacy guarantees \emph{always} hold. $\alpha$ in the non-private streaming bound is the maximum arboricity in the stream.}
\label{table:results}
\resizebox{0.95\textwidth}{!}{
\begin{NiceTabular}{|c|c|c|c|c|c|}
\noalign{\global\arrayrulewidth=0.2mm}
\hline
Problem & Space & Approximation & Reference & Model & Solution Type \\
\noalign{\global\arrayrulewidth=0.8mm} 
\hline
\noalign{\global\arrayrulewidth=0.2mm}
\multirow{4}{*}{\makecell[c]{$k$-Core\\ Decomp.\ }} & \cellcolor{YellowGreen!20}$\tilde{O}\left( \frac{n}{{\eta^4}\eps} \right)$ & \Gape[0pt][2pt]{\cellcolor{YellowGreen!20}}$\left( 2+\eta, O\left( \frac{\log^3(n)}{{\eta^2}\eps} \right) \right)$ & \cellcolor{YellowGreen!20}\Cref{thm:kcore-formal} & \cellcolor{YellowGreen!20}\makecell[c]{Insertion\\Only} & \cellcolor{YellowGreen!20}Upper Bound \\ \cline{2-6}
 & $\Theta(m + n)$ & $\left(1, O\left( \frac{\log(n)}{\eps} \right) \right)$ & \cite{DLL23} & Static & Upper Bound \\ \cline{2-6}
 & & $\left( \beta, \Omega\left( \frac{\log n}{\eps\beta} \right)\right)$ & \cite{HSZ24} & Static & Lower bound \\ \cline{2-6}
  & \cellcolor{gray!20} $\tilde{O}\left( \frac{n}{\eta^2} \right)$ & \cellcolor{gray!20} $(1 + \eta, 0)$ & \cellcolor{gray!20} \cite{Esfandiari2018} & \cellcolor{gray!20} \makecell[c]{Non-\\Private\\Streaming} & \cellcolor{gray!20} Upper Bound \\ \hline
\multirow{8}{*}{\makecell[c]{Densest\\Subgraph}} & \cellcolor{YellowGreen!20}$\tilde{O}\left( \frac{n}{{\eta^2}\eps} \right)$ & \cellcolor{YellowGreen!20}$\left( 2+\eta, O\left( \frac{\log^2(n)}{{\eta}\eps} \right) \right)$ & \cellcolor{YellowGreen!20}\Cref{thm:densest-subgraph} & \cellcolor{YellowGreen!20}\makecell[c]{Insertion\\Only} & \cellcolor{YellowGreen!20}\makecell[c]{Upper Bound\\Vertex Subset} \\ \cline{2-6}
 & \cellcolor{YellowGreen!20}$\tilde{O}\left( \frac{n}{{\eta^5}\eps} \right)$ & \cellcolor{YellowGreen!20}$\left(1+\eta, O\left( \frac{\log^5(n)}{{\eta^4}\eps} \right) \right)$ & \cellcolor{YellowGreen!20}\Cref{thm:densest-subgraph better approx more space} & \cellcolor{YellowGreen!20}\makecell[c]{Insertion\\Only} & \cellcolor{YellowGreen!20}\makecell[c]{Upper Bound\\Vertex Subset} \\ \cline{2-6}
 & $\Theta(m + n)$ & $\left( 1+\eta, O\left( \frac{\log^4(n)}{\eta^3\eps} \right) \right)$ & \cite{DLRSSY22} & Static & \makecell[c]{Upper Bound\\Vertex Subset} \\ \cline{2-6}
 & $\Theta(m + n)$ & $\left( 2, O\left( \frac{\log(n)}\eps \right) \right)$ & \cite{DLL23} & Static & \makecell[c]{Upper Bound\\Vertex Subset} \\ \cline{2-6}
 & $\Theta(m + n)$ & $\left( 2 + \eta, O\left( \frac{\log^2(n)}{\eta\eps} \right) \right)$ & \cite{dinitz2024tight} & Static & \makecell[c]{Upper Bound\\Vertex Subset} \\ \cline{2-6}
&  & $\left( \beta, \Omega\left( \frac1\beta \sqrt{\frac{\log(n)}\eps} \right) \right)$ & \cite{AHS21} & Static & Lower bound \\ \cline{2-6}
 & $\Theta(m)$ & $\left( 1+\eta, O\left( \frac{\log^2(n)}{\eta\eps} \right) \right)$ & \cite{FichtenbergerHO21} & \makecell[c]{Insertion\\Only} 
 & \makecell[c]{Upper Bound\\Density Only} \\ \cline{2-6}
 & \cellcolor{gray!20} $\tilde{O}\left( \frac{n}{\eta^2} \right)$ & \cellcolor{gray!20} $(1 + \eta, 0)$ &  \cellcolor{gray!20} \makecell[c]{\cite{mcgregor2015densest}\\\cite{EHW16}} & \cellcolor{gray!20} \Gape[0pt][2pt]{\makecell[c]{Non-\\Private\\Streaming}} & \cellcolor{gray!20} Upper Bound \\ \hline
\multirow{6}{*}{\makecell[c]{Maximum\\Matching \\ Size}} & \cellcolor{YellowGreen!20}$O\left( \frac{\log^2(n)\log(\tilde\alpha)}{{\eta^2}\eps} \right)$ & \cellcolor{YellowGreen!20}$\left((1+\eta)(2+\tilde{\alpha}), O\left( \frac{\log^2(n)}{{\eta}\eps} \right)\right)$ & \cellcolor{YellowGreen!20}\Cref{thm:edge-dp-matching} & \cellcolor{YellowGreen!20}\makecell[c]{Insertion\\Only} & \cellcolor{YellowGreen!20}\makecell[c]{Upper Bound\\Arboricity} \\ \cline{2-6}
 & $\Theta(m)$ & $\left(1+\eta, O\left( \frac{\log^2(n)}{\eta\eps} \right)\right)$ & \cite{FichtenbergerHO21} & \makecell[c]{Insertion\\Only} & Upper Bound \\ \cline{2-6}
 &  & $\left(1, \Omega\left( \log(T) \right)\right)$ & \cite{FichtenbergerHO21} & \makecell[c]{Insertion\\Only} & Lower bound \\ \cline{2-6}
 &  & \cellcolor{YellowGreen!20}$\left( 1, \Omega\left( \min\left( \sqrt{\frac{n}\eps}, \frac{T^{1/4}}{\eps^{3/4}}, n, T \right) \right) \right)$ & \cellcolor{YellowGreen!20}\Cref{thm:matching lower bound small epsilon} & \cellcolor{YellowGreen!20}\makecell[c]{Fully\\ Dynamic} & \cellcolor{YellowGreen!20}Lower Bound \\ \cline{2-6}
 &  & $\left( 1, \Omega\left( \log(T) \right) \right)$ & \cite{FichtenbergerHO21} & \makecell[c]{Fully\\ Dynamic} & \makecell[c]{Lower Bound}\\ \cline{2-6}
 & \cellcolor{gray!20} $O\left( \frac{\log(n)\log(\alpha)}{\eta^2} \right)$ & \cellcolor{gray!20} $\left((1+\eta)(2+\alpha), 0\right)$ & \cellcolor{gray!20} \cite{mcgregor2018simple} & \cellcolor{gray!20} \makecell[c]{Non-\\Private\\Streaming} & \cellcolor{gray!20} \makecell[c]{Upper Bound\\Arboricity} \\ \hline
\multirow{2}{*}{\makecell[c]{Triangle\\ Count}} &  & \cellcolor{YellowGreen!20}$\left( 1, \Omega\left( \min\left( \sqrt{\frac{n}\eps}, \frac{T^{1/4}}{\eps^{3/4}}, n, T \right) \right) \right)$ & \cellcolor{YellowGreen!20}\Cref{sec:further lower bounds} & \cellcolor{YellowGreen!20}\makecell[c]{Fully\\ Dynamic} & \cellcolor{YellowGreen!20}Lower Bound \\ \cline{2-6}
 &  & $\left( 1, \Omega\left( \log(T) \right) \right)$ & \cite{FichtenbergerHO21} & \makecell[c]{Fully\\ Dynamic} & \makecell[c]{Lower Bound}
\\ \hline
\makecell[c]{Connected\\Components\\Counting} &  & \cellcolor{YellowGreen!20}$\left( 1, \Omega\left( \min\left( \sqrt{\frac{n}\eps}, \frac{T^{1/4}}{\eps^{3/4}}, n, T \right) \right) \right)$ & \cellcolor{YellowGreen!20}\Cref{thm:connected components lower bound small epsilon} & \cellcolor{YellowGreen!20}\makecell[c]{Fully\\ Dynamic} & \cellcolor{YellowGreen!20}Lower Bound \\ \hline
\end{NiceTabular}}
\end{table}

\myparagraph{$k$-Core Decomposition (\cref{sec:kcore})}
We show the first $\eps$-edge DP algorithm for the $k$-core decomposition problem in the continual release model.
Our algorithm returns approximate $k$-core numbers at every timestep while using $o(m)$ space.
The additive errors of our algorithms {are near-optimal} as they match the recent {additive error} lower bounds up to logarithmic factors (see \Cref{table:results}).
{We note that these lower bounds hold even when allowing for constant multiplicative error.}

\begin{theorem}[See~\cref{thm:kcore-formal}]\label{thm:kcore}
    We obtain a $\left(2+\eta, O\left( \frac{\log^3(n)}{{\eta^2} \eps} \right) \right)$-approximate
    $\eps$-edge differentially private $k$-core decomposition algorithm that outputs core number estimates 
    in the insertion-only continual release model using $\tilde{O}\left( \frac{n}{\eta^4\eps} \right)$ space. 
\end{theorem}

{The essence of our $k$-core algorithm is a novel \emph{sparse} level data structure
that may be of independent interest beyond our work.
To the best of our knowledge,
such a data structure did not previously exist in the privacy or graph algorithms literature.}

\myparagraph{Densest Subgraph (\cref{sec:densest-subgraph})}
We obtain the first $\eps$-edge DP densest subgraph algorithms that return subsets
of vertices with near-optimal density in the continual release model. 
Our algorithms also use $o(m)$ space.
We emphasize that the approximation guarantees nearly match, 
up to logarithmic terms,
the guarantees of their private static or non-private streaming counterparts.
In turn,
these nearly match the information-theoretic lower bound up to logarithmic factors.
{Again,
we emphasize that these lower bounds hold even when allowing for constant multiplicative error.}
\begin{theorem}[See \Cref{thm:densest-subgraph}]\label{thm:densest-subgraph-intro}
    We obtain a $\left(2+\eta, O\left( \frac{\log^2(n)}{{\eta^2}\eps} \right) \right)$-approximate
    $\eps$-edge differentially private densest subgraph algorithm that outputs a subset of vertices
    in the insertion-only continual release model using $\tilde{O}\left( \frac{n}{\eta^2\eps} \right)$ space.
\end{theorem}
At the cost of slightly more additive error and space usage,
we can improve the multiplicative error.

\begin{theorem}[See \Cref{thm:densest-subgraph better approx more space}]\label{thm:densest-subgraph-better-approx-intro}
    We obtain a $\left(1+\eta, O\left( \frac{\log^5(n)}{{\eta^4}\eps} \right) \right)$-approximate
    $\eps$-edge differentially private densest subgraph algorithm that outputs a subset of vertices
    in the insertion-only continual release model using $\tilde{O}\left( \frac{n}{\eta^5\eps} \right)$ space.
\end{theorem}

\myparagraph{Maximum Matching (\cref{sec:matching})}
We give the following results for estimating the size of a maximum matching in bounded arboricity graphs. 
Arboricity, a measure of local sparsity in the input graph, is formally defined as the minimum number of forests into which a graph's edges can be partitioned.
In real-world graphs, the arboricity of a graph is typically $\poly(\log(n))$.
In this setting, $\tilde\alpha > 0$ is a given public bound on the maximum arboricity $\alpha$ of the graph given by the input stream.\footnote{
Such public bounds are commonly assumed in non-private streaming literature~\cite{CJMM17,mcgregor2018simple}.}
We also briefly sketch how to remove this assumption in \Cref{apx:guess arboricity} at the cost of more space
and a worse approximation guarantee.
The additive error of our algorithms matches the information-theoretic lower bound in bounded arboricity graphs up to logarithmic factors.

Our {matching} algorithm falls under the \emph{truly sublinear model}, 
where space usage is sublinear in the \emph{number of nodes} in the graph.
The truly sublinear model is the most difficult streaming model to obtain
algorithms, and very few graph algorithms, even in the non-private setting,
use truly sublinear space.
We give an edge-DP algorithm for estimating the size of a
maximum matching in bounded arboricity graphs. Our approximation guarantees hold when
our public bound $\tilde{\alpha}$ upper bounds the arboricity of the graph, i.e.\ 
$\tilde{\alpha} \geq \alpha$, while our privacy guarantees always hold.

\begin{theorem}[See~\cref{thm:edge-dp-matching}]\label{thm:edge-matching}
    Given a public bound $\tilde{\alpha}$, we
    obtain an $\eps$-edge differentially private approximate maximum cardinality matching algorithm that outputs an 
    approximate matching size using $O\left(\frac{\log^2(n)\log(\tilde\alpha)}{\eta^2 \eps}\right)$ space 
    in the insertion-only continual release model. 
    If $\tilde\alpha \geq \alpha$, where $\alpha$ is the (private) maximum arboricity of 
    the input graph, then our algorithm returns a $\left((1+\eta)(2+\tilde\alpha), O\left(\frac{\log^2(n)}{{\eta}\eps}\right)\right)$-approximation
    of the size of the maximum matching.
\end{theorem}

\myparagraph{Fully Dynamic Lower Bounds (\cref{sec:lower-bounds})}
Last but not least, 
we give improved lower bounds in the fully-dynamic setting in~\cref{sec:lower-bounds}
for estimating the size of a maximum matching
and even simpler problems like connected component counting and triangle counting.
{Our reductions follow a natural framework that generalizes to further problems (\Cref{sec:further lower bounds}).}

Note that previous lower bounds all show a $\Omega(\poly\log(T))$ lower bound in the additive error
while we strengthen these lower bounds to $\Omega(\poly(T))$,
yielding an exponential improvement. 
Comparing the upper bounds with our new 
lower bounds for the same problems yield polynomial separations in the error between the insertion-only and 
fully dynamic settings.

See \Cref{table:results} for a summary of our results and previous works.

\section{Technical Overview}\label{sec:tech-overview}
{In this section,
we provide a detailed overview of our techniques.
We begin in \Cref{sec:challenges} by illustrating the various algorithmic challenges with the graph continual release model.
\Cref{sec:detailed-techniques} then describes how we address these challenges.
}

\subsection{Challenges in Continual Release
}\label{sec:challenges}

\myparagraph{Returning Multiple Solutions}
The best one-pass streaming algorithms in the non-private setting for the problems we study~\cite{Esfandiari2018,mcgregor2018simple,mcgregor2015densest,EHW16}
return the solution only once, at the end of the stream. In the continual release setting, we must return a solution \emph{after every update}. This is a fundamental challenge, especially
when considering privacy.
Indeed, with a single solution release,
each edge update is intuitively used only once to produce the solution.
In contrast, releasing solutions multiple times means earlier
updates are potentially used many times to produce solutions for each release. Such overuse could result in greater loss of privacy.
Thus, for our continual release algorithms, we use novel techniques to simultaneously ensure we 
achieve edge-privacy for the entire vector of releases as well as ensure our utility bounds for each release.

{
Previous value-only continual release algorithms ensure privacy over multiple releases via the sparse vector technique (SVT).
Their algorithms release the outputs of one or more instances of SVT over the stream
and the estimates are a deterministic function of the SVT outputs.
Thus, the proof of privacy directly follows from the privacy guarantees of SVT.
The outputs of our algorithms cannot be directly derived from standard SVT as we need to output information about individual vertices (as part of 
vertex subsets or individual outputs for each vertex for $k$-core decomposition).
Moreover,
our outputs depend on
our sampling procedures, which in turn depend on private properties of the stream.
These subtleties prevent the use of straightforward analyses of privacy using composition.
Instead,
we must directly prove the privacy of our algorithms using the definition of edge-privacy and analyze the output distributions of our algorithms over the entire stream.
This is accomplished via careful conditioning arguments.
}

\myparagraph{Ensuring Sublinear Space}
Obtaining sublinear space continual release algorithms presents a new set of challenges. In particular, 
we must ensure that our sampling techniques are \emph{stable}.\footnote{This term was also used in node-private continual release literature
\cite{jain2024time}, i.e.\ edge-to-edge and node-to-edge stability.} Stable sampling algorithms use sampling techniques that do not cause our samples from
neighboring streams to differ by more than one edge update. For randomized sampling algorithms, this means that the 
distribution of sampled edges from edge-neighboring streams should be similar. Consider the following naive sampling scheme that is not stable.
Suppose we want to sample each edge insertion in an insertion-only stream with probability proportional to the number of edges we have seen so far. 
Thus, the $i$-th inserted edge will be sampled with probability $\nicefrac1i$.
Suppose we have two neighboring streams $S$ and $S'$ where the edge that differs
is the first update; thus, the first update in $S'$ is an edge insertion, while the first update in $S$ is $\bot$ 
(an empty update). 
Then, the $j$-th update in $S$, if it is an edge insertion, has
a higher probability of being sampled than the $j$-th update in $S'$. Such a discrepancy produces {significantly different output distributions when} sampling from neighboring streams,
which may lead to sampled graphs that significantly differ.

Hence, a major component of our algorithms and proofs is to formulate and prove stable sampling algorithms with similar output distributions 
on edge-neighboring streams. Common sampling approaches in the non-private streaming and sparsification literature do not immediately 
translate to stable algorithms in the continual release setting. 

\myparagraph{$k$-Core Decomposition} Consider the $k$-core decomposition problem which is defined in terms of the \emph{cores} of a graph. A $k$-core is a maximal induced subgraph where
every vertex in the induced subgraph has degree at least $k$. Then, a vertex $v$'s core number is the largest value of $k$ for
which $v$ is part of the $k$-core. In edge-neighboring graphs, the removal of one edge can change the core numbers of \emph{all} vertices.
Hence, we cannot use a function that computes the exact core number of each node in the continual release setting, as this can incur 
$\Omega(n)$ error through composition over the vertices.
In the continual release setting, the insertion of a single edge can alter the core number of one or more vertices; hence, 
we also cannot run a static differentially private
$k$-core decomposition algorithm each time a core number changes since we can potentially see $\Omega(n)$ core number changes incurred over all vertices. Calling 
the static DP algorithm $\Omega(n)$ times would result in $\Omega(n)$ error through composition over the calls. 
To overcome these challenges, we develop a novel $k$-core decomposition algorithm using several new techniques (\cref{sec:kcore}).

\myparagraph{Releasing Vertex Subsets} 
Now, consider the setting where we return vertex subset solutions. Unlike returning real-valued solutions, 
even a few updates that do not significantly alter the solution's value (e.g.\ the maximum density in the densest subgraph problem)
can still cause the entire set of vertices to change. This can happen even when the global sensitivity is small.
Such behavior differs from real-valued solutions, which are less susceptible to
large changes due to small updates. To give a concrete example, consider the densest subgraph problem where the 
goal is to obtain a subset of vertices with maximum induced density. Suppose we are given a number of edge insertions
resulting in two disjoint almost $k$-cliques, $A$ and $B$, that is, two $k$-cliques, each with 2 edges removed. 
A sequence of insertions—one into $A$ followed by two into $B$—could lead to large changes in the set of vertices in the densest subgraph.
After the first insertion, the densest subgraph consists of vertices in $A$; then, after the third insertion, the densest subgraph consists of vertices in $B$.
This sequence of only $3$ insertions causes a complete change in the optimal vertex subset. 
It is not hard to extend this example and show that we can force any $O(1)$-approximation algorithm to completely change its output with $O(1)$ insertions.
Because of this challenge {and the privacy requirement}, 
we allow for $(1+\eta)$-multiplicative approximations
along with $\poly(\log n)/\eps$ additive error
on the density of the subgraph {induced by} the vertex subsets we release.

To release a private densest subgraph, we can use a static $\eps$-edge DP densest subgraph algorithm that releases the vertex subset.
However, as mentioned above,
we cannot afford to release a new private vertex subset each time we
receive an update, even if the true vertex subset changes, since we would incur $\Omega(n)$ error from composition.
Thus, our algorithms in~\cref{sec:densest-subgraph} determine appropriate timesteps for releasing new vertex subsets.

\subsection{Detailed Description of Our Techniques}\label{sec:detailed-techniques}
All of our algorithms use
$\tO(n)$ space, i.e., near-linear in the number of vertices in the stream, and some are sublinear in the number of vertices. 
We present $\eps$-differentially private algorithms in 
edge-neighboring streams (\cref{def:neighboring-streams}), where the streams differ
by one edge update.

We now describe how we solve each of the aforementioned challenges in~\cref{sec:challenges} and obtain algorithms for 
a variety of problems in the next sections.
Each of these problems uses a unique sampling scheme as sparsifiers in graph literature are often problem-dependent.
We prove the privacy of every algorithm as well as 
the utility of each of our algorithms within their individual
sections. We summarize on a high level
our technical contributions for each of the following problems.

\myparagraph{$k$-Core Decomposition (\cref{sec:kcore})} 
Given a stream of updates $S$, the $k$-core decomposition problem asks for a number for each vertex after every update 
indicating the maximum value $k$ for which the vertex is contained in a $k$-core. 
A $k$-core is defined as a maximal set of 
vertices $U \subseteq V$ such that the degree of every vertex in the subgraph induced by $U$ is at least $k$. We return 
approximate core numbers for every vertex at every timestep $t$ that are 
$\left(2+\eta, O\left(\frac{\log^3(n)}{\eta^2 \eps}\right)\right)$-approximations of their true core numbers (see~\cref{thm:kcore}) while using $\tilde{O}\left( \frac{n}{\eta^4\eps} \right)$ total space.

We formulate a novel differentially private continual release algorithm for $k$-core decomposition that is loosely inspired by the level data
structure of the static $\left(2+\eta, O\left(\frac{\log^3(n)}{\eta^2 \eps}\right)\right)$-approximation algorithm of~\cite{DLRSSY22}. 
As described in~\cref{sec:challenges}, we cannot use the static private algorithm in a black-box manner since we need to return the value of \emph{every}
vertex at each timestep. Every vertex may change its core number many times throughout the duration of the stream. 
Thus, using any static $k$-core decomposition algorithm in a black-box manner incurs a privacy loss of a factor of 
$n$ via composition.

In the static setting, the level data structure algorithm for $k$-core decomposition partitions the vertices of the input graph across the 
levels based on the induced degree of each vertex among the vertices in the same or higher levels.
Then, the levels roughly partition the vertices into cores of the graph with higher levels containing larger valued 
cores and lower levels containing smaller valued cores. In the static setting, 
vertices move up levels using all of the edges in the graph for each move. However,
in the continual release setting, edges are inserted into the graph progressively
and vertices move up the levels as new edges are inserted. 
In the static setting, 
each vertex releases their level at most $\poly(\log n)$ times,
whereas in the continual release setting,
each vertex releases its approximate core number 
$\Omega(T) = \poly(n)$ times. 
While composition over $\poly(\log n)$ releases leads to $\poly(\log n)$ additive error, 
composition over $\Omega(n)$ releases in the continual release model would lead to $\Omega(n)$ error.

Due to the above challenge, we need to be able to have vertices release their approximate core numbers $\Omega(T)$ times
without losing a factor of $T$ in the additive error due to composition. This problem is not present in the static setting.
In our novel continual release algorithm, we use a variant of the \emph{sparse vector technique} adapted to the level data structure.
In particular, we use the idea of the \emph{multidimensional sparse vector 
technique} given in~\cite{DLL23} and adapt it to the 
continual release setting with a new formulation of the algorithm and a new proof 
allowing for adaptive
queries based on new (subsampled) edge insertions to the graph. We call this the \emph{adaptive multidimensional sparse vector technique}.

In the static setting, we can use the multidimensional sparse 
vector technique to 
determine when a vertex stops moving up levels
and the level in which a vertex resides
directly corresponds with an approximation of the core number of the vertex. 
Thus, once a vertex stops moving 
up levels in the static setting, 
it will immediately output a new approximation.
Hence, each vertex fails the SVT check at most once.
However, in the continual release setting, 
our usage of the adaptive multidimensional sparse vector technique requires that 
each vertex must output a value at every timestep. This means that the adaptive multidimensional sparse vector technique must
answer adaptive queries resulting from edge insertions that occur
after previous levels are released. 
Moreover, the mechanism needs to allow for $\poly(\log n)$ above-threshold checks for each vertex, indicating whether the
vertex moved up one of the $\poly(\log n)$ levels.

In addition to solving the above challenge, we further present a sublinear space algorithm for $k$-core decomposition, which
uses $\tilde{O}\left( \frac{n}{\eta^4\eps} \right)$ space. 
The main workhorse is a \emph{sparse} level data structure.
No sparsified versions of the level data structure existed in either the non-private graph literature or 
the private graph literature. While many non-private sparsification algorithms in the streaming model sample edges
uniformly at random, such sampling techniques are \emph{insufficient} for $k$-core decomposition. 
When all edges are sampled with the same fixed probability, vertices in smaller cores are expected to have no adjacent edges included in the sample. If no adjacent
edges are included in the sample for vertex $v$, then we cannot approximate $v$'s core number.
Hence, we employ a more sophisticated sampling algorithm; specifically, we sample each inserted 
edge in the stream with probability inversely proportional
to the level of its lower-level endpoint. This means that edges on higher levels are sampled with a smaller probability than edges on lower levels. 
In particular, such a sampling scheme ensures that edges in larger valued cores--which inherently contain more edges--are sampled using smaller probabilities than smaller valued cores. 
Hence, such a sampling scheme simultaneously ensures sublinear space while maintaining a large enough sample of adjacent edges to each vertex
to satisfy our approximation bounds on the core numbers of \emph{all} vertices.

Finally, we prove our approximation bounds using a careful Chernoff bound argument on our sampled sets of edges
that takes into account the noise resulting from our usage of DP mechanisms,
including the adaptive multidimensional sparse vector technique.
Our approximation bound uses an original analysis that shows that vertices on a given level in our sparsified structure will be on 
approximately the same level, with high probability, as in the non-sparsified structure. Thus, the approximation bounds that we 
obtain from our non-sparsified structure also hold for our sparsified structure.

\myparagraph{Densest Subgraph (DSG) (\cref{sec:densest-subgraph})} 
Given an edge-neighboring stream and timestamp $t$, the densest subgraph in $G_t$ is 
a subset of vertices $V_{\OPT}$ which maximizes the density of the induced subgraph,
$V_{\OPT} = \arg\max_{V' \subseteq V}\left(\frac{|E_t(V')|}{|V'|}\right)$. 
We return a subset of vertices $V_{approx}$
whose induced subgraph gives a $\left(1+\eta, O\left(\frac{\log^5(n)}{\eta^4 \eps}\right)\right)$-approximation of the 
density of the densest subgraph (see \Cref{table:results}) in $\tilde{O}\left( \frac{n}{\eta^5\eps} \right)$ space.
Our algorithm calls the static $\eps$-DP DSG algorithm of \cite{DLRSSY22} in a black-box manner.
Using a different static $\eps$-edge DP DSG algorithm \cite{DLL23}, 
we can obtain a $\left(2+\eta, O\left(\frac{\log^2(n)}{\eta \eps}\right)\right)$-approximation
while saving a polylogarithmic factor in space.

Our continual release densest subgraph algorithm takes inspiration from the 
non-private $(1+\eta)$-approximate sparsification algorithms of~\cite{EHW16,mcgregor2015densest}, although we must solve several crucial challenges (\cref{sec:challenges})
to ensure privacy.
Indeed, to ensure accurate approximation guarantees, we present a novel 
concentration bound proof based on an intricate Chernoff bound argument that accounts for errors introduced by our various private subroutines.
Previous works in the continual release model only release the \emph{value} of the densest subgraph and not a 
\emph{vertex subset}~\cite{FichtenbergerHO21,jain2024time}. 
Since the density of the densest subgraph has sensitivity $1$, previous
works use the sparse vector technique~\cite{dwork2009complexity, roth2010interactive, hardt2010multiplicative, lyu2017understanding}
and private prefix sums 
to release these values after adding appropriate Laplace noise. 
In our work, we release a private set of vertices at timestep $t$ whose induced subgraph is an approximation of the densest subgraph in 
$G_t$.\footnote{
the induced subgraph consists of all updates that occurred on or before timestep $t$}
We use the sparse vector technique to determine when to release a new private set of vertices, and 
we use one of the many existing differentially private densest subgraph algorithms~\cite{AHS21,DLL23,DLRSSY22,NV21,dinitz2024tight} 
to return a vertex subset in the sparsified graph. However, there are several challenges in simply returning 
the subgraph obtained from these static algorithms. First, compared to the non-private setting, we cannot 
obtain an exact densest subgraph in the sparsified graph. Hence, we must ensure that neither the additive nor multiplicative
error blows up when the subgraph obtained in the sparsified graph is scaled up to the original graph. 
To solve this issue, we ensure that the first subgraph we release has size (approximately) at least the additive error returned by the 
private static algorithms. 

There is another important challenge when making our algorithm use sublinear space.
Previous non-private algorithms~\cite{mcgregor2015densest,EHW16} sample using a \emph{fixed} probability because 1) they assume knowledge
of the total number of updates in the stream and 2) only produce the approximate densest subgraph once, at the end of the stream.
However, in our setting, we must release an approximate densest subgraph after each update. Furthermore, 
we do not know the total number of edges
in the stream. Thus, we adaptively {adjust} our probability of sampling as we see more edge insertions. 
{This must be carefully implemented since} a naive adjustment procedure may produce an unstable algorithm (as described in~\cref{sec:challenges}). 

Towards solving all of these challenges, our sparsification algorithm works as follows. 
We sparsify the stream
by sampling each edge in the stream with probability $p_t$. Earlier on in the stream, our probability of sampling should be 
set high enough in order to sample enough edges to ensure our concentration bounds.
However, as we see more edges, we reduce the probability of 
sampling to ensure our sublinear space bounds. 
In edge-neighboring streams,
such reductions in probability may not occur at the same 
time, thus leading to sparsified graphs that are \emph{not} stable. 
Hence, we use the sparse vector 
technique to determine when to reduce our probability of sampling an edge. 
We ensure privacy based on a careful conditioning on the sampling probabilities. Intuitively, the sparse vector technique ensures that we 
decrease our sampling probability at the same timestamps in edge-neighboring streams. Conditioning on the timestamps when the sampling probabilities are
reduced, we can show that our sampled streams are stable.

Unlike the static setting where the privacy analysis of algorithms that rely on an SVT instance can apply simple composition,
we {need to derive a} probabilistic proof ``from scratch'' since we must output a solution after every update
but only lose privacy on timestamps where the SVT answers ``above''. Furthermore, our analysis must take into account our usage
of two SVTs, one for determining the probabilities of sampling and the other for determining when we release a new vertex subset. The second use of 
SVT depends on the outputs of the first; that is, determining when to release a new vertex subset depends on our sampling probabilities. 
Such SVTs lead to errors that affect our approximation guarantees
and our careful Chernoff bound argument must account for these errors while maintaining our approximation guarantees.

\myparagraph{Maximum Cardinality Matching Size (\cref{sec:edge-matching})%
}
Given a stream of updates $S$, the maximum cardinality matching problem asks for a real-valued integer at each timestep $t$ equal to the 
size of the maximum matching in $G_t$. 
We give an algorithm that uses sublinear space in the number of vertices in the graph with
edge-DP guarantees in the continual release model. We provide approximation guarantees using a certain property of the input graph stream: the 
maximum arboricity of the stream. The arboricity of a graph is a measure of local sparsity. 
Formally, a graph's arboricity, $\alpha$, is the minimum number of forests into which its edges can be partitioned.
Our algorithm uses a public bound $\tilde\alpha$. If this public bound upper bounds the 
private maximum arboricity of the input stream, 
our algorithm gives approximation guarantees in terms of $\tilde\alpha$. 
We emphasize that our algorithm is always private for all inputs.
Specifically, if $\tilde\alpha$ upper bounds the private maximum arboricity of the input stream,
then we obtain a $\left((1+\eta)(2+\tilde\alpha), O\left(\frac{\log^2 n}{\eta \eps}\right)\right)$-approximation of the matching size. 

\sloppy
For this algorithm, we take inspiration from the sublinear space non-private \mbox{$((1+\eta)(2 + \alpha))$-approximate}
maximum matching algorithm of~\cite{mcgregor2018simple},
which also requires an upper bound $\alpha$ on the maximum arboricity of the stream.
Their approximation guarantees are given in 
terms of this upper bound. 
{To the best of our knowledge,
\cite{mcgregor2018simple} is the only truly sublinear matching algorithm beyond bipartite graphs that terminates after a single pass, as required by continual release.}

The key observation that~\cite{mcgregor2018simple} makes 
is that maintaining $\Omega(\log n/\eta^2)$ edge samples, sampled uniformly at random in the stream, and counting the number of 
edges in the sample which are adjacent to at most $\alpha + 1$ edges occurring later in the stream is sufficient for providing a good estimate
of the maximum matching size in terms of the arboricity upper bound. Edges that are adjacent to more than $\alpha + 1$ edges that 
occur later in the stream are discarded from the sample. Unfortunately, in edge-neighboring streams, we cannot
add noise to the counts of every edge in the sample 
since the sensitivity of such counts is $\Omega(\tilde\alpha)$ (our public bound). 
Instead, we show through a precise charging argument that the sensitivity of the size of the sample is bounded by $2$. 
Since the sensitivity of the size of the sample is bounded by $2$, we can use the sparse vector technique (SVT)
to determine when to %
release a new estimate when the estimate changes by a large amount. Using SVT results in noisy sample sizes which
requires a new analysis for the utility bounds. 

Previous approximation bound analyses~\cite{mcgregor2018simple} relied on the intuition
that edge samples approximately fall into buckets determined by the sampling probability. However, with the use of SVT, 
it is no longer clear which buckets the sampled edges fall into as SVT introduces an additive $O(\log(n)/\eps)$ noise.
Thus, to prove our approximation bounds, we first argue that SVT only causes an additive error of $\tilde{O}\left( \frac{\log^2(n)}{\eta\eps} \right)$ on the estimated
number of sampled edges; hence, initializing a sufficiently large valued 
bucket ensures that the bucket for sampled edges with noise does not deviate by too much from the bucket the sampled edges would fall
into without noise.
{The proof of this property requires a detailed casework analysis.}
We can further remove the assumption of $\tilde\alpha$ using additional space 
via the parallel guessing trick in~\cref{apx:guess arboricity}.

\myparagraph{Fully Dynamic Lower Bounds (\cref{sec:lower-bounds})}
We show polynomial lower bounds for the maximum matching, triangle counting, and connected components problems in fully dynamic streams. 
Our lower bounds reduce each of these graph problems via novel 
graph constructions to answering inner product queries~\cite{Dinur2003Revealing, De2012Lower,
Dwork2007Price,MMNW11}. The known lower bounds for inner product queries then apply to our problems.
We encode the secret dataset given by the inner product query problem via $\Theta(n)$ insertions to construct our 
graph. Then, we map the value of the answers to each of our graph problems to the original inner product query via the inclusion-exclusion principle.
Finally, we delete all insertions through edge deletions and repeat the process for the next query.
We remark that this technique is quite general and likely extends to many other natural graph problems.
Roughly speaking,
the only problem-specific requisite is that we can encode the bits of a secret dataset using the problem structure
and that we can easily flip the bits by adding and deleting edges.
For example,
we can encode each bit of a secret dataset as a (non-)edge in an induced matching
so that flipping a bit equates to (adding) deleting the (non-)edge.

\section{Related Works}\label{sec:related-works}
\myparagraph{Streaming Algorithms}
The streaming model of computation is a prominent model for large-scale data analysis that has been studied for multiple decades~\cite{morris1978counting}. In this model, one usually seeks space- and time-efficient algorithms that can process the stream on the fly without the need to store all data. A long line of work in this area includes classical results such as the Flajolet-Martin algorithm for counting distinct elements~\cite{flajolet1985probabilistic} and many others~\cite{flajolet2007hyperloglog, alon1996space}.  In the context of streaming graph algorithms, ideally one would like to obtain space sublinear in the number of vertices~\cite{mcgregor2018simple,HuangP19}; but for many graph problems it is necessary to work in the semi-streaming model settling on space near-linear in the number of vertices~\cite{assadi2022semi,AG11,,Esfandiari2018,feigenbaum2005graph,FMU22,mcgregor2015densest}. Online streaming algorithms release graph statistics after 
every update while maintaining the low space bounds~\cite{chen2023sublinear,halldorsson2016streaming,cormode2018approximating,GS24}. \
In the non-private streaming setting, there exists many works on graph algorithms, including approximating the size of the maximum matching~\cite{AG11,ALT21,Assadi22,FMU22,mcgregor2018simple}, vertex cover~\cite{assadi2019coresets}, densest subgraph~\cite{EHW16,mcgregor2015densest}, $k$-core decomposition~\cite{Esfandiari2018,KingTY23},
number of connected components~\cite{BerenbrinkKM14}, and others~\cite{assadi2022semi,AG11,feigenbaum2005graph,FMU22,HuangP19,henzinger1998computing,
muthukrishnan2005data,mcgregor2014graph}

\myparagraph{Continual Release Model}
In the context of streaming computation, the DP model of reference is the continual release model~\cite{DworkNPR10,CSS11}
where we require algorithms to abide by a strong privacy notion: an observer obtaining \emph{all} future outputs of the algorithm must in essence learn almost nothing about the existence of any single input. 
Since its introduction, this research area has received vast attention outside of graphs, including many
recent works (see e.g.\ \cite{ChanLSX12,henzinger2024unifying, henzinger2023differentially,fichtenberger2023constant,jain2023counting,JRSS21}).

Among the insertion-only continual release work, prior research has tackled classical estimation problems~\cite{CSS11,henzinger2024unifying,henzinger2023differentially}, as well as heavy hitters-related problems~\cite{ChanLSX12,epasto2023differentially}. 
More recent works tackle the fully dynamic continual release setting~\cite{jain2023counting, dupre2023differential, FichtenbergerHO21}. 
Particularly relevant is~\cite{FichtenbergerHO21}, which shows hardness results for graph estimation in fully dynamic streams. Our work contributes new hardness results in this new emerging area as well.

Most relevant to our paper is the literature on continual release algorithms in graphs~\cite{upadhyay2021differentially,FichtenbergerHO21,Song18,jain2024time}. In this area, \cite{Song18} studied graph statistics (degree distribution, subgraph counts, etc.) on bounded-degree graphs. \cite{FichtenbergerHO21} focused on estimating 
a number of graph statistics like maximum matching, triangle counting, and the density of the densest subgraph for both edge and node privacy; 
they provide approximation guarantees in terms of the maximum degree in the graph as well as lower bounds in insertion-only and fully dynamic streams.
\cite{jain2024time} explored counting problems in graphs (counting edges, triangles, stars, connected components) for node-privacy
and where privacy must hold for arbitrary graphs (e.g., graphs with arbitrary degrees).  They give time-aware projection algorithms that can transform 
any continual release algorithm that gives approximation guarantees for bounded degree graphs into a truly private algorithm on nearly bounded degree graphs.
Like in~\cite{jain2024time} for our algorithms that assume a public bound, say on the stream arboricity, this bound affects only the approximation guarantees but not the privacy claims, hence our algorithms satisfy their notion of \emph{truly private}~\cite{jain2024time}.

\myparagraph{Private Graph Algorithms}
The private literature on graph algorithms includes a large body of work on static graph algorithms (see e.g.\ \cite{AU19,blocki2022privately,dinitz2024tight,DLRSSY22,ELRS22,kasiviswanathan2013analyzing,kalemaj2023node,
LUZ24,mueller2022sok,NRS07,RS16,raskhodnikova2016differentially,Upadhyay13} and references therein). 
Aside from the problems we study, work in this area includes results on preserving graph cuts~\cite{AU19}, 
the stochastic block model recovery~\cite{chen2023private}, 
graph clustering~\cite{bun2021differentially,imola2023differentially}, 
and many other areas. %

Our paper focuses on the pure, $\eps$-DP setting. In this setting, there has been a variety of recent works that we study for static graphs.
The densest subgraph (DSG) problem has been extensively studied in the non-private streaming context~\cite{BHNT15,EHW16,mcgregor2015densest}. Our work builds on the results of~\cite{EHW16,mcgregor2015densest}, which show that edge sampling approximately preserves the density of the densest subgraph. 
In the context of privacy, beyond the already cited work of~\cite{FichtenbergerHO21} that focused on density estimation only, 
all other works are in the static DP or LEDP setting~\cite{NV21,AHS21,DLRSSY22,dinitz2024tight,DLL23}. 
In the $\eps$-DP setting, the best-known lower bound on the additive error of the densest subgraph problem is $\Omega(\sqrt{\eps^{-1}\log n})$~\cite{AHS21,NV21}. 
Comparatively, the best known upper bounds in the central $\eps$-DP setting are the 
$(2, O(\eps^{-1} \log n))$-approximation~\cite{DLL23}, $(2+\eta, O(\eps^{-1}\log^2 n))$-approximate~\cite{dinitz2024tight},
and the $(1+\eta, O(\eps^{-1}\log^4 n))$-approximation~\cite{DLRSSY22} upper bounds.
In the $\eps$-LEDP setting, the best known upper bounds are the $(2+\eta, O(\eta^{-1}\log^2 n))$-approximate~\cite{dinitz2024tight} and $O(2, O(\eps^{-1}\log n))$-approximate~\cite{DLL23} bounds.
Our paper presents the first continual release algorithm for releasing an actual approximate densest subgraph (as opposed to estimating its density) in edge-private
insertion-only streams. %

Related to the densest subgraph problem is the $k$-core decomposition of the graph for which some streaming results are known~\cite{Esfandiari2018,KingTY23,SGJ13}. In the private setting, this problem has been tackled by~\cite{DLRSSY22,DLL23,HSZ24} where the best error bound achieved in both the central $\eps$-DP and 
$\eps$-LEDP settings provide $(1, O(\eps^{-1}\log n))$-approximations of the core number~\cite{DLL23}. This is tight against the recently shown lower bound 
of $\Omega(\eps^{-1}\log n)$~\cite{HSZ24}.
In our paper, we provide the first continual release algorithm for approximating the core number of nodes in a graph.

\subsection{Concurrent Work}
Recent concurrent work of \textcite{raskhodnikova2024fully} 
study edge-DP algorithms in the fully dynamic continual release model.
They give a number of upper and lower bounds for a variety of 
problems, including triangle counting,
connected component counting,
maximum matching size, and degree histogram. They consider \emph{event-level}
and \emph{item-level} privacy. \emph{Event-level} describes 
fully dynamic neighboring streams 
where one update differs, while \emph{item-level} describes
neighboring streams where all updates on a particular edge differ. 
Note that for insertion-only streams with no duplicate insertions, both 
settings describe the same set of neighboring streams. They give the first {$\poly(T, n)$ error}
lower bounds for $(\eps, \delta)$-DP graph algorithms in continual release
for $\delta > 0$;
previous works only considered $\delta = 0$ in their lower bounds.
Their item-level upper bounds
are tight against their lower bounds, with additive error proportional
to $\poly(T, n)$.

Our paper considers different problems (densest subgraph and $k$-core decomposition)
and focuses on designing algorithms with $\poly(\log n)$ additive errors
in insertion-only streams. \cite{raskhodnikova2024fully} focuses
on fully dynamic streams which in general require additive $O(\poly(T, n))$ error 
{(as demonstrated by \cite{raskhodnikova2024fully} and our lower bounds in \Cref{sec:lower-bounds})}.
Furthermore, our paper focuses on the sublinear
space setting and returning vertex subset solutions, two settings not discussed in~\cite{raskhodnikova2024fully}. 

\section{Preliminaries}
We now introduce the notation we use in this paper as well as some definitions. 
Further standard preliminaries are deferred to \Cref{apx:prelims}.

\subsection{Setting \& Notation}
A graph $G=(V,E)$ consists of a set of vertices $V$ and a set of edges $E$ where edge $\{u,v\} \in E$ if and only if there is 
an edge between $u \in V$ and $v \in V$. We write $n:= \vert V \vert$ and $m:= \vert E \vert$. 

We consider the insertion-only setting within the continual release model,
where we begin with an empty graph $G_0 = (V, E_0)$ where $E_0 = \varnothing$
and edge updates in the form of $\{e_t, \perp\}$ arrive in a stream at every timestep $t \in [T]$
{where each non-empty edge is inserted at most once}.\footnote{This is without loss of generality as redundant edge insertions can be treated as an empty update.}
We are required to release an output for the problem of interest after every {(possibly empty)} update. Our goal is to obtain algorithms that achieve sublinear space, either in the number of edges $\tilde{o}(m)$ (note that $T \geq m$ as we allow for empty updates), 
or the number of vertices $\tilde{o}(n)$. 
We assume $T\leq n^c$ for some absolute constant $c\geq 1$ to simplify our presentation. 
If there are no empty edge updates then it is sufficient to consider $c=2$. %

Throughout the paper,
we write $\eta\in (0, 1]$ to denote a fixed multiplicative approximation parameter.
We assume $\eta$ is a fixed constant that is ignored under asymptotic notation.
Constants used throughout the paper may depend on $\eta$.

We use $\mcal M$ to denote a mechanism,
$M(\cdot)$ to denote the size of a maximum matching,
and $\mu$ to denote the mean of a random variable.

\subsection{Continual Release}
We use the phrasing of the below definitions as given in~\cite{jain2024time}.

\begin{definition}[Graph Stream~\cite{jain2024time}]\label{def:graph-stream}
    In the continual release model, a graph stream $S \in \mathcal{S}^T$ of length $T$ is a $T$-element vector
    where the $i$-th element is an edge update $u_i = \{v, w, insert\}$ indicating an edge insertion of edge $\{v, w\}$,
    $u_i = \{v, w, delete\}$ indicating an edge deletion of edge $\{v, w\}$, or $\bot$ (an empty operation). 
\end{definition}

We use $G_t$ and $E_t$ to denote the graph induced by the set of set of updates in the stream $S$ up to and including update $t$.
Now, we define neighboring streams as follows. Intuitively,
 two graph streams are edge neighbors if one can be obtained from the other by removing one edge update 
(replacing the edge update by an empty update in a single timestep). 

\begin{definition}[Edge Neighboring Streams]\label{def:neighboring-streams}
    \sloppy
    Two streams of edge updates, $S = [u_1, \dots, u_T]$ and $S' = [u'_1, \dots, u'_T]$, are \emph{edge-neighboring} if there exists
    exactly one timestamp $t^* \in [T]$ where $u_{t^*} \neq u'_{t^*}$ and for all $t \neq t^* \in [T]$, 
    it holds that $u_t = u'_t$. Streams may contain
    any number of empty updates, i.e.\ $u_t = \bot$. Without loss of generality, we assume for
    the updates $u_{t^*}$ and $u'_{t^*}$ that $u'_{t^*} = \bot$ and $u_{t^*} = \pm e_{t^*}$ is an
    edge insertion or deletion.
\end{definition}

The notion of neighboring streams leads to the following definition of an edge differentially private algorithm.
\begin{definition}[Edge Differential Privacy]\label{def:edge DP}
  Let $\varepsilon, \delta\in (0, 1)$.
  An algorithm $\mcal A(S): \mathcal{S}^T \rightarrow \mathcal{Y}^T$ that takes as input a graph stream $S \in \mathcal{S}^T$
  is said to be \emph{$(\varepsilon, \delta)$-edge differentially private (DP)}
  if for any pair of edge-neighboring graph streams $S, S'$ that differ by 1 edge update
  and for every $T$-sized vector of outcomes $Y\sset \Range(\mcal A)$,
  \[
    \prob\left[ \mcal A(S)\in Y \right]
    \leq e^\varepsilon \cdot \prob\left[ \mcal A(S')\in Y \right] + \delta.
  \]
  When $\delta=0$,
  we say that $\mcal A$ is \emph{$\varepsilon$-edge DP}.
\end{definition}

We now formalize the concepts of edge edit distance between streams and the concept of \emph{stability} under sparsification. 

\begin{definition}[Edge Edit Distance]\label{def:edit-distance}
    Given two streams $S, S' \in \mathcal{S}^T$, the edge edit distance between the two streams is the shortest chain of 
    graph streams $S_0, S_1, \dots, S_d$ where $S_0 = S$ and $S_d = S'$ where every adjacent pair of streams in the chain are edge-neighboring.
    The edge edit distance is $d$, the length of the chain. 

\end{definition}

An algorithm that takes as input a stream $S$ and outputs a chosen set of updates from the stream is \emph{stable} if on edge-neighboring 
streams $S$ and $S'$, there exists a coupling\footnote{{
A \emph{coupling} of random variables $(X, Y)$ is a joint distribution such that the marginal distributions correspond to $X, Y$ respectively.}}
between the randomness used in the algorithm on the inputs such that the edge distance between the output streams 
is $O(1)$. 

\subsection{Differential Privacy Tools} \label{prelim:dp}

Here, we define the privacy tools commonly used in differential privacy in terms of the continual release model.
Throughout the paper,
we use some standard privacy mechanisms as building blocks (see~\cite{dwork2014algorithmic} for a reference).
\begin{definition}[Global sensitivity]
The global sensitivity of a function $f:\cD\to\mathbb{R}^d$ is defined by
\[\Delta_f=\max_{D,D'\in\cD,D \sim D'}\|f(D)-f(D')\|_1.\]
where $D \sim D'$ are neighboring datasets and differ by an element. 
\end{definition}

\begin{definition}[Laplace Distribution]
We say a random variable $X$ is drawn from a Laplace distribution with mean $\mu$ and scale $b>0$ if the probability density function of $X$ at $x$ is $\frac{1}{2b}\exp\left(-\frac{|x-\mu|}{b}\right)$. 
We use the notation $X\sim\Lap(b)$ to denote that $X$ is drawn from the Laplace distribution with scale $b$ and mean $\mu=0$. 
\end{definition}
The Laplace mechanism for $f: X\to \R$ with global sensitivity $\sigma$ 
adds Laplace noise to the output of $f$ with scale $b=\nicefrac\sigma\eps$
before releasing.
\begin{proposition}[\cite{dwork2014algorithmic}]
    The Laplace mechanism is $\varepsilon$-DP.
\end{proposition}

\begin{theorem}[Adaptive Composition; \cite{DMNS06,DL09,DRV10}]\label{thm:composition}
    \sloppy
    A sequence of DP algorithms, $(\alg_1, \dots, \alg_k)$, with privacy parameters $(\eps_1, \dots, \eps_k)$ form at worst an $\left(\eps_1 + \cdots + \eps_k\right)$-DP algorithm under \emph{adaptive composition} (where the adversary can adaptively select algorithms after
    seeing the output of previous algorithms).%
\end{theorem}

\begin{theorem}[Group Privacy; Theorem 2.2 in \cite{dwork2014algorithmic}]\label{thm:group-privacy}
    Given an $\eps$-edge (node) differentially private algorithm, $\alg$, for all pairs of input streams $S$ and $S'$, it holds that for 
    all possible outcomes $Y \in \Range(\alg)$,
    \begin{align*}
        e^{-k\eps} \leq \frac{\prob[\alg(S') \in  Y]}{\prob[\alg(S) \in Y]} \leq e^{k\eps}
    \end{align*}
    where $k$ is the edge (node) edit distance between $S$ and $S'$.
\end{theorem}

\subsubsection{Sparse Vector Technique}
Below, we define the \emph{sparse vector technique} and give its privacy and approximation guarantees. The sparse
vector technique is used to answer \emph{above threshold} queries where an above threshold query checks whether the 
output of a function that operates on an input graph $G$ exceeds a threshold $T$.

We use the variant introduced by \cite{lyu2017understanding}
and used by \cite{FichtenbergerHO21}.
Let $D$ be an arbitrary (graph) dataset,
$(f_t, \tau_t)$ a sequence of (possibly adaptive) query-threshold pairs,
$\Delta$ an upper bound on the maximum sensitivity of all queries $f_t$,
and an upper bound $c$ on the maximum number of queries to be answered ``above''.
Typically, the AboveThreshold algorithm stops running at the first instance of 
the input exceeding the threshold, but we use the variant where the input 
can exceed the threshold at most $c$ times where $c$ is a parameter passed into the function. 

Throughout
 this paper, we use the class $\textsc{SVT}(\eps, \Delta, c)$ (\Cref{alg:sparse vector technique}) where $\eps$ is our privacy parameter, 
 $\Delta$ is an upper bound on the maximum sensitivity of incoming queries,
 and $c$ is the maximum number of ``above'' queries we can make.
 The class provides a $\textsc{ProcessQuery}(query, threshold)$ function
 where $query$ is the query to SVT and $threshold$ is the threshold
 that we wish to check whether the query exceeds.

\begin{theorem}[Theorem 2 in \cite{lyu2017understanding}]\label{thm:sparse vector technique}
    \Cref{alg:sparse vector technique} is $\varepsilon$-DP.
\end{theorem}

We remark that the version of SVT we employ (\Cref{alg:sparse vector technique}) does not require us to resample the noise for the thresholds (\Cref{svt:threshold noise}) after each query
but we do need to resample the noise (\Cref{svt:query noise}) for the queries after each query.

\begin{algorithm}[htp!]
\caption{Sparse Vector Technique}\label{alg:sparse vector technique}
\SetKwFunction{FAlg}{SVT}
\SetKwProg{Fn}{Class}{}{}
Input: privacy budget $\varepsilon$, upper bound on query sensitivity $\Delta$, maximum allowed ``above'' answers $c$ \\
\Fn{\FAlg{$\varepsilon, \Delta, c$}}{
    $\varepsilon_1, \varepsilon_2 \gets \nicefrac\varepsilon2$ \\
    $\rho \gets \Lap(\nicefrac\Delta{\varepsilon_1})$ \label{svt:threshold noise} \\
    $\Count \gets 0$ \\

    \SetKwFunction{FAlg}{ProcessQuery}
    \SetKwProg{Fn}{Function}{}{}
    \Fn{\FAlg{$f_t(D), \tau_t$}}{
        \If{$\Count > c$} {
            \Return ``abort'' \\
        }
        \If{$f_t(D) + \Lap({2c\Delta}/{\varepsilon_2}) \geq \tau_t + \rho$}{ \label{svt:query noise}
          \Return ``above'' \\
          $\Count\gets \Count + 1$ \\
        }
        \Else{
          \Return ``below'' \\
        }
    }
}
\end{algorithm}

\section{\texorpdfstring{$k$}{k}-Core Decomposition}\label{sec:kcore}

In this section, we introduce a semi-streaming sampling algorithm that preserves the $k$-cores in an
input graph $G = (V, E)$ while ensuring privacy. 
Specifically, we have the following theorem.

\begin{restatable}[Sublinear Space Private $k$-Core Decomposition]{theorem}{kCoreFormal}\label{thm:kcore-formal}
    Fix $\eta \in (0, 1]$.
    \Cref{alg:k-core-main} is an $\eps$-DP algorithm for the 
    $k$-core decomposition problem in the continual release model for insertion-only streams.
    At every $t\in [T]$,
    the algorithm returns a value for each vertex, such that every value
    is a $\left(2+ \eta, 
    O\left({\eta^{-2}}\eps^{-1} \log^3(n) \right) \right)$-approximation of the corresponding vertex's
    true core value, with probability $1 - 1/\poly(n)$. The maximum space used is $O\left({\eta^{-4}}\varepsilon^{-1} n\log^5 n \right)$, 
    with probability $1 - 1/\poly(n)$.
\end{restatable}

\myparagraph{Algorithm Intuition} 
First, we give some intuition for our algorithm (\Cref{alg:k-core-main}). Our algorithm essentially performs a sampling version of the 
classic peeling algorithm for $k$-core decomposition. The classic peeling algorithm successively \emph{peels} (removes) 
vertices with the minimum degree until all vertices are removed from the graph. 
A core that is formed
during the peeling process is the induced subgraph consisting of the remaining vertices after
a vertex is peeled
and the value of such a core is the minimum induced degree within the subgraph.
The core number for each vertex $v$
is equal to the maximum valued core that $v$ is a part of during any stage of the peeling. A dynamic version of this algorithm 
can be obtained by maintaining a \emph{level data structure} where a vertex is moved up a level if its 
induced degree among vertices in the same or higher levels is larger than a cutoff $C$. One can show that 
having $O(\log n)$ levels of the structure and appropriately setting $C$ among $O(\log n)$ duplicates of the 
structure gives a $(2+\eta)$-approximation of the core numbers of the nodes in the non-private, insertion-only 
setting~\cite{SCS20,LSYDS22,DLRSSY22}.
{Our main innovation is a private and sparse level data structure.}

When sparsifying the graph, we cannot simply take a uniform sample of the edges adjacent to each vertex. An easy
example to consider is a vertex $v$ which is part of a $10$-clique and also adjacent to $n/2$ degree one vertices. A uniform
sample of the edges adjacent to $v$ will most likely not discover the $9$ edges connecting it to the $10$-clique (for large $n$).
We call the edges connecting $v$ to the $10$-clique the set of \emph{important} edges.
Thus, we must take a smarter sample of edges adjacent to $v$. 
To maintain a sparsified, sampling-based version of the level data structure, we maintain samples of large enough 
size of the \emph{up-edges} adjacent to each vertex. The up-edges adjacent to each vertex are the 
edges connecting each vertex to neighbors in the same or higher levels. Once we see enough sampled up-edges,
we move the vertex up one level and continue sampling edges until we either reach the topmost level or 
the vertex is adjacent to only a very small sample of up-edges. Such a sampling method allows us to keep enough of the 
important edges which connect to other vertices in higher valued cores. 

Finally, to make the above algorithm differentially private, we use SVT to determine when to move the vertex up a level.
We show that although many vertices may move up levels, we only lose privacy when the vertices that are adjacent to the 
edge that differs between neighboring streams move up. 
Since our total number of levels and duplicates is bounded by $O(\log^2 n)$, the
privacy loss from SVT is also bounded by $O(\log^2 n)$.

\myparagraph{Analysis}
We provide the pseudocode and proof of \Cref{thm:kcore-formal} in \Cref{apx:kcore-formal}.
Specifically,
the pseudocode of the sparse level data structure is provided in \Cref{alg:k-core}
and the pseudocode of the actual algorithm is in \Cref{alg:k-core-main}.
\Cref{apx:k-core-description} then presents a detailed description of the algorithm
and we prove its privacy and utility guarantees in \Cref{apx:k-core-privacy,apx:k-core-utility},
respectively.

\section{Densest Subgraph}\label{sec:densest-subgraph}
In this section, we focus on the densest subgraph problem and provide the first differentially private 
algorithm for densest subgraph in the continual release model using space sublinear in the total 
number of edges in the graph. %

\begin{restatable}[Sublinear Space Private Densest Subgraph]{theorem}{densestSubgraph}\label{thm:densest-subgraph} 
    Fix $\eta \in (0, 1]$.
    \Cref{alg:insertion-main} is an $\eps$-edge differentially private algorithm for the densest subgraph problem in the continual release model for insertion-only streams. The algorithm returns a set of vertices whose induced subgraph is a $\left(2+\eta, 
    O({\eta^{-1}}\eps^{-1}\log^2(n)) \right)$-approximation of the densest subgraph in $G_t$, 
    with probability at least $1 - 1/\poly(n)$, 
    for all $t \in [T]$. 
    The maximum space used is $O\left({\eta^{-2}}\eps^{-1} n\log^2(n)\right)$, with probability at 
    least $1 - 1/\poly(n)$, 
    for all $t \in [T]$.
\end{restatable}

We can also reduce the multiplicative error to $(1+\eta)$
at the cost of increasing the space usage by a $\poly(\log(n))$ factor.
\begin{restatable}[Sublinear Space Private Densest Subgraph]{theorem}{densestSubgraphBetterApproxMoreSpace}\label{thm:densest-subgraph better approx more space} 
    Fix $\eta \in (0, 1]$.
    There exists an $\eps$-edge differentially private algorithm for the densest subgraph problem in the continual release model for insertion-only streams. The algorithm returns a set of vertices whose induced subgraph is a $\left(1+\eta, 
    {\eta^{-4}}\eps^{-1} \log^5(n) \right)$-approximation of the densest subgraph in $G_t$, 
    with probability at least $1 - 1/\poly(n)$, 
    for all $t \in [T]$. 
    The maximum space used is $O\left({\eta^{-5}}\eps^{-1} n\log^5(n) \right)$, with probability at 
    least $1 - 1/\poly(n)$, for all $t \in [T]$.
\end{restatable}

\myparagraph{Algorithm Intuition} 
We revise the algorithms of~\cite{mcgregor2015densest} and~\cite{EHW16} to the insertion-only continual release setting. On a high-level, our algorithm maintains a sample of edges over time and releases a DP set of vertices by running a black-box DP densest subgraph algorithm (e.g., \Cref{thm:private static densest subgraph}) on the subgraph induced by the sample. At every timestep $t\in [T]$, an edge $e_t$ is sampled with probability $p$ --- the sampling probability is initialized to 1 as in the beginning we can afford to store every edge. We privately check whether the number of edges seen so far exceeds a certain threshold using a sparse vector technique (SVT) query and adjust the sampling probability $p$ accordingly. In order to avoid privacy loss that grows linearly in $T$, we do not invoke the black-box DP densest subgraph algorithm at every timestep and instead invoke it only if the current density of the sample exceeds a certain threshold using another SVT query.  

\myparagraph{Analysis}
The pseudocode and proof of \Cref{thm:densest-subgraph,thm:densest-subgraph better approx more space} are deferred to \Cref{apx:densest-subgraph}.
In particular,
pseudocode for a useful data structure is provided in \Cref{alg:insertion-only dsg}
and the main algorithm is in \Cref{alg:insertion-main}.
\Cref{apx:dsg-description} describes the algorithm in detail
and we prove its privacy and utility guarantees in \Cref{apx:dsg-privacy,apx:dsg-utility},
respectively.

\section{Maximum Matching }\label{sec:matching}\label{sec:edge-matching}
In this section we give an edge-DP algorithm for 
maximum matching that uses truly sublinear space. 
We adapt the algorithm of~\cite{mcgregor2018simple} to obtain a private approximate maximum cardinality algorithm
in the continual release model using sublinear space in the number of vertices. 
As in their algorithm, we assume that we are provided with a public upper bound $\tilde \alpha$ on the maximum arboricity $\alpha$ of the graph 
at any point in the stream. 
We also briefly sketch how to remove the assumption on $\tilde \alpha$ in \Cref{apx:guess arboricity} at the cost of more space
and a worse approximation guarantee.

The privacy of our algorithm is always guaranteed.
However, we do not assume that $\tilde{\alpha}$ is guaranteed to upper bound $\alpha$.
When $\tilde\alpha\geq \alpha$,
our approximation guarantees hold with high probability. 
Otherwise,
our approximation guarantees do not necessarily hold. 
Note that the same type of guarantee holds for the original non-private streaming algorithm~\cite{mcgregor2018simple} where 
their utility guarantee is only given when their public estimate of $\alpha$ upper bounds the maximum arboricity of the input.
We prove the following result in this section.

\begin{restatable}{theorem}{edgeDPMatching}\label{thm:edge-dp-matching}
    Fix $\eta \in (0, 1]$.
    Given a public estimate $\tilde\alpha$ of the maximum arboricity $\alpha$ over the stream,\footnote{We can eliminate this assumption with an additional pass of the stream or notify the observer when the utility guarantees no longer hold in the one-pass setting. See \Cref{sec:arboricity assumption}.}
    \Cref{alg:matching-sampling} is an $\eps$-edge DP 
    algorithm for estimating the size of the maximum matching in the continual release model
    for insertion-only streams.
    If $\tilde \alpha\geq \alpha$,
    then with probability at least $1-1/\poly(n)$, 
    our algorithm returns a $\left((1+\eta)(2+\tilde{\alpha}), O( \eta^{-1} \eps^{-1} \log^2(n) )\right)$-approximation 
    of the size of the maximum
    matching at every timestamp.
    Moreover,
    our algorithm uses $O( \eta^{-2} \eps^{-1} {\log^2(n)\log(\tilde\alpha)} )$ space with probability $1-1/\poly(n)$.
\end{restatable}

\myparagraph{Algorithm Intuition}
We revise the algorithm of \cite{mcgregor2018simple} to the insertion-only continual release setting.
On a high-level,
\cite{mcgregor2018simple} showed that the cardinality of a carefully chosen subset of edges $F\sset E$ is a good estimator for the size of the maximum matching for graphs of bounded arboricity $\alpha$.
$F$ is obtained from $E$ by deleting edges $e$ adjacent to a vertex $v$
if more than $\alpha$ other edges adjacent to $v$ arrived after $e$.
Then their algorithm maintains a small sample $S\sset F$
throughout the algorithm
by down-sampling edges when the current sample exceeds some threshold.
Our algorithm follows a similar approach with two main adjustments to satisfy privacy.
First,
we release powers of $(1+\eta)$ based on an SVT comparison against the current value of the estimator.
Second,
the decision to down-sample is also based on an SVT comparison against the threshold.

\myparagraph{Analysis} The proof and pseudocode for \Cref{thm:edge-dp-matching} is deferred to \Cref{apx:matching}.
The pseudocode for handling edge updates is presented in \Cref{alg:matching-sampling}.
Then,
\Cref{apx:matching-description} provides a detailed description of the algorithm.
We prove its privacy and utility guarantees in \Cref{apx:matching-privacy,apx:matching-utility},
respectively.

\section{Lower Bounds for Fully Dynamic Streams}\label{sec:lower-bounds}

In this section,
we establish lower bounds on the additive error of differentially private algorithms
for estimating the size of a maximum matching
and the number of connected components
in the continual release model.
Similar to the lower bound for counting distinct elements in the continual release model \cite{jain2023counting},
we reduce the problem of answering matching queries and connected component queries
to answering \emph{inner product queries}.
Then,
we leverage known lower bounds for the inner product problem \cite{Dinur2003Revealing,Dwork2007Price,MMNW11,De2012Lower} to obtain our lower bounds.

\begin{restatable}{theorem}{matchingLowerBoundSmallEpsilon}\label{thm:matching lower bound small epsilon}
    Fix $\varepsilon\in (0, 1)$.
    If $\mcal A$ is an $\varepsilon$-DP mechanism that answers maximum matching queries on graphs with $n$ vertices in the continual release model within additive error $\zeta$
    with probability at least $0.99$
    for fully dynamic streams of length $T$,
    then
    \[
        \zeta
        = \Omega\left( \min\left( \sqrt{\frac{n}\varepsilon}, \frac{T^{1/4}}{\varepsilon^{3/4}}, n, T \right) \right).
    \]
    Moreover,
    we may assume the graph is bipartite,
    has maximum degree 2,
    and arboricity 1.
\end{restatable}

\begin{restatable}{theorem}{connectedComponentsLowerBoundSmallEpsilon}\label{thm:connected components lower bound small epsilon}
    Fix $\varepsilon\in (0, 1)$.
    If $\mcal A$ is an $\varepsilon$-DP mechanism that answers connected component queries on graphs with $n$ vertices in the continual release model within additive error $\zeta$
    with probability at least $0.99$
    for fully dynamic streams of length $T$,
    then
    \[
        \zeta
        = \Omega\left( \min\left( \sqrt{\frac{n}\varepsilon}, \frac{T^{1/4}}{\varepsilon^{3/4}}, n, T \right) \right).
    \]
    Moreover,
    we may assume the graph is bipartite,
    has maximum degree 2,
    and arboricity 2.
\end{restatable}

The proofs of \Cref{thm:matching lower bound small epsilon,thm:connected components lower bound small epsilon} are presented in \Cref{apx:lower-bounds}.
In \Cref{apx:inner product queries},
we review the key problem which we reduce to matching and connected components.
The lower bound for maximum matching is proven in \Cref{apx:matching lower bound}
and similarly for connected components in \Cref{apx:connected components lower bound}.

\subsection{Further Graph Statistics}\label{sec:further lower bounds}
The underlying idea for the maximum matching and connected components lower bounds 
(\Cref{thm:matching lower bound small epsilon}, \Cref{thm:connected components lower bound small epsilon})
is that we can encode the bits of a secret database $y\in \set{0, 1}^n$ within the structure of a sparse graph.
Privately answering an inner product query on this database
then reduces to answering a ``bitwise OR'' query
by the inclusion-exclusion principle.

It is not hard to see that this technique extends to $k$-edge-connected component queries,
$k$-vertex-connected component queries,
and triangle counting queries for sparse graphs.

\section{Conclusion \& Future Work}
In this paper, we initiated the study of low-space continual release algorithms for general graph problems. Using techniques from the non-private graph sparsification literature, we provided continual release algorithms for a variety of general graph problems
that for the first time, 
achieve nearly the same space and approximation guarantees of their non-private streaming counterparts. The improved space bounds are especially relevant for enabling computations on massive datasets, which are the core motivation of the field of online streaming algorithms. 

For our upper bounds, we focused on the insertion-only setting of continual release. As the area of fully-dynamic algorithms in the continual release model is largely unexplored, we believe that an interesting future research direction is closing the gap in our theoretical understanding of this model. As observed in prior work, and hinted by our hardness results, the fully dynamic setting is significantly harder in continual release, with even basic graph problems requiring $\Omega(\poly(n))$ additive error for dynamic streams while admitting $\tilde O(\poly(\log(n))/\eps)$ additive error in the insertion-only case. In this context, it would be especially interesting to deepen our understanding of the interplay between the dynamicity of the continual release setting (insertion-only vs fully-dynamic) and the \emph{space} lower bounds (as opposed to error lower bounds) imposed by privacy. This is an area that has only recently received attention~\cite{dinur2023differential} and is an intriguing future direction to explore.

\section*{Acknowledgments}
Felix Zhou acknowledges the support of the Natural Sciences and Engineering Research Council of Canada (NSERC).

\begingroup
\sloppy
\printbibliography
\endgroup

\clearpage
\addtocontents{toc}{\protect\setcounter{tocdepth}{1}}
\appendix

\section{Deferred Preliminaries}\label{apx:prelims}
We now remind the reader of some standard definitions and tools.

\subsection{Approximation Algorithms}
We consider optimization problems in both the minimization and maximization setting. Let $\OPT \in \R$ be the optimum value for the problem. For a minimization problem,
a \emph{$(\beta, \zeta)$-approximate algorithm} outputs a solution with cost at most $\beta\cdot \OPT + \zeta$.
For a maximization problem,
$(\beta, \zeta)$-approximate algorithm outputs a solution of value at least $\frac1\beta\cdot \OPT - \zeta$.
For estimation problems, a $(\beta, \zeta)$-approximate algorithm outputs an estimate $\tilde R\in \R$ of some quantity $R\in \R$
where $R - \zeta \leq \tilde R \leq \beta \cdot R + \zeta$ (one-sided multiplicative error). 
As a shorthand,
we write $\beta$-approximation algorithm to indicate a $(\beta, 0)$-approximation algorithm. Our approximation bounds hold, with high probability, 
{over all} releases.

\subsection{Concentration Inequalities}
Below, {we state the version of the multiplicative Chernoff bound we use throughout this work.}

\begin{theorem}[Multiplicative Chernoff Bound; Theorems 4.4, 4.5 in \cite{mitzenmacher2015probability}]\label{thm:multiplicative-chernoff}
    Let $X = \sum_{i = 1}^n X_i$ where each $X_i$ is a Bernoulli variable which takes value $1$ with probability $p_i$ and value $0$
    with probability $1-p_i$. Let $\mu = \expect[X] = \sum_{i = 1}^n p_i$. Then, it holds:
    \begin{enumerate}
        \item Upper Tail: $\prob[X \geq (1+\psi) \cdot \mu] \leq \exp\left(-\frac{\psi^2\mu}{2 + \psi}\right)$ for all $\psi > 0$;
        \item Lower Tail: $\prob[X \leq (1-\psi) \cdot \mu] \leq \exp\left(-\frac{\psi^2\mu}{3}\right)$ for all $0 < \psi < 1$.
    \end{enumerate}
\end{theorem}

\section{Proof of \texorpdfstring{\Cref{thm:kcore-formal}}{k-Core Theorem}}\label{apx:kcore-formal}
In this section,
we restate and prove \Cref{thm:kcore-formal}.
\kCoreFormal*

Our algorithm must solve several crucial and non-trivial challenges,
as mentioned in the Technical Overview (\Cref{sec:tech-overview}), 
making our algorithm conceptually more complicated than any non-private streaming $k$-core algorithm.
{We give a short reminder below of our previously detailed challenges for $k$-core decomposition as it is our first result. 
Densest subgraph and matchings share similar issues which are detailed in \Cref{sec:tech-overview}.}

\begin{itemize}
    \item We cannot use a differentially private $k$-core decomposition algorithm as a black-box because all vertices may change their core numbers at some point in the stream, resulting in $\Omega(n)$ changes from the black-box private $k$-core decomposition algorithm. Such a black-box usage of the private algorithm is difficult to analyze without losing a factor of $\eps \cdot n$ in the pure DP setting resulting from composition. 
    \item We cannot produce \emph{one} uniform sample of edges from the graph since this results in an uneven distribution of edges among the cores. Suppose we sample each edge with probability $p$, if $p$ is set to be too small, then 
    vertices with smaller core numbers will not have enough adjacent edges to produce a good concentration bound.
    \item We must also be able to deal with vertices which have large degree and very small cores (consider a star graph). Sampling edges
    uniformly from the original graph, without maintaining some form of a data structure on the sampled edges, will not allow us to 
    distinguish between vertices with large degree and large core numbers from vertices with large degree but small core numbers.
    \item We require a new composition theorem that does not lose privacy for \emph{each} release of a new core number among the $n$ vertices. Intuitively, we should not lose privacy for each release because for edge-neighboring insertion-only streams, only 
    one edge differs between the two streams.
    This theorem uses the recently introduced multidimensional sparse vector technique~\cite{DLL23}.
\end{itemize}

We give the pseudocode for our data structure and algorithm in \Cref{alg:k-core} and \Cref{alg:k-core-main}, respectively. 
Below, our algorithm is inspired by the insertion-only sketching algorithm
of~\cite{Esfandiari2018} but adapted to the level data structure~\cite{DLRSSY22}; our level data structure sparsification
is simpler in nature and uses an entirely new analysis.
Furthermore, since 
the degeneracy of the graph is equal to the maximum core number of any node in the input graph, our algorithm 
also gives a private approximation of the degeneracy of the input graph in the continual release model.

\SetKwProg{Fn}{Function}{}{end}\SetKwFunction{FRecurs}{FnRecursive}%
\SetKwFunction{FnInsertKCore}{PrivateCore}
\SetKwFunction{FnUpdate}{SampleEdge}
\SetKwFunction{FnUpdateLevel}{UpdateLevels}
\SetKwFunction{FnGetPrivateLevels}{GetPrivateLevels}
\SetKwProg{myclass}{Class}{}{}
\begin{algorithm}[htp]
    \caption{Data Structure for $k$-Core Decomposition for Adaptive Insertion-Only Streams\label{alg:k-core}}
    \myclass{\FnInsertKCore{$\eps, \eta, L, n, T$}}{
        Maintain sampled edges $X_j \leftarrow \varnothing$ for each $j \in Q$
        where $Q = \{\ceil{\log_{(1+\eta)}(L)}, \dots, \ceil{2\log_{(1+\eta)}(n)}\}$ \label{kcore:sampled}\\
        Initialize $F \leftarrow \ceil{2\log_{(1+\eta)}(n)}$\label{kcore:top-level}\Comment{$F-1$ denotes the topmost level of any level data structure}\\
        
        Initialize $\eps_1 \leftarrow \eps/(6|Q|F)$ \label{kcore:eps_2}\\ %
        \For{$j \in Q$}{
            Initialize $p_j \leftarrow \frac{c_1 \log(n)}{{\eps_1} (1+\eta)^j}$ \label{kcore:set-prob}\\
            \For{$v \in V$}{
                Initialize class $\textsc{SVT}^{j, v}(\eps_1, 1, {1})$ \qquad (\Cref{alg:sparse vector technique}) \label{kcore:svt-class}\\
                Initialize $levels[j][v] \leftarrow 0$ \label{kcore:initial-level}\\
            }
        }
        \Fn{\FnUpdate{$e_t$}}{
                \For{$j \in Q$}{\label{kcore:j-in-q}
                    \For{$w \in e_t = \{u, v\}$}{\label{core:check-endpoint}
                        \If{$levels[j][w] < F - 1$}{\label{core:not-last-level}
                            Sample $e_t$ into $X_j$ with probability $p_j$. \label{core:sample-edge-x_j}\\
                        }
                    }
                }
        }

        \Fn{\FnUpdateLevel{}}{
            \For{$j \in Q$}{\label{core:update-j}
                \For{level $\ell \in \{0, \dots, F-2\}$}{\label{core:iterate-level}
                    \For{$v \in V$}{\label{core:update-vertex}
                        \If{$levels[j][v] \neq \ell$}{ \label{core:current-level}
                            go to next iteration \\
                        }
                        \If{$\textsc{SVT}^{j, v}\textsc{.ProcessQuery}(\deg^+_{X_j}(v), p_j \cdot (1+\eta)^{j-1})$ is ``above''}{\label{core:update-check}
                            $levels[j][v] \leftarrow levels[j][v] + 1$\label{core:move-up-level}\\
                        }
                    }
                }
            }
        }
        \Fn{\FnGetPrivateLevels{}}{ 
            \Return $levels$\label{kcore:return-levels}
        }
	}

    \storelines
\end{algorithm}

{%
\begin{algorithm}[htp]
    \caption{{Algorithm for $k$-Core Decomposition for Adaptive Insertion-Only Streams}\label{alg:k-core-main}}
    \restorelines
    Initialize initial core threshold
    $L \leftarrow \frac{c_3 \log^3(n)}{\eps}$ \label{kcore:initial-threshold}\\
    Initialize \textsc{PrivateCore}($\eps, \eta, L$)\label{kcore:initialize-dsg}\\
    $levels \leftarrow \textsc{PrivateCore.GetPrivateLevels}()$\\
    \For{each new update $e_t$}{\label{core:for}
        \If{$e_t \neq \bot$}{
            \textsc{PrivateCore.SampleEdge($e_t$)}\label{kcore:sample}\\ 
        }
        
        $\textsc{PrivateCore.UpdateLevels}()$\label{kcore:update-levels}\\
        $levels \leftarrow \textsc{PrivateCore.GetPrivateLevels}()$\label{kcore:new-levels}\\
        \For{every vertex $v \in V$}{\label{core:vertex-iterate}
            $j_{now} \gets \max\set{j: levels[j][v] = F-1}$\\
            \If{$(1+\eta)^{j_{now}} > L$}{\label{core:j_now-greater}
                \textbf{Release} $v$'s core as $(2+\eta) \cdot (1+\eta)^{j_{now}}$\label{core:release-new-estimate}\\
            }
            \Else{
                \textbf{Release} $v$'s core as 1 \\
            }
        }
    }
\end{algorithm}}

\subsection{Detailed Algorithm Description}\label{apx:k-core-description}
We now give the detailed description of our algorithm.
Our algorithm proceeds as follows. We maintain $O\left(\log_{(1+\eta)}(n)\right)$ subgraphs where for each subgraph, the 
probability we use to sample edges into the subgraph is different. Let $Q =  [\ceil{\log_{1+\eta}(L)}, \dots, \ceil{2\log_{(1+\eta)}(n)}]$. 
Specifically, we consider all subgraphs $j \in Q$ (\cref{kcore:sampled}) where the
integer $j \geq \log_{(1+\eta)}(L)$ and we set $L = \frac{c_3 \log^3(n)}{\eps}$ (\cref{kcore:initial-threshold}).
For each update we receive in the stream (\cref{core:for}), we first determine which subgraphs to sample the edge into using
the procedure \textsc{PrivateCore.SampleEdge($e_t$)} (\cref{kcore:sample}). The procedure
determines whether to sample an edge by 
iterating through all of the subgraphs $X_j$ for $j \in Q$ (\cref{kcore:j-in-q}). We decide to sample $e_t = \{u, v\}$ into $X_j$
by looking at both endpoints of the edge update (\cref{core:check-endpoint}). If 
for either endpoint $w \in \{u, v\}$, vertex $w$ is not on level $F-1$ of $X_j$ (\cref{core:not-last-level}), then we 
sample the edge using probability $p_j$ (\cref{core:sample-edge-x_j}). 

Each $X_j$ is organized into levels 
where vertices are moved up levels if they have induced degree among vertices at the same or higher level 
\emph{approximately} greater than $(p_j) \cdot (1+\eta)^{j-1}$. We require the degree to be approximately 
greater rather than exactly greater to preserve privacy. %

After sampling the new edge $e_t$, 
we then perform $\textsc{PrivateCore.UpdateLevels}()$ (\cref{kcore:update-levels})
which updates the levels of each vertex. Within the procedure, for each sampled graph $X_j$ (\cref{core:update-j}) 
and for each level $\ell$ starting from the bottom most level and iterating to the top level (\cref{core:iterate-level}),
we check each vertex $v \in V$ (\cref{core:update-vertex}) to determine whether we need to move that vertex up a level.

Let $\deg^+_{X_j}(v)$ be the degree of $v$ in the induced subgraph of its neighbors at the same level or higher. 
To see whether we should move that vertex up a level, we first determine whether the current level of that 
vertex is the same level that we are iterating (\cref{core:current-level}). If it is the same level and 
we pass the SVT check that $\deg^+_{X_j}(v)$ exceeds
$p_j \cdot (1+\eta)^j$ (\cref{core:update-check}), then we move the vertex up a level by incrementing $levels[j][v]$
(\cref{core:move-up-level}). %
The particular threshold of $p_j \cdot (1+\eta)^j$ that we use will become apparent once we discuss our analysis of the approximation 
factor. Intuitively, the levels mimic the traditional peeling algorithm for the static $k$-core decomposition problem. 
The threshold $(1+\eta)^j$ is used in previous works (e.g.\ \cite{DLRSSY22}) and because we are sampling edges instead of 
using the entire graph, the threshold is multiplied by $p_j$. 

Finally, in the last part of the algorithm, we iterate through each vertex (\cref{core:vertex-iterate})
and release an estimate if the vertex $v$ is in the topmost level of a subgraph $X_{j_{now}}$ 
where $j_{now}$ is maximized. 
If this is the case, we release the new estimate $(2+\eta)\cdot (1+\eta)^{j_{now}}$ for vertex
$v$ (\cref{core:release-new-estimate}),
assuming $(1+\eta)^{j_{now}}$ exceeds some data-oblivious lower bound $L$.

\subsection{Privacy Guarantee}\label{apx:k-core-privacy}

We now prove the privacy guarantees of our algorithm. Specifically, we show that our algorithm maintains $\eps$-differential privacy
on edge-neighboring insertion-only streams. 
Note that although we are using intuition from the multidimensional sparse vector technique (MAT)
given in~\cite{DLL23}, our proof is different in that
{the subsampling in our algorithm requires us to analyze the coupled sensitivity of queries, with the coupling depending adaptively on the current output. MAT~\cite{DLL23} handles releasing values for $n$ vertices without a factor-$n$ composition loss, but was not proven to handle adaptive coupled sensitivity.
}

This is a subtle difference since the original mechanism works in the static setting.
Hence, we call our version the \emph{adaptive} multidimensional sparse vector technique
and provide the full proof below.

\begin{lemma}\label{lem:dp-k-core}
    \cref{alg:k-core-main} is $\eps$-differentially private on edge-neighboring insertion-only streams as defined in~\cref{def:neighboring-streams}.
\end{lemma}

\begin{proof}
    We first observe that the only  information that depends on private data is the level of each node in each $X_j$ for $j \in Q$.
The level of any vertex $v$ depends on the edges that are sampled from the stream for each graph
    as well as their $\deg^+_{X_j}(v)$ values. This proof is a variant of the multidimensional 
    AboveThreshold (MAT) technique used in~\cite{DLL23}
    but we state a version of the proof that directly proves the privacy of the algorithm in our setting; we follow the proof style 
    of~\cite{lyu2017understanding} below. Let $e_{t^*} = \{x, y\}$ be the fixed edge that differs between the two neighboring streams $G$ and $G'$.
    
    SVT is called in~\cref{core:update-check} of~\cref{alg:k-core} to determine whether a vertex moves up a level.
    In fact, vertices can change levels if and only if this SVT passes (and when the current level of the vertex is $\ell$).
    We make the observation that for each vertex $v$ which is not incident to 
    $e_{t^*}$, the distributions of the outputs of the SVT queries are the same in $G_t$ and $G'_t$ when conditioned on the levels
    of the vertices in each $X_j$.
    Let $E_1$ be the event where 
    {we fix a state of the algorithm} 
    prior to the SVT call in~\cref{core:update-check}
    at time $t$.
    {That is,} all vertices are in the levels {specified} in $E_1$
    {and} $E_1$ {also fixes} a sample of edges, $e_t$'s, which may or may not contain $e_{t^*}$.
    We now show the probability of the next $A$ failed queries 
    to~\cref{core:update-check} by vertex $w$ in graph $X_j$; a failed query is one that
    does not return ``above''. We denote a failed SVT query (returns ``below'') by $fail$ and a successful SVT query (returns ``above'') by $success$.
    We denote the output of~\cref{core:update-check} by $\mathcal{S}_{j, v}(\mathcal{X})$ where $\mathcal{X}$ is the state of our 
    subgraphs in $G$ and $\mathcal{X}'$ is the state of our subgraphs in $G'$. First, for every $w \not\in \{x, y\}$, 
    it holds that $\prob[\mathcal{S}_{j, w}(\mathcal{X}) = \{fail\}^A \mid E_1] =
    \prob[\mathcal{S}_{j, w}(\mathcal{X}') = \{fail\}^A \mid E_1]$ since $\deg_{X_j}^+(w) = \deg_{X'_j}^+(w)$ when conditioned on $E_1$.
    Then, we show that for every $w \in \{x, y\}$,
    \begin{align}
        \prob[\mathcal{S}_{j, w}(\mathcal{X}) = \{fail\}^A \mid E_1] &\leq e^{\varepsilon_1} \cdot \prob[\mathcal{S}_{j, w}(\mathcal{X}') = \{fail\}^A \mid E_1].\label{core:privacy-sampling}
    \end{align}
    We show~\cref{core:privacy-sampling} as follows where $f$ is the probability density function for picking a threshold noise and 
    $g_{\mathcal{\mathcal{X}}}$ is the probability density function for failing~\cref{kcore:update-levels} using $\mathcal{X}$.
    Let $\nu^i_{j, w}$ be the individual noises that are picked each time a query is made to~\cref{kcore:update-levels}.
    For simplicity of expression, we do not write the conditioning on $E_1$ on the RHS but all expressions below are conditioned on $E_1$.
    \begin{align}
        \prob[\mathcal{S}_{j, w}(\mathcal{X}) = \{fail\}^A \mid E_1] &= \int_{-\infty}^{\infty} f(z) g_{\mathcal{X}}(z)dz\\
        &=\int_{-\infty}^{\infty}f(z) \cdot \prod_{i \in [A]} \prob\left[\deg^+_{i, X_j}(w) + \nu^i_{j, w} < \frac{c_1\log^3(n)}{\eps} + z\right] dz\\
        &=\int_{-\infty}^{\infty}f(z) \cdot \prod_{i \in [A]} \prob\left[\nu^i_{w, j} < \frac{c_1\log^3(n)}{\eps} -\deg^+_{i, X_j}(w) + z\right] dz \\
        &\leq \int_{-\infty}^{\infty}f(z) \cdot \prod_{i \in [A]} \prob\left[\nu^i_{w, j} < \frac{c_1\log^3(n)}{\eps} -\deg^+_{i, X'_j}(w) + z + 1\right] dz \\
        &= \int_{-\infty}^{\infty}f(z) \cdot g_{\mathcal{X'}}(z + 1) dz\\
        &\leq \int_{-\infty}^{\infty}\exp(\eps_1) \cdot f(z+1) \cdot g_{\mathcal{X'}}(z + 1) dz\\
        &= \int_{-\infty}^{\infty}\exp(\eps_1) \cdot f(z') \cdot g_{\mathcal{X'}}(z') dz' \;\;\; \text{let } z' = z + 1\\
        &=e^{\varepsilon_1} \cdot \prob[\mathcal{S}_{j, w}(\mathcal{X}') = \{fail\}^A \mid E_1].\label{eq:failures}
    \end{align}

    We now prove the privacy characteristics of the cases where the queries succeed. 
    {Similar to $E_1$,
    we write $E_2$ to denote an event where we fix some sampled edges and levels of
    vertices,
    with possibly different edge samples and vertex levels.}
    As in the case above, for all 
    $w \not\in \{x, y\}$, it holds that $\prob[\mathcal{S}_{j, w}(\mathcal{X}) = success \mid E_2] =
    \prob[\mathcal{S}_{j, w}(\mathcal{X}') = success \mid E_2]$ conditioned on $E_2$. Now, we consider the case where $w \in \{x, y\}$. As before, we show that 
    \begin{align}
        \prob[\mathcal{S}_{j, w}(\mathcal{X}) = success \mid E_2] &\leq e^{\varepsilon_1} \cdot \prob[\mathcal{S}_{j, w}(\mathcal{X}') = success \mid E_2].\label{core:privacy-sampling-2}
    \end{align}
    Using $f$ as the probability density function for the noise picked for the threshold,
    $h_{\mathcal{\mathcal{X}}}$ is the probability density function for success, and conditioning on $E_2$ (for simplicity 
    of expression, we do not write the conditioning on the RHS):
    \begin{align}
        \prob[\mathcal{S}_{j, w}(\mathcal{X}) = success \mid E_2] &= \int_{-\infty}^{\infty} f(z) h_{\mathcal{X}}(z)dz\\
        &=\int_{-\infty}^{\infty}f(z) \cdot \prob\left[\deg^+_{X_j}(w) + \nu_{w, j} \geq \frac{c_1\log^3(n)}{\eps} + z\right] dz\\
        &=\int_{-\infty}^{\infty}f(z) \cdot \prob\left[\nu_{w, j} \geq \frac{c_1\log^3(n)}{\eps} -\deg^+_{X_j}(w) + z\right] dz \\
        &\leq \int_{-\infty}^{\infty}f(z) \cdot \exp(\eps_1)\cdot \prob\left[\nu_{w, j} \geq \frac{c_1\log^3(n)}{\eps} -\deg^+_{X'_j}(w) + z + 1\right] dz \\
        &= \exp(\eps_1) \cdot \int_{-\infty}^{\infty}f(z) \cdot h_{\mathcal{X'}}(z + 1) dz\\
        &\leq \exp(2\eps_1) \int_{-\infty}^{\infty} f(z+1) \cdot h_{\mathcal{X'}}(z + 1) dz\\
        &= \exp(2\eps_1) \int_{-\infty}^{\infty} f(z') \cdot h_{\mathcal{X'}}(z') dz' \;\;\; \text{let } z' = z + 1\\
        &=e^{2\varepsilon_1} \cdot \prob[\mathcal{S}_{j, w}(\mathcal{X}') = success \mid E_2].\label{eq:successes}
    \end{align}

   Using the above proofs, the threshold in~\cref{core:update-check} is exceeded at most $|Q| \cdot F$ total times and
    so~\cref{core:update-check}
    is satisfied at most $|Q| \cdot F$ times. Because there is a deterministic mapping between the SVT answers and the levels, i.e., the levels only change when the SVT accepts, it is sufficient to show that the levels can be published differentially
    privately as the rest of the computation in \cref{alg:k-core-main} can be achieved through postprocessing.  
    
    Below, our algorithm is denoted $\mathcal{M}$ and $\mathcal{M}$ outputs a set of levels, one for each vertex, subgraph, and timestamp, $\ell_{j, v}^t$ for 
    each $v \in V, t \in \{1, \dots, T\}$, and $j \in \{1, \dots, \ceil{2\log_{(1+\eta)}(n)}\}$,
    which directly determines the approximate core number for each vertex. Each $\ell_{j, v}^t$ is determine
    solely by the success and failure of the threshold queries. Let $\mathcal{X}_{t}$ be the state of our subgraphs at time $t$ (where $\mathcal{X}_{0}$
    is the initial state where all vertices are on level $0$ for each subgraph)
    and we fix the set of outputs of $\mathcal{S}_{j, w}(\mathcal{X}_t)$ by $\textbf{b}^t = \left[b_{1, 1}^t, \dots, b_{\ceil{2\log_{(1+\eta)}(n)}, n}^t\right]$.
    We use $\mathcal{S}(\mathcal{X}_t) = \textbf{b}^t$ as shorthand for $\mathcal{S}_{j, w}(\mathcal{X}_t) = b^t_{j, w}$ for all $j, w$.
    Then, since $\eps_1 = \eps/(6|Q|F)$, we can use the chain rule 
    and the above expressions to show the following:
    \begin{align}
        &\prob[\mathcal{M}(G) = (\ell_{1, v_1}^1, \dots, \ell_{j, v}^t, \dots, \ell_{\ceil{2\log_{(1+\eta)}(n)}, v_n}^T)]\\ 
        &\leq \prod_{t \in \{1, \dots, T\}}\prob\left[\bigcap_{j, w}\mathcal{S}_{j, w}(\mathcal{X}_{t}) = b_{j, w}^t
        \mid \mathcal{S}(\mathcal{X}_{t-1}) = \textbf{b}^{t-1} \cap \cdots \cap \mathcal{S}(\mathcal{X}_{0}) = \textbf{b}^0\right] \label{eq:kcore1} \\
        &\leq \prod_{t \in \{1, \dots, T\}}\prod_{j, w}\prob\left[\mathcal{S}_{j, w}(\mathcal{X}_{t}) = b_{j, w}^t
        \mid \mathcal{S}(\mathcal{X}_{t-1}) = \textbf{b}^{t-1} \cap \cdots \cap \mathcal{S}(\mathcal{X}_{0}) = \textbf{b}^0\right] \label{eq:kcore2}\\
        &\leq \exp(\eps_1)^{2|Q|F} \cdot \exp(2\eps_1)^{2|Q|F} \cdot
        \prod_{t \in \{1, \dots, T\}}\prod_{j, w} \prob\left[\mathcal{S}_{j, w}(\mathcal{X}'_{t}) = b_{j, w}^t
        \mid \mathcal{S}(\mathcal{X}_{t-1}') = \textbf{b}^{t-1} \cap \cdots \cap \mathcal{S}(\mathcal{X}_{0}') = \textbf{b}^0\right] \label{eq:kcore3}\\
        &\leq \left(\exp(\eps_1)^{|Q| \cdot F}\right)^2 \cdot \left(\exp(\eps_1)^{2|Q| \cdot F}\right)^2 \cdot \prob[\mathcal{M}(G') = (\ell_{1, v_1}^1, \dots, \ell_{j, v}^t, \dots, \ell_{\ceil{2\log_{(1+\eta)}(n)}, v_n}^T)] \label{eq:kcore4}\\
        &\leq \exp(\eps)\cdot \prob[\mathcal{M}(G') 
        = (\ell_{1, v_1}^1, \dots, \ell_{j, v}^t, \dots, \ell_{\ceil{2\log_{(1+\eta)}(n)}, v_n}^T)].\label{eq:kcore5}
    \end{align}
    \cref{eq:kcore1} follows from the chain rule and since the levels are determined by the success and failures of the SVT threshold queries.
    \cref{eq:kcore2} follows since each output of the threshold query at $t$ is independent conditioned on the states of the vertices from $t-1$. 
    \cref{eq:kcore3} follows from~\cref{eq:failures} and~\cref{eq:successes} for vertices $x$ and $y$ and the fact that 
    there are at most $|Q|\cdot F$ successes for each of $x$ and $y$ (where $e_{t^*} = \{x, y\}$)
    and hence also $|Q|\cdot F$ consecutive runs of failures.
    \cref{eq:kcore4} follows since the levels of the vertices in $G'$ are determined by the successes and failures of the queries on $G'$. 
    Finally,~\cref{eq:kcore5} follows because we set $\eps_1 = \eps/(6|Q|F)$.
    This concludes the proof of the privacy of our algorithm since only the levels are used (via post-processing) to obtain the approximations.
\end{proof}

\subsection{Approximation and Space Guarantees}\label{apx:k-core-utility}
Given our privacy guarantees, we now prove our approximation guarantees in this section. We first prove our 
concentration bound on the samples we obtain in our procedure. Specifically, we show that with high probability, 
the sampled edges give an approximately accurate estimate of $\deg^+_{X_j}(v)$. Our proof strategy is as follows. 
We use the proof of the approximation given in \cite[Theorem 4.7]{DLRSSY22}. In order to use their theorem, 
we consider a hypothetical set of level data structures where we keep \emph{all} of the edges in the graph. Note that we do not 
maintain these level data structures in our algorithm but only use them for the sake of analysis. Let this set of 
level data structures be denoted as $\mathcal{K}$. 
We place each vertex $v$ in $\mathcal{K}$ on the exact same level as $v$ in the level data structures within \Cref{alg:k-core}. Then, we show that the vertices in 
$\mathcal{K}$ satisfy modified versions of Invariant 3 (\Cref{lem:kcore-invariant-1}) and Invariant 4 (\Cref{lem:kcore-invariant-2}) in~\cite{DLRSSY22}, which directly gives our approximation factor by a modified version of \cite[Theorem 4.7]{DLRSSY22}.

In the below proofs, let $\deg_{\mathcal{K}}(\ell, j, v)$ be the induced degree of $v$ in {the $j$-th level data structure within} $\mathcal{K}$ consisting
of all neighbors of $v$ that are on level $\ell$ and higher. We prove the following lemmas for graph $G_t$ and $G'_t$
which are formed after the first $t \in [T]$ updates. 

\begin{lemma}\label{lem:kcore-invariant-1}
    If vertex $v$ is in level $\ell < F - 1$, with high probability,
    in subgraph $X_j$ after the levels are updated by~\cref{kcore:update-levels}, then,
    $\deg_{\mathcal{K}}(\ell, j, v) \leq (1+\eta)^j + O\left(\frac{\log^3 n}{{\eta^2}\eps}\right)$, with high probability.
\end{lemma}

\begin{proof}
    We prove this statement via contradiction. 
    Suppose $v$ is in level $\ell < F - 1$ within subgraph $X_j$ with high probability
    and $\deg_{\mathcal{K}}(\ell, j, v)
    > (1+\eta)^j + \frac{a_1\log^3 n}{\eps}$ with probability at least $n^{-a_2}$
    for some fixed constants $a_1, a_2 \geq 1$. We are in graph $X_j$, hence the probability we used to sample is 
    $p_j = \frac{c_1\log^3(n)}{\eps(1+\eta)^j}$. The expected number of edges we sample out of the $\deg_{\mathcal{K}}(\ell, j, v)$ edges
    using $p_j$ is at least 
    \begin{align}
    \mu \geq \frac{c_1\log^3 n}{\eps(1+\eta)^j} \cdot \left((1+\eta)^j + \frac{a_1\log^3 n}{\eps}\right) = \frac{c_1 \log^3 n}{\eps} + 
    \frac{a_1 c_1 \log^6 n}{\eps^2 (1+\eta)^j} > \frac{c_1 \log^3 n}{\eps}.
    \end{align} 
    Let $S_{\ell, j, v}$ be the set of edges we sampled.
    Using a multiplicative Chernoff bound (\Cref{thm:multiplicative-chernoff}), we have that the probability that our sample 
    has size smaller than $\frac{(1-\psi)c_1\log^3 n}{\eps}$ 
    is as follows, for $\psi \in (0, 1)$:
    \begin{align}
        \prob\left[|S_{\ell, j, v}| \leq (1-\psi) \mu \right] \leq \exp\left(-\frac{\mu \psi^2}{3}\right) \leq \exp\left(-\frac{\psi^2 c_1\log^3 n}{3\eps}\right).
    \end{align}

    The SVT introduces at most $\frac{a_3\log^3(n)}{\eps}$ additive error for some constant $a_3$, with high probability,
    and so the value we are comparing against the threshold of $p_j \cdot (1+\eta)^{j-1}$ is at least 
    $|S_{\ell, j, v}| - \frac{2a_3\log^3(n)}{\eps}$,
    with high probability. We proved above that the probability that $|S_{\ell, j, v}| > \frac{(1-\psi)c_1\log^3(n)}{\eps}$ is at least
    $1 - \exp\left(-\frac{\psi^2 c_1\log^3(n)}{3\eps}\right)$. 
    Thus, conditioned on the SVT error bound and sampling bound above,
    the SVT comparison succeeds if
    \begin{align}
        |S_{\ell, j, v}| - \frac{2a_3\log^3(n)}{\eps} 
        &\geq \frac{(1-\psi)c_1\log^3(n)}{\eps} - \frac{2a_3\log^3(n)}{\eps} \\
        &\geq p_j \cdot (1+\eta)^{j-1} \\ 
        &= \frac{c_1\log^3 n}{\eps(1+\eta)}.
    \end{align}
    Note that this inequality always holds when
    \begin{align}
        (1-\psi)c_1 - 2a_3 
        \geq \frac{c_1}{1+\eta}.
    \end{align}
    To satisfy the above expression, we require $\psi < \frac{\eta}{\eta +1}$ and $\eta > 0$, which is easily satisfied by the 
    constraints of our problem, and
    also $c_1 \geq -\frac{(\eta + 1)2a_3}{\eta(\psi-1)+\psi}$. We set $c_1$ 
    to be an appropriately large enough constant in terms of 
    $\eta, a_3, \psi$ to amplify the probability of success to satisfy with high probability. 
    
    For an appropriate setting of $c_1, \psi, a_3$, 
    we have that the SVT succeeds 
    with probability at least $1 - n^{-c}$ for any constant $c \geq 1$. Taking the union bound over
    all levels $\ell' \leq \ell$ (such that SVT outputs succeed for all such $\ell'$),
    this contradicts with the fact that 
    $v$ is at level $\ell$ with high probability. 
\end{proof}

\cref{lem:kcore-invariant-1} directly shows that Invariant 3 is satisfied in~\cite{DLRSSY22}. Now, we prove that Invariant $4$ is also
satisfied. The proof follows a similar structure to the proof of~\cref{lem:kcore-invariant-1}.

\begin{lemma}\label{lem:kcore-invariant-2}
    If vertex $v$ is in level $\ell > 0$, with high probability,
    in subgraph $X_j$ after the levels are updated by~\cref{kcore:update-levels}, then,
    $\deg_{\mathcal{K}}(\ell - 1, j, v) \geq (1+\eta)^{j-2} - O\left(\frac{\log^3 n}{{\eta^2}\eps}\right)$, with high probability.
\end{lemma}

\begin{proof}
    As in the proof of~\cref{lem:kcore-invariant-1}, we prove this lemma by contradiction. Suppose $v$ is in level $\ell > 0$,
    with high probability, in subgraph $X_j$ and $\deg_{\mathcal{K}}(\ell, j, v)
    < (1+\eta)^{j-2} - \frac{a_1\log^3 n}{\eps}$ with probability at least $n^{-a_2}$
    for some fixed constants $a_1, a_2 \geq 1$. Suppose that $\deg_{\mathcal{K}}(\ell-1, j, v) = 
    (1+\eta)^{j-1} - \frac{a_1\log^3 n}{\eps} - 1$ since this is the worst case. Here, worst 
    case means the probability that the SVT is satisfied and hence the vertex moves up a level is maximized.
    We are in graph $X_j$, hence the probability we used to sample is 
    $\frac{c_1\log^3(n)}{\eps(1+\eta)^j}$. The expected number of edges we sample out of the $\deg_{\mathcal{K}}(\ell, j, v)$ edges
    using $p_j$ is at least
    \begin{align}
    \mu \geq \frac{c_1\log^3 n}{\eps(1+\eta)^j} \cdot \left((1+\eta)^{j-2} - \frac{a_1\log^3 n}{\eps} 
    - 1\right) &= \frac{c_1 \log^3 n}{\eps(1+\eta)^2} - 
    \frac{a_1 c_1 \log^6 n}{\eps^2 (1+\eta)^j} - p_j\\
    &\geq \frac{c_1 \log^3 n}{\eps(1+\eta)^3} - 
    \frac{2a_1 c_1 \log^6 n}{\eps^2 (1+\eta)^j}.\label{eq:mu-invariant-lower-bound}
    \end{align} 
    By~\cref{kcore:initial-threshold}, we set $j$ such that $(1+\eta)^j \geq \frac{c_3\log^3(n)}{\eps}$. Hence, 
    using this setting, we can further simplify~\cref{eq:mu-invariant-lower-bound} as follows:
    \begin{align}
        \frac{c_1 \log^3 n}{\eps(1+\eta)^2} - \frac{2a_1 c_1 \log^6 n}{\eps^2 (1+\eta)^j} \geq \frac{c_1 \log^3 n}{\eps(1+\eta)^2} - \frac{2a_1 c_1 \log^3(n)}{c_3\eps} = \left(\frac{1}{(1+\eta)^2} - \frac{2a_1}{c_3}\right)\frac{c_1\log^3(n)}{\eps}.
    \end{align}

    Furthermore, we can upper bound $\mu$ by
    \begin{align}
        \mu \leq \frac{c_1\log^3 n}{\eps(1+\eta)^2}.
    \end{align}
    
    Let $S_{\ell, j, v}$ be the set of edges we sampled.
    Using a multiplicative Chernoff bound (\Cref{thm:multiplicative-chernoff}), we have that the probability that our sample 
    has size larger than $(1 + \psi)\mu$ 
    is as follows, for $\psi \in (0, 1)$:
    \begin{align}
        \prob\left[|S_{\ell, j, v}| \geq (1+\psi) \mu \right] \leq \exp\left(-\frac{\mu \psi^2}{3}\right) \leq \exp\left(-\frac{\psi^2 \left(\frac{1}{(1+\eta)} - \frac{2a_1}{c_3}\right) c_1 \log^3 n}{3\eps}\right).
    \end{align}

    The SVT introduces at most $\frac{a_3\log^3(n)}{\eps}$ additive error for some constant $a_3$, with high probability,
    and so the value we are comparing against the threshold of $p_j \cdot (1+\eta)^{j-1}$ is at most
    $|S_{\ell, j, v}| + \frac{2a_3\log^3(n)}{\eps}$,
    with high probability. We proved above the probability that $|S_{\ell, j, v}| \leq \frac{(1+\psi)c_1\log^3(n)}{\eps(1+\eta)}$ is at least
    $1 - \exp\left(-\frac{\psi^2 \left(\frac{1}{(1+\eta)} - \frac{2a_1}{c_3}\right) c_1\log^3(n)}{3\eps}\right)$. 
    Conditioning on the SVT error bound
    and the sampling bound above,
    the SVT fails the check at level $\ell-1$ if
    \begin{align}
        |S_{\ell, j, v}| + \frac{2a_3\log^3(n)}{\eps} 
        &\leq \frac{(1+\psi)c_1\log^3(n)}{\eps(1+\eta)^3} + \frac{2a_3\log^3(n)}{\eps} \\
        &< p_j \cdot (1+\eta)^{j-1} \\ 
        &= \frac{c_1\log^3 n}{\eps(1+\eta)}.
    \end{align}
    Note that the inequality is satisfied as long as we have
    \begin{align}
        \frac{(1+\psi)c_1}{(1+\eta)^2} + 2a_3 < \frac{c_1}{1+\eta}.
    \end{align}
    To satisfy the above expression, we require $\psi < \eta$ and $\eta > 0$, which is easily satisfied by the 
    constraints of our problem, and
    also $c_1 > -\frac{(\eta + 1)^2 2a_3}{\psi-\eta}$. We set $c_1$ and $c_3$
    to be appropriately large enough constants in terms of 
    $\eta, a_1, a_3, \psi$ to amplify the probability of success to satisfy with high probability. 
    
    For an appropriate setting of $c_1, c_3, \psi, a_1, a_3$, 
    we have that the SVT failed the check at level $\ell-1$
    with probability at least $1 - n^{-c}$ for any constant $c\geq 1$. %
    This contradicts the fact that 
    $v$ is at level $\ell$, with high probability. 
\end{proof}

For the formal proof of approximation guarantee,
we also require the following folklore result.
\begin{lemma}[Folklore]\label{lem:folklore}
    Suppose we perform the following \emph{peeling} procedure. Provided an input
    graph $G = (V, E)$, we iteratively remove (peel) nodes with degree $\leq d^*$ for
    some $d^* > 0$ until no nodes with degree $\leq d^*$ remain. Then, $\core(i) \leq d^*$ for each removed node $v$
    and all remaining nodes $w$ (not peeled) have core number $\core(w) > d^*$.
\end{lemma}

\cref{lem:kcore-invariant-1} and~\cref{lem:kcore-invariant-2} together with the proof of Theorem 4.7 of~\cite{DLRSSY22} 
gives our final utility guarantees for \cref{thm:kcore-formal} below.
\begin{lemma}\label{lem:kcore-approx-alg}
    \cref{alg:k-core-main} returns a $\left(2 + \eta, O\left({\eta^{-2}}\eps^{-1}\log^3 n\right)\right)$-approximate $k$-core decomposition.
\end{lemma}

\begin{proof}
    Our proof follows almost verbatim from \cite[Theorem 4.7]{DLRSSY22} except we use~\cref{lem:kcore-invariant-1} and~\cref{lem:kcore-invariant-2}
    instead. Since we return the same core estimate as the equivalent vertex in $\mathcal{K}$, we only need to show the approximation returned from 
    $\mathcal{K}$ gives the correct approximation. In the below proof, a \emph{group} $g$ is defined as the $g$-th data structure in $\mathcal{K}$; hence
    the topmost level in the $g$-th group is the topmost level in the $g$-th level data structure in $\mathcal{K}$. We let $Z_{\ell}$ denote the 
    set of nodes in levels $\geq \ell$ in any group $g$.

    In this proof, when we refer to the level of a node $i$ in subgraph $j$, we mean the level of any vertex $i \in V$
    in $\mathcal{K}$. Using notation and definitions from~\cite{DLRSSY22},
    let $\kest(i)$ be the core number estimate of $i$ in $\mathcal{K}$ 
    and $\core(i)$ be the
    core number of $i$. First, we show that 
    \begin{align}
        \text{if } \kest(i) \leq (2+\eta)(1+\eta)^{g'}, \text{ then }
        \core(i) \leq (1+\eta)^{g' + 1} + \frac{d\log^3 n}{\eps}\label{eq:kcore-bound-1}
    \end{align}
    with probability at least $1 - \frac{1}{n^{c-1}}$ for any group $g'
    \geq 0$, for any constant $c \geq 2$, and for appropriately large constant $d > 0$.
    Recall $F-1$ is the topmost level of group $g'$. In order
    for $(2+\lambda)(1+\eta)^{g'}$ to be the estimate of $i$'s core number, 
    $i$ must be in the topmost level of $g'$ but below the topmost levels of all $g'' > g'$.
    Suppose we consider $i$'s level in group $g' + 1$ and let $\ell$ be $i$'s level in $g' + 1$.
    By~\cref{lem:kcore-invariant-1}, if $\ell < F-1$,
    then $\deg_{\mathcal{K}}(\ell, g', i) \leq (1+\eta)^{g'} + O\left(\frac{\log^3 n}{\eps}\right)$ 
    with probability at least $1- \frac{1}{n^c}$ for any constant
    $c \geq 1$. Furthermore, each node $w$ at the same or lower level $\ell' \leq \ell$ 
    has $\deg_{\mathcal{K}}(\ell', g', w) \leq (1+\eta)^{g'} + O\left(\frac{\log^3 n}{\eps}\right)$, also by~\cref{lem:kcore-invariant-1}. 
    
    Suppose we perform the following iterative procedure:
    starting from level $level = 0$ in $g'$,
    remove all nodes in level $level$ during this turn and set $level
    \leftarrow level + 1$ for the next turn. Using this procedure, the nodes in
    level $0$ are removed in the first turn, the nodes in level $1$ are
    removed in the second turn, and so on until the graph is empty. 
    Let $d_{level}(i)$ be the induced degree of
    any node $i$ after the removal in the $(level-1)$-st turn and 
    prior to the removal in the $level$-th turn. Since we showed
    that $\deg_{\mathcal{K}}(level, g', i) \leq (1+\eta)^{g'} + O\left(\frac{\log^3 n}{\eps}\right)$ for any node
    $i$ at level $level < F-1$, node $i$ on level $level < F-1$
    during the $level$-th turn has $d_{level}(v) \leq (1+\eta)^{g'} + O\left(\frac{\log^3 n}{\eps}\right)$.
    Thus, when $i$ is removed in the $level$-th turn, it has degree $\leq
    (1+\eta)^{g'} + O\left(\frac{\log^3 n}{\eps}\right)$. Since all nodes removed before $i$
    also had degree $\leq (1+\eta)^{g'} + O\left(\frac{\log^3 n}{\eps}\right)$ when they were removed,
    by~\cref{lem:folklore}, node $i$
    has core number $\core(i) \leq (1+\eta)^{g'} + O\left(\frac{\log^3 n}{\eps}\right)$ with probability $\geq 1
    - \frac{1}{n^c}$ for any $c\geq 1$. By the union bound over all nodes,
    this proves that $\core(i) \leq (1+\eta)^{g'} + O\left(\frac{\log^3 n}{\eps}\right)$ for all $i \in [n]$ 
    with $\kest(i) \leq (2+\eta)(1+\eta)^{g'}$
    for all constants $c \geq 2$ with probability at least $1 - \frac{1}{n^{c-1}}$.
    
    Now we prove our lower bound on $\kest(i)$.
    We prove that for any $g' \geq 0$, 
    \begin{align}
    \text{if } \kest(i) \geq (1+\eta)^{g'}, \text{ then }
    \core(i) \geq \frac{(1+\lf)^{g'-2} - O\left(\frac{\log^3 n}{\eps}\right)}{\upexpold}\label{eq:kcore-bound-2}
    \end{align}
    for all nodes $i$
    in the graph with probability at least $1 -
    \probfactorminusone$ for any constant $c \geq 2$.
    We assume for contradiction 
    that with probability $> \frac{1}{n^{c-1}}$
    there exists a node $i$ where $\kest(i) \geq (1+\lf)^{g'}$ and $\core(i) <
    \frac{(1+\lf)^{g'-2} - O\left(\frac{\log^3 n}{\eps}\right)}{\upexpold}$.
    This, in turn, implies that $\core(i) \geq \frac{\downexp^{g'-2} - O\left(\frac{\log^3 n}{\eps}\right)}{\upexpold}$ 
    for all $i \in [n]$ where $\kest(i) \geq (1+\lf)^{g'}$ holds
    with probability $< 1 - \probfactorminusone$.
    To consider this case, we use the \emph{pruning} process defined
    in~\cref{lem:folklore}. In the below proof, let $d_S(i)$ denote the 
    induced degree of node $i$ in the subgraph induced by nodes in $S$.
    For a given subgraph $S$, we \emph{prune} $S$ 
    by repeatedly removing all nodes $i \in S$ whose $d_S(i)
    < \frac{(1+\lf)^{g'-2} - \frac{a\log^3 n}{\eps}}{\upexpold}$ for appropriately selected constant $a > 0$. 
    We consider levels from the same group $g'$ since levels 
    in groups lower than $g'$ will also have smaller upper bound cutoffs, leading to an easier proof. 
    Let $j$ be the number of levels below level $F-1$. We prove via induction that the
    number of nodes pruned from the subgraph induced by $Z_{F-1 - j}$ must
    be at least
    \begin{align}
        \left(\frac{\upexpold}{2}\right)^{j-1}\left(\lbexp\right)\label{eq:pruned}.
    \end{align}

    We first prove the base case when $j = 1$. In this case, we know that for any node $i$ in level $F-1$ in group $g'$,
    it holds that $\deg_{\mathcal{K}}(F - 1, g', i) \geq (1+\eta)^{g'-2} - O\left(\frac{\log^3 n}{\eps}\right)$ with probability
    $\geq 1 - \frac{1}{n^c}$ for any $n\geq 1$ by~\cref{lem:kcore-invariant-2}. Taking the union bound over all nodes in level $F-1$
    shows that this bound holds for all nodes $i \in [n]$ with probability at least $1 - \frac{1}{n^{c-1}}$.
    All below expressions hold with probability at least $1 - \frac{1}{n^{c-1}}$; for simplicity, we omit
    this phrase when giving these expressions.
    In order to prune $i$ from the graph, we must prune at least
    \begin{align*}
        &\left((1+\lf)^{g'-2} - \frac{a\log^3 n}{\eps}\right) -
        \frac{(1+\lf)^{g'-2} - \frac{a\log^3 n}{\eps}}{\upexpold}\\
        &= \left(\downexp^{g'-2} - \frac{a\log^3 n}{\eps}\right) \cdot \left(1 -
            \frac{1}{\upexpold}\right).
    \end{align*}
    neighbors of $i$ from $Z_{F - 2}$. We must prune at least this many neighbors
    in order to reduce the degree of $i$ to below the cutoff for pruning a vertex (as we show more formally below).
    In the case when $\downexp^{g'} \leq 4(1+\eta)^2(\frac{a\log^3 n}{\eps})$,\footnote{$4$ is an arbitrarily chosen large enough constant.}
    then our original approximation statement in the lemma
    is trivially satisfied because the core number is always non-negative and so even if $\core(i) = 0$, 
    this is still within our additive approximation bounds.
    Hence, we only need to prove the case when $\downexp^{g'} > 4(1+\eta)^2\left(\frac{a\log^3 n}{\eps}\right)$.

    Then, if fewer than $\lbexp$
    neighbors of $i$ are pruned from the graph, then $i$ is not pruned from the
    graph. If $i$ is not pruned from the graph, then $i$ is part of a $\left(
    \frac{(1+\lf)^{g'-2} - \frac{a\log^3 n}{\eps}}{\upexpold}\right)$-core (by~\cref{lem:folklore})
    and $\core(i) \geq
    \frac{(1+\lf)^{g'-2} - \frac{a\log^3 n}{\eps}}{\upexpold}$ for all $i \in [n]$ where $\kest(i) \geq (1+\lf)^{g'}$
    with probability $\geq 1 - \frac{1}{n^{c-1}}$,
    a contradiction with our assumption. Thus, it must be
    the case that there exists at least one $i$ where at least $\lbexp$
    neighbors of $i$ are pruned in $Z_{T(g') - 1}$.
    For our induction hypothesis, we
    assume that at least the number of nodes as indicated in~\cref{eq:pruned}
    is pruned for $j$ and prove this for $j + 1$.

    Each node $w$ in levels $F - 1 - j$ and above has
    $d_{Z_{F-2- j}}(w) \geq
    (1+\lf)^{g'-2} - \frac{a\log^3 n}{\eps}$ by~\cref{lem:kcore-invariant-2} (recall that all $j$ levels
    below $F-1$ are in group $g'$). For simplicity of expression,
    we denote $\lbabrv \triangleq \lbexp$.
    Then, in order to prune the
    $\left(\frac{\upexpold}{2}\right)^{j-1}\lbabrv$
    nodes by our induction hypothesis, we must prune at least
    \begin{align}
        &\left(\frac{\upexpold}{2}\right)^{j-1}\lbabrv \cdot
        \left(\frac{(1+\lf)^{g'-2} - \frac{a\log^3 n}{\eps}}{2}\right)
        \label{eq:pruned-edges}
    \end{align}
    edges where we ``charge'' the edge to the endpoint that is pruned last. 
    (Note that we actually need to prune at least 
    $\left(\downexp^{g'-2} - \frac{a\log^3 n}{\eps}\right) \cdot \left(1 -
    \frac{1}{\upexpold}\right)$ edges per pruned node
    as in the base case but 
    $\frac{\left(\downexp^{g'-2} - \frac{a\log^3 n}{\eps}\right)}{2}$ lower bounds this amount.)
    Each pruned node prunes less than
    $\frac{(1+\lf)^{g'-2} - \frac{a\log^3 n}{\eps}}{\upexpold}$
    edges. Thus, using~\cref{eq:pruned-edges}, the number of nodes
    that must be pruned from $Z_{F -2 - j}$ is
    \begin{align}
        \left(\frac{\upexpold}{2}\right)^{j-1} \lbabrv \cdot
        \frac{(1+\lf)^{g'-2} - \frac{a\log^3 n}{\eps}}{2\left(\frac{(1+\lf)^{g'-2}-\frac{a\log^3 n}{\eps}}
        {\upexpold}\right)} = \left(\frac{\upexpold}{2}\right)^j \lbabrv.
        \label{eq:induction}
    \end{align}
    \cref{eq:induction} proves our induction step.
    Using~\cref{eq:pruned}, the number of nodes that must be pruned from
    $Z_{F - 1 - 2\log_{\upexpold/2}\left(n\right)}$ is
    greater than $n$ since $J \geq 1$ by our assumption that $\downexp^{g'-2} > 4\left(\frac{a\log^3 n}{\eps}\right)$:
    \begin{align}
        \left(\frac{\upexpold}{2}\right)^{2\log_{\upexpold/2}\left(n\right)} \cdot \lbabrv \geq n^2.
        \label{eq:final-eq}
    \end{align}
    Thus, at $j =
    2\log_{\upexpold/2}\left(n\right)$, we run out of
    nodes to prune. We have reached a contradiction as we require pruning greater than 
    $n$ nodes with probability at least $\frac{1}{n^{c-1}} \cdot \left(1- \probfactorminusone\right) > 0$ 
    via the union bound over all nodes where $\kest(i) \geq (1+\lf)^{g'}$ and using our assumption that
    with probability $> \frac{1}{n^{c-1}}$ there exists an $i \in [n]$ 
    where $\core(i) < \frac{(1+\lf)^{g'-2} - \frac{a\log^3 n}{\eps}}{\upexpold}$. 
    This contradicts with the fact that
    more than $n$ nodes can be pruned with $0$ probability.
    
    From~\cref{eq:kcore-bound-1}, we can first obtain the inequality $\core(i) \leq \kest(i) + \frac{d\log^3 n}{\eps}$ since
    this bounds the case when $\kest(i) = (2+\eta)(1+\eta)^{g'}$; if $\kest(i) < (2+\lambda)(1+\eta)^{g'}$
    then the largest possible value for $\kest(i)$ is $(2+\eta)(1+\eta)^{g'-1}$ by our algorithm and
    we can obtain the tighter bound of $\core(i) \leq (1+\eta)^{g'} +  \frac{d\log^3 n}{\eps}$. 
    We can substitute $\kest(i) = 
    (2+\eta)(1+\eta)^{g'}$ since $(1+\eta)^{g' + 1} < (2+\eta)(1+\eta)^{g'}$ for all $\eta \in (0, 1)$ and 
    $\eta > 0$. Second, from~\cref{eq:kcore-bound-2},
    for any estimate $(2+\eta)(1+\eta)^{g}$, the largest $g'$ for which this estimate
    has $(1+\eta)^{g'}$ as a lower bound is $g' = g + \floor{\log_{(1+\eta)}(2+\eta)} \geq g 
    + \log_{(1+\eta)}(2+\eta) -1$. Substituting this $g'$ into $\frac{(1+\lf)^{g'-2} - \frac{a\log^3 n}{\eps}}{\upexpold}$ 
    results in 
    \begin{align*}
        \frac{(1+\lf)^{g + \log_{(1+\eta)}(2+\lambda) -3} - \frac{a\log^3 n}{\eps}}{\upexpold} = \frac{\frac{(2+\lambda)(1+\lf)^{g}}{(1+\lf)^3} - \frac{a\log^3 n}{\eps}}{\upexpold} = \frac{\frac{\kest(i)}{(1+\lf)^3} - \frac{a\log^3 n}{\eps}}{\upexpold}.
    \end{align*}
    Thus,
    we can solve $\core(i) \geq \frac{\frac{\kest(i)}{(1+\lf)^3}-\frac{a\log^3 n}{\eps}}{(2+\lambda)(1+\lf)^4}$
    and $\core(i) \leq \kest(i) + \frac{d\log^3 n}{\eps}$ to obtain
    $\core(i) - \frac{d\log^3 n}{\eps} \leq \kest(i) \leq
    ((2+\eta)(1+\eta)^5)\core(i) + \frac{a(2+\eta)(1+\eta)^5\log^3 n}{\eps}$. Simplifying, we obtain
    \begin{align*}
        \core(i) - &O(\eps^{-1}\log^3 n) \leq \kest(i) \leq 
        (2+\eta)(1+\eta)^5 \core(i) + O({\eps^{-1} \log^3 n})
    \end{align*}
    which is consistent with the definition of
    a $(2+\eta', O(\eps^{-1}\log n))$-factor approximation algorithm
    for core number for any constant $\eta' > 0$ and appropriately
    chosen constant $\eta' > \eta \in (0, 1)$ that depend on $\eta'$.
\end{proof}

It remains only to analyze the space usage of \Cref{alg:k-core-main}.
\begin{lemma}\label{lem:kcore-space}
    \cref{alg:k-core-main} uses $O\left(\frac{n\log^5 n}{{\eta^4}\eps}\right)$ space, with high probability.
\end{lemma}

\begin{proof}
    
    The proof follows from a Chernoff bound.
    We set the probability of sampling to $\frac{c_1\log^3 n}{\eps(1+\eta)^j}$,
    we sample $O\left(\frac{\log^3 n}{\eps}\right)$ edges per level, per node, for each of $F$ levels in each of $|Q|$ graphs. Hence, 
    in total over $F \cdot |Q|$ levels, we use $O\left(\frac{\log^5 n}{\eps}\right)$ space per node, with high probability.
\end{proof}

Finally,
combining \Cref{lem:dp-k-core,lem:kcore-approx-alg,lem:kcore-space} yields the proof of \Cref{thm:kcore-formal}.

\section{Proofs of \texorpdfstring{\Cref{thm:densest-subgraph,thm:densest-subgraph better approx more space}}{Densest Subgraph Theorems}}\label{apx:densest-subgraph}
\SetKwProg{Fn}{Function}{}{end}\SetKwFunction{FRecurs}{FnRecursive}%
\SetKwFunction{FnInsert}{PrivateDSG}
\SetKwFunction{FnUpdate}{SampleEdge}
\SetKwFunction{FnGetDensity}{GetNonPrivateDensity}
\SetKwFunction{FnSubgraph}{GetPrivateApproxDensestSubgraph}
\SetKwProg{myclass}{Class}{}{}

We now restate and prove \Cref{thm:densest-subgraph,thm:densest-subgraph better approx more space}.
\densestSubgraph*
\densestSubgraphBetterApproxMoreSpace*

We begin with an $\eps$-edge differentially private algorithm in the continual 
release model for insertion-only streams that releases, in addition to the approximate density of the densest subgraph,
a differentially private set of vertices whose induced subgraph is an approximation of the 
maximum densest subgraph. 
We now introduce a sampling method that uses $O\left(\eps^{-1} \eta^{-2} n\log n \right)$ %
total space\footnote{Here, $\eta$ is the 
factor used in the multiplicative approximation.} 
over the course of the stream allowing us to match the space bounds (up to factors of $\poly(\log n)$ and $1/\eps$)
of the best non-private algorithms for approximate densest subgraphs while also matching the multiplicative 
approximation factor.

\begin{algorithm}[htp]
    \caption{Data Structure for Densest Subgraph in Adaptive Insertion-Only Streams\label{alg:insertion-only dsg} }
    \myclass{\FnInsert{$\varepsilon$, $\eta$, $n$, $T$}}{
        Maintain sampled edges $X \leftarrow \varnothing$.\label{one-shot:sampled}\\
        Set $m' \leftarrow \frac{c_3 n \log^2(n)}{\varepsilon \eta^2}$ using a large enough constant $c_3 > 0$.\label{one-shot:max-edges}\\ %
        Initialize empty sampling hashmap $H$\label{one-shot:sample-hashmap} to store edges and their associated random weight.\\
        Initialize class \textsc{SVT}($\eps, 1, c_5\log_{1+\eta}(n)$). \qquad(\Cref{alg:sparse vector technique})\label{one-shot:svt-func} \\
        \Fn{\FnUpdate{$e_t$, $m$, $\eps$}}{
            \If{\textsc{SVT.ProcessQuery}$(m, m')$ is ``above''}{\label{one-shot:svt-m-query} %
                $m' \leftarrow (1+\eta) \cdot m'$.\label{one-shot:increase-m}\\

                $p \leftarrow \min\left(1, \frac{c_4 n\log^2(n)}{\eps \cdot \eta^2 m'}\right)$.\label{one-shot:set-prob}\\
                \For{each edge $e \in H$}{\label{one-shot:iterate-stored-edges}
                    \If{$H[e] \leq p$}{\label{one-shot:check-ht}
                        Keep $e$ in $X$.
                        \label{one-shot:keep-edge}
                    } 
                    \Else{
                        Remove $e$ from $X$ and $H$.\label{one-shot:remove-edge} %
                    }
                }
            }
            \If{$e_t\neq \perp$} {
                Sample $h_{e_t} \sim U[0, 1]$ uniformly at random in $[0, 1]$.\label{one-shot:sample-h-t-value}\\
                \If{$h_{e_t} \leq p$}{\label{one-shot:less-than-prob}
                    Store $X \gets X \cup \{e_t\}$.\label{one-shot:sample-edge}\\
                    Add $H[e_t] = h_{e_t}$.\label{one-shot:store-sample}
                }
            }
        }
        \Fn{\FnSubgraph{$\eps$}}{
            Return $\textsc{PrivateDensestSubgraph}(\eps, X)$. \qquad (\Cref{thm:improved private static densest subgraph}) \label{one-shot:return-private-ds}      
        }
        \Fn{\FnGetDensity{}}{
            Return $(\textsc{ExactDensity}(X), p)$. \qquad (\Cref{thm:densest value}) \label{one-shot:return-ds-value}
        }
	}

    \storelines
\end{algorithm}

\subsection{Detailed Algorithm Description}\label{apx:dsg-description}
We first define a data structure $\FnInsert$ (see \Cref{alg:insertion-only dsg}) for maintaining an approximate densest subgraph in insertion-only streams. 
This data structure takes as input an accuracy parameter $\eta$, the number of nodes $n$, and a bound on the stream size $T$ and uses as a black-box $\textsc{PrivateDensestSubgraph}$ (an $\eps$-DP static densest subgraph algorithm) and
$\textsc{Densest-Subgraph}$ (a non-private static densest subgraph algorithm).
This data structure has three procedures:
an update procedure ($\FnUpdate$), a procedure for getting a private densest subgraph ($\FnSubgraph$), and a procedure 
for returning the non-private exact density of the sampled subgraph ($\FnGetDensity$).

Our main algorithm (see~\cref{alg:insertion-main}) uses the data structure above and performs its initialization based on the input accuracy parameter $\eta$, the number of nodes $n$, the privacy parameter $\epsilon$, and
an upper bound $T$ on the total number of updates in the stream.

{
\begin{algorithm}[htp]
    \caption{Algorithm for Densest Subgraph in Adaptive Insertion-Only Streams \label{alg:insertion-main}}
    \restorelines
    Initialize counter of number of edges $m \leftarrow 0$. \label{one-shot:initialize-edge-counter}\\
    Initialize privacy parameters $\eps_1 \leftarrow \frac{\eps}{3c_2\log_{1 + \eta}(n)}, \eps_2 \leftarrow \eps/3$. \label{one-shot:initialize-eps}\\
    Initialize \textsc{PrivateDSG}($\eps_1, \eta$, $n$, $T$).\label{one-shot:initialize-dsg} \\
    Initialize density threshold
    $L \leftarrow \frac{(1+\eta) c_1 \cdot \log^2(n)}{\eps {\eta^2}}$. \label{one-shot:initialize-density-threshold}\\
    Initialize class \textsc{SVT}$(\eps_2, 1, c_2\log_{1+\eta}(n))$. \qquad (\Cref{alg:sparse vector technique}) \label{one-shot:initialize svt} \\
    \For{each new update $e_t$}{\label{one-shot:for}
        \If{$e_t \neq \bot$}{\label{one-shot:check-if-bot}
            $m \leftarrow m + 1$.\label{one-shot:increment-m}
        }
       \textsc{PrivateDSG.SampleEdge($e_t, m, \eps_2$)}.\label{one-shot:sample}\\ 
        $D, p \leftarrow \textsc{PrivateDSG.GetNonPrivateDensity}()$.\label{one-shot:non-private-density} \\
        $S \leftarrow V$.\\
        \If{\textsc{SVT.ProcessQuery}$(D, p\cdot L)$ is ``above''}{
        \label{one-shot:exceed-threshold-svt}       
            $L \leftarrow (1+\eta) \cdot L$.\label{one-shot:increase-L}\\
            $S \leftarrow \textsc{PrivateDSG.GetPrivateApproxDensestSubgraph}(\eps_1)$.\label{one-shot:get-new-private-graph}
        }
        Release $S$.\label{one-shot:release-stored} %
    }
\end{algorithm}}

We perform the following initializations:
\begin{enumerate}[(i)]
    \item an instance of \textsc{PrivateDSG}($\eps$, $\eta$, $n$, $T$)
     (\cref{one-shot:initialize-dsg}),
    \item a counter for the number of edge insertions we have seen so far in our stream (\cref{one-shot:initialize-edge-counter}),
    \item the privacy parameters for the different parts of our algorithm (\cref{one-shot:initialize-eps}), 
    \item and the initial density threshold for the density for returning a new private densest subgraph solution (\cref{one-shot:initialize-density-threshold}) where $c_1 > 0$ is a constant.
\end{enumerate}
The initial cutoff for the density is 
equal to our additive error since we can just return the entire set of vertices
as long as the density of the densest subgraph is (approximately)
less than our additive error.

We now receive an online stream of updates (which can be empty $\bot$) one by one. The $t$-th update is denoted $e_t$ (\cref{one-shot:for}).
For each update, we first check whether the update is an edge insertion or $\bot$ (\cref{one-shot:check-if-bot}); if it is not $\bot$, 
then we increment $m$ (\cref{one-shot:increment-m}). We then
call \textsc{PrivateDSG.SampleEdge($e_t, m, \eps_2$)} which decides how to sample the edge (\cref{one-shot:sample}).

The $\textsc{SampleEdge}$ procedure is within the \textsc{PrivateDSG} class which maintains the following:
\begin{enumerate}[(i)]
    \item a set $X$ (\cref{one-shot:sampled}) of sampled edges in the stream,
    \item an estimate, $m'$, of the number of edges seen so far in the stream (\cref{one-shot:max-edges}),
    \item and the probability $p$ by which the current edge in the stream is sampled (\cref{one-shot:set-prob}). 
\end{enumerate}
The initial probability that we sample an edge is $1$ since we have not seen many edges and can keep all of them within our space bounds. %

In order to sample according to the desired (private) probability,
we first privately check whether our current number of edges we have seen exceeds the threshold $m'$ using the \emph{sparse vector technique} (\cref{alg:sparse vector technique}). 
Our sparse vector technique function (\cref{one-shot:svt-func}) is initialized with 
\begin{enumerate}[(i)]
    \item the privacy parameter $\eps_2$, 
    \item the sensitivity of the query (which is $1$ in our setting), 
    \item and the maximum number of successful queries (we perform queries until we've reached the end of the stream or we exceed the successful queries threshold). 
\end{enumerate}
Then, the SVT query is run with the query, $m$, (which is the current number of edges we've seen so far) and the 
threshold, $m'$ (\cref{one-shot:svt-m-query}). 
If the SVT query passes, then we update $m'$ to a larger threshold (\cref{one-shot:increase-m}) and
also update the probability of sampling in terms of the new $m'$ (\cref{one-shot:set-prob}) (using constant $c_4$). 
We then sample the 
edges by sampling a value $h_{e_t}$ uniformly from $[0, 1]$ (\cref{one-shot:sample-h-t-value}).
If $h_{e_t} \leq p$ (\cref{one-shot:less-than-prob}), then we store $e_t$ in $X$ (\cref{one-shot:sample-edge}) and the value $h_{e_t}$ 
in our hashmap $H$ (\cref{one-shot:store-sample}). Whenever $p$ changes, we also have to resample the edges in $X$; to do this, 
we keep the edge in $X$ if $H[e_t] \leq p$ and remove $e_t$ otherwise (\cref{one-shot:check-ht,one-shot:keep-edge,one-shot:remove-edge}).
This ensures that the marginal probability of sampling an edge at any timestamp $t$ is $p$.

After we have sampled our edges, we now need to obtain the \emph{non-private} density of our current subgraph (\cref{one-shot:non-private-density}).
If the non-private density exceeds our density threshold $L$ via private SVT (\cref{one-shot:exceed-threshold-svt}), then
we increase our threshold (\cref{one-shot:increase-L}) and use our private densest subgraph algorithm to return 
a private approximate densest subgraph (\cref{one-shot:get-new-private-graph}). We release our 
stored private graph (\cref{one-shot:release-stored}) for all 
new updates until we need to compute a new private densest subgraph.

The densest subgraph algorithm that we use to obtain a differentially
private densest subgraph can be 
any existing static edge-DP algorithm for densest subgraph like \Cref{thm:improved private static densest subgraph} which gives a $(2, O(\eps_1^{-1} \log(n)))$-approximation.

\subsection{Privacy Guarantee}\label{apx:dsg-privacy}

In this section, we prove the privacy guarantees of our algorithm.

\begin{lemma}\label{lem:one-shot dsg privacy}
     \algoref{insertion-main} is $\eps$-DP in the continual release model on edge-neighboring insertion-only streams. 
 \end{lemma}

\begin{pf}[\Cref{lem:one-shot dsg privacy}]\sloppy
 By the definition of edge-neighboring streams (\cref{def:neighboring-streams}), 
 let $\cS =(e_1,\ldots, e_{{T}}), \cS'=(e'_1, \ldots, e'_{{T}})$ be neighboring streams of edges, i.e., 
 there exists $t^* \in [{T}]$ such that $e_{t^*} \neq  e'_{t^*}$ and $e'_{t^*} = \bot$.
 Note that $e_t$ for any $t \in [T]$ can either be an edge $\{u, v\}$ or $\bot$.
 We define the event $\cP$ where the probability for sampling $e_t$ 
 from the stream $\cS$ is set to $p(e_t)=p_t$ for every $t \in [{T}]$. 
 Similarly define $\cP'$ where the probability for sampling $e'_t$ from the stream $\cS'$ is set to $p(e'_t)=p_t$ for 
 every $t \in [{T}]$. Although we only use the probability to sample an update $e_t$ if $e_t \neq \bot$,
 our sampling of the update depends on our setting of $p_t$ which we show is the same across the 
 two streams with roughly equal probability. 

\begin{adjustwidth}{1em}{1em}
\begin{claim}\label{clm:prob-dp}
    \[\frac{\prob[\cP]}{\prob[\cP']} \leq e^{\eps_2}\]
\end{claim}
\begin{proof}
     Consider any output answer vector $\mathbf{a} \in \{ \text{``above''}, \text{``below''}, \text{``abort''} \}^T$ to SVT. 
     First, observe that the SVT (in \texttt{SampleEdge} of \Cref{alg:insertion-only dsg}) is $\eps_2$-DP by~\cref{thm:sparse vector technique}, therefore $\frac{\prob[SVT(m(\cS)) =\mathbf{a}]}{\prob[SVT(m(\cS')) = \mathbf{a}]} \leq e^{\eps_2}$
    where $m(S)$ is the number of edges in $S$. Now, observe that the probability $p$ of sampling only changes whenever SVT accepts. Since privacy is preserved after post-processing, this implies the statement in the claim. 
\end{proof}
\end{adjustwidth}

\begin{adjustwidth}{1em}{1em}
\begin{claim}
Define $I_{-t^*}^{(t)} =(I^{(t)}_1, \ldots, I^{(t)}_{t})$ for stream $\cS$ where for $t^* \neq j\leq t$,
\begin{align}
    I_j^{(t)} &= \begin{cases}
        1 & e_j \text{ is in the sample }X^{(t)} \\
        0 & \text{otherwise}.
    \end{cases}
\end{align}  
Define ${I'}_{-t^*}^{(t)}=({I'}^{(t)}_1, \ldots, {I'}^{(t)}_t)$ for stream $\cS'$ analogously. Let $\mathbf{b}^{(t)}$ be a $\{0,1\}$-vector with the same dimension as $I_{-t^*}^{(t)}$ and $ {I'}_{-t^*}^{(t)}$. 
Also let $\cI=(I_{-t^*}^{(1)}, \ldots, I_{-t^*}^{({T})})$ and $\cI'=(I_{-t^*}'^{(1)}, \ldots, I_{-t^*}'^{({T})})$.
Then for any choice of $\mathbf{b}^{(t)}$'s,
\begin{align}
    \frac{\prob[\cI=(\mathbf{b}^{(1)}, \ldots, \mathbf{b}^{({T})})\vert \cP]}{\prob[\cI'=(\mathbf{b}^{(1)}, \ldots, \mathbf{b}^{({T})}) \vert \cP']}\cdot \frac{\prob[\cP]}{\prob[\cP']} = \frac{\prob[\cI=(\mathbf{b}^{(1)}, \ldots, \mathbf{b}^{({T})}), \cP]}{\prob[\cI'=(\mathbf{b}^{(1)}, \ldots, \mathbf{b}^{({T})}), \cP']}\leq e^{\eps_2}
\end{align}
\end{claim}

\begin{proof}
    By definition, $\cP$ and $\cP'$ defines the sampling probability for all elements in $\cS$ and $\cS'$ to be $p_t$. 
    Hence, 
    it holds that $\frac{\prob[\cI=(\mathbf{b}^{(1)}, \ldots, 
    \mathbf{b}^{({T})})\vert \cP]}{\prob[\cI'=(\mathbf{b}^{(1)}, \ldots, \mathbf{b}^{({T})}) \vert \cP']} = 1$. By~\cref{clm:prob-dp},
    $\frac{\prob[\cP]}{\prob[\cP']} \leq e^{\eps_2}$ and our claim follows.
\end{proof}
\end{adjustwidth}
Let $\cQ$ be the event that $(\cI=(\mathbf{b}^{(1)}, \ldots, \mathbf{b}^{(T)}), \cP)$ and define $\cQ'$ analogously.
\begin{adjustwidth}{1em}{1em}
\begin{claim}
    Let $\mathbf{a}\in \set{\text{``above''}, \text{``below''}, \text{``abort''}}^T$ be a sequence of SVT answers, $D(\cdot)$ denote
    the function computing the sequence of densities, and $X(\cdot)$ be the produced samples from the stream.
    We have
    \begin{align}
        \frac{\prob[SVT(D(X(\cS)))= \mathbf{a} \vert \cQ]}{\prob[SVT(D(X(\cS')))= \mathbf{a} \vert \cQ']} &\leq e^{\eps_2} 
    \end{align}
\end{claim}  
\begin{proof}
Note that conditioned on the events $\cQ$ and $\cQ'$, the samples $X(\cS)$ and $X(\cS')$ differ by at most $1$ edge, specifically
the edge $e_{t^*}$. This is due to 
the fact that the same $e_t$ are sampled from both $\cS$ and $\cS'$ conditioned on $\cQ$ and $\cQ'$ except for 
$e_{t^*}$, and by definition of 
$\cS \sim \cS'$, for only one $t^* \in [{T}]$, it holds that $e_{t^*} \neq e_{t^*}'$.
Therefore the sensitivity of exactly computing the density on the sample is given by $\vert D(X(\cS)) - D(X(\cS')) \vert \leq 1$.  And SVT is $\eps_2$-DP conditioned on these events by~\cref{thm:sparse vector technique}.
\end{proof}
\end{adjustwidth}
By the chain rule of conditional probability, putting the above claims together gives us 
\begin{align}\label{eq:comb}
    \frac{\prob[SVT(D(X(\cS)))= \mathbf{a} \vert \cQ]}{\prob[SVT(D(X(\cS)))= \mathbf{a} \vert \cQ']} \cdot  \frac{\prob[\cI=(\mathbf{b}^{(1)}, \ldots, \mathbf{b}^{({T})})\vert \cP]}{\prob[\cI'=(\mathbf{b}^{(1)}, \ldots, \mathbf{b}^{({T})}) \vert \cP']}\cdot \frac{\prob[\cP]}{\prob[\cP']} &\leq e^{2\eps_2} = e^{2\eps/3}
\end{align}
Lastly, we call on the private densest subgraph algorithm at most $c_2\log_{1+\eta}(n)$ times, and since the privacy budget assigned to this procedure is $\eps_1=\frac{\eps}{3 c_2\log_{1+\eta}(n)}$, by sequential composition, this means that this operation is $(\nicefrac\eps3)$-DP. 

Hence combining \cref{eq:comb} and the observation about the private densest subgraph algorithm operation being $(\nicefrac\eps3)$-DP gives us our main theorem statement by sequential composition. 
\end{pf}

\begin{remark}
    We remark that although we are running a DP algorithm on a sample which may imply some privacy amplification bounds, these privacy savings only apply to our case if the sampling probability of the edges decreases per time step. Since we cannot guarantee that this will be the case for all input graphs, the privacy amplification bounds do not improve the worst-case privacy guarantees. 
\end{remark}

\subsection{Approximation and Space Guarantees}\label{apx:dsg-utility}

In this section, we prove the approximation factor of the approximate
densest subgraph we obtain via our sampling procedure. We first prove that 
our sampling procedure produces, with high probability, a set of edges $X = X_t \subseteq E_t$ at time $t$
such that $|\opt(G_t) - \den_{G_t}(\opt(X_t))| \leq \eta \cdot \opt(G_t)$ where $\opt(G_t)$ is the density
of the densest subgraph in the dynamic graph $G_t$ at time $t$
and $\den_{G_t}(\opt(X_t))$ is the density of the densest subgraph in the sampled graph $G[X_t]$
but with respect to the full graph $G_t$.
Then, we show that the approximate solution returned by the private static densest subgraph subroutine applied to the sub-sampled graph
translates to an approximate solution in the original graph.
In the following, we assume full independence in each of the sampled values for each edge (which is guaranteed by our algorithm).

\begin{lemma}\label{lem:induced-approx}
    Fix $\eta\in (0, 1]$.
    For every $t \in [T]$, \Cref{alg:insertion-main} returns a set of vertices $V^* \subseteq V$ where the induced
    subgraph $G_t[V^*]$ has density $\den_{G_t}(V^*) \geq 
    \frac{\den_{G_t}(V_{\opt})}{(2+\eta)} - O\left( \frac{\log^2 n}{\eps \eta} \right)$ (where $V_{\opt}$ is the set of vertices in 
    a densest subgraph at time $t$ consisting of all edges seen in the stream so far),  
    with high probability.
\end{lemma}

\begin{proof}
    For the first part of this proof, we take inspiration from the techniques given in \cite[Lemma 2.3]{EHW16}, although our proof differs in several respects due to the noise
    resulting from our private mechanisms.

    \myparagraph{Sub-Sampling Concentration}
    Let $x_e$ be the random variable indicating whether edge $e$ exists in
    $X$ and $p$ be the probability of sampling the edges. %
    Fix an arbitrary non-empty subset of vertices $\varnothing\neq V'\sset V$.
    The number of edges in $X[V']$ is given by
    $\card{X[V']} = \sum_{e \in X[V']} 1 = \sum_{e \in G_t[V']}x_e$. We use $G_t$ to denote
    the graph on edges inserted in updates $u_1, \ldots, u_t$ at timestep $t$.
    Then, we know the expectation of $\den(X[V'])$ is
    \begin{align}
        \mu = \expect[\den(X[V'])] &= \expect\left[\frac{\sum_{e \in G_t[V']}x_e}{|V'|}\right] = \frac{\expect\left[\sum_{e \in G_t[V']}x_e\right]}{|V'|} = \frac{\sum_{e \in G_t[V']}p}{|V'|}\\ 
        &= p \cdot \frac{\sum_{e \in G_t[V']}1}{|V'|} = 
        p \cdot \den_{G_t}(V'),\label{eq:density}
    \end{align}
    and the expectation of $|X[V']|$ is also 
    \begin{align}
        \mu_{|X[V']|} = \expect[|X[V']|] = \expect\left[\sum_{e \in G_t[V']}x_e\right] = p \cdot |G_t[V']|,\label{eq:expected-num-edges}
    \end{align}
    where $|G_t[V']|$ denotes the number of edges in the induced subgraph of $G_t[V']$. 
    
    Since all of the variables $x_e$ are independent, we have 
    by the multiplicative Chernoff bound (\cref{thm:multiplicative-chernoff}),

    \begin{align}
        \prob\left[ |X[V']| \geq (1+\psi) \cdot \mu_{|X[V']|} \right]
        &\leq \exp\left(-\frac{\psi^2 \mu_{|X[V']|}}{2 + \psi}\right)\label{eq:chernoff-original}\\
        & \leq \exp\left(-\frac{\psi^2 \mu_{|X[V']|}}{14}\right) &\text{When $\psi\in (0, 12]$} \label{eq:upper-bound-special}\\
        &\leq \exp\left(-\frac{\psi^2 \cdot p \cdot |G_t(V')|}{14}\right) \label{eq:upper-bound}. 
    \end{align}

    Consider the private estimate $m_t'$ of $m_t$ stored in~\cref{one-shot:max-edges} of \Cref{alg:insertion-only dsg},
    where $m_t$ is the number of true edges in the graph. 
    We add $O(T)$ instances of $\Lap(O(\eps / \log_{1+\eta}(n)))$ noise in the SVT subroutine (\Cref{one-shot:initialize svt}).
    By a union bound over $T = \poly(n)$,
    all such noises are of order $O(\eps^{-1} \log_{1+\eta}(n) \log(n)) = O(\eps^{-1}\eta^{-1} \log^2(n))$ with probability $1-1/\poly(n)$.
    Hence
    $m_t - O\left(\frac{\log^2 n}{\eps\eta}\right) 
    \leq m'_t \leq (1+\eta) \cdot m_t + O\left(\frac{\log^2 n}{\eps \eta} \right)$, with high probability.
    Note that by the choice of initialization of $m'$ (\Cref{one-shot:max-edges}),
    we sub-sample edges only once $m_t \geq \Omega(\eps^{-1}\eta^{-2} n\log^2(n)) - O(\eps^{-1}\eta^{-1}\log^2(n)) = \omega(\eps^{-1}\eta^{-1} \log^2(n))$.
    Thus, we also know that (by~\cref{one-shot:set-prob} of \Cref{alg:insertion-only dsg}), for a constant $c > 0$,
    \begin{align}
    1 \geq p &= \frac{cn\log^2(n)}{\eps \cdot \eta^2 \cdot m_t'}\\
    &\geq \frac{cn\log^2(n)}{\eps \cdot \eta^2 \cdot \left((1+\eta) \cdot m_t + O\left(\frac{\log^2 n}{\eps \eta}\right)\right)}\\
    &\geq \frac{cn\log^2(n)}{c'\left(\eps \cdot \eta^2 \cdot m_t \right)} &&\text{since $m_t=\omega(\eps^{-1}\eta^{-1}\log^2 n)$} \label{eq:p-lower-bound}
    \end{align}
    for an appropriately large constant $c' > 0$. %

    Let $k = |V'|$.
    We set $\psi = \frac{pk\eta \opt(G_t)}{2\mu_{|X[V']|}}$ in~\cref{eq:upper-bound} where $\opt(G_t)$ is
    the density of the densest subgraph in $G_t$.
    For the analysis, we consider two cases. The first case is when $\psi \leq 12$ and the second is when $\psi > 12$.\footnote{The
    two cases selected are bounded by $12$ which is an arbitrarily chosen large enough constant.}
    Suppose in the first case that $\psi \leq 12$; substituting this 
    value for $\psi$ into~\cref{eq:chernoff-original} yields the following bound
    \begin{align}
        &\prob\left[|X[V']| \geq (1+\psi) \cdot \mu_{|X[V']|}\right] \\
        &\leq \exp\left(-\frac{p^2 k^2 \eta^2 \opt(G_t)^2 \cdot \mu_{|X[V']|}}{14 \cdot 4\mu_{|X[V']|}^2}\right) \\
        &\leq \exp\left(-\frac{pk\eta^2 \opt(G_t)}{56}\right), %
    \end{align}
    where the last expression follows since $p \cdot k \cdot \opt(G_t) \geq \mu_{|X[V']|}$ due to the fact
    that by~\cref{eq:expected-num-edges}, $p \cdot k \cdot \opt(G_t) \geq p \cdot k \cdot \den(G[V']) = p \cdot (|G[V']|) = \mu_{|X[V']|}$. 
    Then, suppose that $\psi > 12$; substituting this $\psi$ into~\cref{thm:multiplicative-chernoff} gives the following probability 
    expression: $\prob[X \geq (1+\psi)\mu] \leq \exp\left(-\frac{\psi \mu}{2}\right)$ since $\frac{\psi}{2+\psi} \geq \frac{1}{2}$ for the 
    given bound on $\psi$. Thus, using this value of $\psi$ in~\cref{eq:chernoff-original} yields the following bound
    \begin{align}
        &\prob\left[|X[V']| \geq (1+\psi) \cdot \mu_{|X[V']|}\right] \\
        &\leq \exp\left(-\frac{p k \eta \opt(G_t) \cdot \mu_{|X[V']|}}{2 \cdot 2\mu_{|X[V']|}}\right) \\
        &\leq \exp\left(-\frac{pk\eta \opt(G_t)}{4}\right), %
    \end{align}

    Altogether, for both cases, our expression is bounded by 
    \begin{align}
        \prob\left[|X[V']| \geq (1+\psi) \cdot \mu_{|X[V']|}\right] \leq \exp\left(-\frac{pk\eta^2 \opt(G_t)}{56}\right).\label{eq:final-probability-upper-bound}
    \end{align}
    
    We can simplify the LHS of the above inequality by:
    \begin{align}
         \prob\left[|X[V']| \geq \mu_{|X[V']|} + \frac{pk\eta\opt(G_t)}{2}\right] &= \prob\left[\frac{|X[V']|}{pk} \geq 
         \frac{\mu_{|X[V']|}}{pk} + \frac{\eta\opt(G_t)}{2}\right]\\
         &= \prob\left[\frac{1}{p} \cdot \den(X[V']) \geq \frac{1}{p} \cdot p \cdot \den_{G_t}(V') + \frac{\eta\opt(G_t)}{2}\right]\\
         &= \prob\left[\frac{1}{p} \cdot \den(X[V']) \geq \den_{G_t}(V') + \frac{\eta\opt(G_t)}{2}\right].\label{eq:lower-bound-lhs}
    \end{align}
    
    We now consider all subsets of vertices
    of size $k$. There are at most 
    $\binom{n}{k}$ such sets.
    There are also $n$ possible values of $k$ since there are at most $n$ vertices. 
    Hence, the total probability that the bound holds for all $V' \subseteq V$ (using~\cref{eq:final-probability-upper-bound}) is
    \begin{align}
        \sum_{k=1}^n \binom{n}{k} \cdot \exp\left(-\frac{pk\eta^2\opt(G_t)}{56}\right) &\leq \sum_{k = 1}^n n^k
        \cdot \exp\left(-\frac{pk\eta^2\opt(G_t)}{56}\right)\\
        &= \sum_{k = 1}^n \exp\left(k \cdot \ln(n) - \frac{pk\eta^2\opt(G_t)}{56}\right).
    \end{align}
    Note that the densest subgraph has density at least the density of the entire graph,
    $\frac{m_t}{n}$. 
    Combined with the substitution of \Cref{eq:p-lower-bound} into the above,
    we obtain
    \begin{align}
        &\sum_{k = 1}^n \exp\left(k \cdot \ln(n) - \frac{pk\eta^2\opt(G_t)}{56}\right) \\
        &\leq \sum_{k = 1}^n \exp\left(k \cdot \ln(n) - \frac{cn\log^2(n)}{c'\left(\eps \cdot \eta^2 \cdot m_t \right)} \cdot 
        \frac{k\eta^2 m_t}{56 n}\right) &&\text{by \Cref{eq:p-lower-bound}} \\
        &\leq \sum_{k = 1}^n \exp\left(k \cdot 2\log^2(n) - \frac{c\log^2(n)}{c'\eps} \cdot 
        \frac{k}{56}\right)\\
        &= \sum_{k = 1}^n \exp\left(- \left(\frac{c}{56c'\eps} - 2 \right) \cdot k \cdot \log^2(n)\right) \\
        &\leq \sum_{k = 1}^n \exp\left(- \left(\frac{c}{56c'\eps} - 2 \right) \cdot \log^2(n)\right) &&\text{since $k \geq 1$} \label{eq:last-approx} \\
        &= n \cdot \exp\left(- \left(\frac{c}{56c'\eps} - 2 \right) \cdot \log^2(n)\right)\\
        &\leq \exp\left(- \left(\frac{c}{56c'\eps} - 4 \right) \cdot \log^2(n) \right)
    \end{align}
    where \Cref{eq:last-approx} holds when $\left(\frac{c}{56c'\eps} - 2 \right) \geq 1$.
    To obtain the high probability bound of $1/\poly(n)$ over all time stamps $t$, 
    it suffices to set $c$ to be a sufficiently large absolute constant
    and take a union bound over $T=\poly(n)$ events. 
    Hence, we have proven that with high probability,
    the density of any subgraph $G_t[V']$ is at least $\frac1p \den_X(V') - \frac\eta2 \OPT(G_t)$
    for every $t$.
    In particular,
    no induced subgraph in $X$ has estimated density greater than $(1+\frac\eta2)p\cdot \opt(G_t)$, with high probability.

    Now, we show that $\opt(G_t)$ has large enough induced density in $X$ such that the returned subgraph is a $(1+\eta)$-approximate
    densest subgraph in $G_t$. We do this by setting $V' = V(\opt(G_t))$, i.e. to the set of vertices of the densest subgraph in $G_t$. 
    By applying a Chernoff bound (\Cref{thm:multiplicative-chernoff}) with $\psi = \eta\in (0, 1]$
    and using the fact that $\OPT(G_t) \geq \nicefrac{m_t}{n}$,
    we have that
    \begin{align}
        &\prob\left[ \frac1p \den(X[V']) \leq (1-\psi) \cdot \OPT(G_t) \right]\label{eq:upper-bound-lhs} \\
        &= \prob\left[\den(X[V']) \leq (1-\psi) \cdot \mu\right] \\
        &\leq \exp\left(-\frac{\psi^2 \mu}{2}\right) = \exp\left(-\frac{\psi^2 \cdot p \cdot \OPT(G_t)}{2}\right)\\
        &\leq \exp\left(-\frac{\psi^2}2 \cdot \frac{cn\log^2(n)}{c'\left(\eps \cdot \eta^2 \cdot m_t \right)} \cdot \frac{m_t}{n} \right) &&\text{by \Cref{eq:p-lower-bound}} \\
        &= \exp\left(-\frac{c \log^2(n)}{2c'\eps}\right) &&\text{by setting } \psi = \eta. \label{eq:lower-bound}
    \end{align}
    By setting a large enough constant $c > 0$ and together with
    what we showed above, we 
    obtain that with high probability, the densest subgraph in $X$ has induced density in $G_t$ that is a $(1+\eta)$-approximation of $\opt(G_t)$.
    We take the union bound over $T=\poly(n)$ to show that for each time step $t$, our bound holds.

    \myparagraph{Approximation}
    By \Cref{thm:improved private static densest subgraph},
    the vertex set $S$ we output is a $(2, O(\eps_1^{-1}\log n))$-approximate densest subgraph of $X$ after each call to the private static densest subgraph subroutine.
    If we have yet to begin sub-sampling,
    i.e. $X=G_t$,
    then we are done.
    
    Suppose the algorithm started to sub-sample.
    Then with high probability,
    after each call to the private static densest subgraph algorithm,
    \begin{align}
        &\den_{G_t}(S) \\
        &\geq \frac1p \den_X(S) - \frac{\eta \OPT(G_t)}2 && \text{by~\cref{eq:lower-bound-lhs}} \\
        &\geq \frac1p \left( \frac{\OPT(X)}2 - O(\eps_1^{-1}\log n) \right) - \frac{\eta \OPT(G_t)}2 &&\text{by \Cref{thm:improved private static densest subgraph}} \\
        &\geq \frac{(1-\eta)\OPT(G_t)}{2} - \frac{c'\left(\eps \cdot \eta^2 \cdot m_t \right)}{cn\log^2(n)}\cdot \frac{c''\log^2(n)}{\eps \eta} - \frac{\eta \OPT(G_t)}2 &&\text{since } \eps_1 = \Theta\left( \frac\eps{\log_{1+\eta}(n)} \right) \label{one-shot:space bottleneck} \\
        &= \frac{(1-2\eta)\OPT(G_t)}{2} - \frac{c'c''\eta m_t}{cn} \\
        &\geq \frac{(1-3\eta) \OPT(G_t)}2, &&\text{for } c\geq 2c'c'', \OPT(G_t)\geq \frac{m_t}n
    \end{align}
    where \cref{one-shot:space bottleneck} follows from~\cref{eq:upper-bound-lhs}. 

    Finally,
    we account for the additional error due to not updating $S$ when the SVT (\Cref{one-shot:exceed-threshold-svt}) does not output ``above''.
    Let $D_t$ denote the exact density of the subgraph 
    and $p_t$ the sampling probability at time $t$ as returned on \Cref{one-shot:non-private-density}.
    If the SVT did not output ``above'', 
    we know that with high probability,
    \begin{align}
         D_t &\leq p_t\cdot L + O(\eps^{-1}\eta^{-1} \log^2(n)) \\
         \frac{D_t}{p_t} &\leq L + \frac{O(\eps^{-1}\eta^{-1} \log^2(n))}{p_t} \\
         &\nonumber \\
         D_t &> p_t\cdot \frac{L}{1+\eta} - O(\eps^{-1}\eta^{-1} \log^2(n)) \\
         \frac{D_t}{p_t} &> \frac{L}{1+\eta} - \frac{O(\eps^{-1}\eta^{-1} \log^2(n))}{p_t}.
    \end{align}
    Similar to the calculation we made regarding $\OPT(X)$,
    either $p=1$ and there is nothing to prove
    or $p \geq \frac{cn\log^2(n)}{c' \eps\eta^2 m_t}$ and
    \begin{align}
        \frac{O(\eps^{-1}\eta^{-1} \log^2(n))}{p_t}
        &= O(\eps^{-1}\eta^{-1} \log^2(n)) \cdot \frac{c' \eps\eta^2 m_t}{cn\log^2(n)} \\
        &\leq \frac{\eta m_t}{n} \\
        &\leq \eta \OPT(G_t).
    \end{align}
    In other words,
    assuming a sufficiently large constant $c$,
    we incur an additional $(1+2\eta)$ multiplicative error 
    and $O(\eps^{-1}\eta^{-1}\log^2 n)$ additive error due to the SVT.
    Hence,
    we obtain a $(2+O(\eta), O(\eps^{-1}\eta^{-1}\log^2 n))$-approximate densest subgraph at every timestamp $t$.
\end{proof}

Combining \Cref{lem:induced-approx,lem:one-shot dsg privacy} proves the privacy and utility guarantees of \Cref{thm:densest-subgraph}.
Before addressing the space usage,
we briefly describe the minor changes needed to obtain similar guarantees for \Cref{thm:densest-subgraph better approx more space}.

\myparagraph{Reducing the Multiplicative Factor} We remark that we require $1/p = \Omega(\eps^{-1} \eta^{-2} n\log^2 n)$ to absorb the additive error of the private static densest subgraph subroutine (\Cref{thm:improved private static densest subgraph}) into the multiplicative error (see \Cref{one-shot:space bottleneck}).
If we instead use an alternative private static densest subgraph subroutine,
say \Cref{thm:private static densest subgraph},
which yields a $O(1+\eta, \eta^{-3}\eps_1^{-1}\log^4(n))$-approximation,
we would need to set $1/p= \Omega(\eps^{-1} {\eta^{-5}} n\log^5(n))$ (\Cref{one-shot:set-prob}) since $\eps_1 = \eps/\log_{1+\eta}(n)$.
Suppose we further adjust the initial edge threshold to $m' = \Theta(\eps^{-1} {\eta^{-5}} n\log^5(n))$ (\Cref{one-shot:max-edges}),
the initial density threshold to $L = \Theta(\eps^{-1} {\eta^{-5}} n\log^5(n))$,
and repeat the proof steps of \Cref{lem:induced-approx} for this version of our algorithm.
This yields a proof of \Cref{thm:densest-subgraph better approx more space}.

\myparagraph{Space Complexity} Finally, we prove that our algorithm uses {$\tilde O(n)$} space, with high probability.

\begin{lemma}\label{lem:space-usage}
    \cref{alg:insertion-only dsg} uses $O(\eps^{-1}\eta^{-2}n\log^2(n))$ (resp.\ $O(\eps^{-1}{\eta^{-5}}n\log^5(n))$)\footnote{This space usage depends on which static algorithm we use.} 
    space with probability at least $1 - 1/\poly(n)$.
\end{lemma}

\begin{proof}
    By~\cref{one-shot:set-prob} of~\cref{alg:insertion-only dsg}, each edge is sampled with probability $p = \min\left(1, \frac{c_4 n\log^2(n)}{\eps \cdot \eta^2 m'}\right)$.
    By the guarantees of SVT and our scaling procedure in~\cref{one-shot:max-edges}, the final $m'$ will be at least $\frac{m}{1+\eta} - O(\log^2(n)/\eps)$. Since 
    we initially set $m' = \frac{c_3 n \log^2(n)}{\varepsilon \eta^2} >> c_5\log^2(n)/\eps$ (or $m' = \frac{c_3 n \log^5(n)}{\varepsilon \eta^2} >> c_5\log^2(n)/\eps$ depending
    on which static algorithm we choose), it follows that $m' \geq m/3$, with high probability. Then, the expected number of edges we sample is $\frac{3c_3\log^2(n)}{\varepsilon \eta^2}$
    (resp.\ $\frac{3c_3\log^5(n)}{\varepsilon \eta^2}$) in expectation. Since all edges are sampled independently, we use the Chernoff bound to show that the total 
    space we use is $O(\eps^{-1}\eta^{-2}n\log^2(n))$ (resp.\ $O(\eps^{-1}\eta^{-2}n\log^5(n))$) with probability at least $1 - 1/\poly(n)$.
\end{proof}

\section{Proof of \texorpdfstring{\Cref{thm:edge-dp-matching}}{Matching Theorem}}\label{apx:matching}
We now prove \Cref{thm:edge-dp-matching},
whose statement we copy below for convenience.
\edgeDPMatching*

\subsection{Detailed Algorithm Description}\label{apx:matching-description}
The pseudocode for our algorithm is given in~\cref{alg:matching-sampling}. Our algorithm works as follows. 
We set the probability of sampling edges from the stream initially to $1$. For each update
$e_t$, we sample the update into $S$ with probability $p_t$ (\cref{match:sample-edge}).
If the sampled update is not $\bot$ (\cref{match:sample-not-bot}), then we initialize counters for its endpoints 
if one or both of these counters do not already exist (\cref{match:sample-edge}). A sampled $e_t$ that is 
$\bot$ does not increase the size of $S$ (it is a no-op). 
Then, we iterate through every other edge in $S$ (\cref{match:other-edge}) and if $e'$
shares an endpoint with $e_t$ (\cref{match:share-endpoint}), then we increment the counter $c_{e'}^w$
associated with $e'$ (\cref{match:increment-counter}). If the counter exceeds our cutoff $\tilde{\alpha}$ (\cref{match:exceed-alpha}),
then we remove $e'$ and its corresponding counters from $S$ (\cref{match:remove-e}).
After performing our edge sampling, we now check using SVT whether $|S|$ (the number of edges in $S$) 
exceeds $\frac{a_3\log^2 n}{\eps \eta^2}$ for some large enough constant $a_3 > 0$ (\Cref{match:svt-s}).
If the SVT AboveThreshold query is satisfied, then we halve the probability of sampling (\cref{match:half-prob})
and resample each edge in $S$ with $1/2$ probability (\cref{match:resample}).
Finally, we output a new estimate if $|S|/p$ is greater than our previous estimate by a significantly large amount. 
Specifically, we check via SVT queries (\cref{match:svt-while}) whether $|S|$ exceeds $p\cdot (1+\eta)^{j_t}$. 
If so, we increment the counter $j_t$ and output our new estimate as $(1+\eta)^{j_t}$ (\cref{match:new-estimate}); otherwise, we output our old estimate.

\begin{algorithm}[htb!]
\caption{Sampling $E_t^*$}\label{alg:matching-sampling}
\DontPrintSemicolon
\SetKwProg{MyClass}{Class}{}{}
\SetKwFunction{ClassName}{EdgePrivateSublinearMatching}
\SetKwProg{Fn}{Function}{}{}
\SetKwFunction{FnProcessUpdate}{ProcessUpdate}
\MyClass{\ClassName{$\tilde \alpha, \varepsilon, \eta, n, T$}}{
  $S \gets \varnothing$\;
  $p_1 \gets 1$\;
  $Q_1 \gets a_1\log(n), Q_2 \gets a_2 \eta^{-1} \log(n)$ \\
  $j_1 \gets 0$ \\
  estimate $\gets (1+\eta)^{j_1}$\; 
  Initialize class $\textsc{SubsampleSVT}\gets \SVT(\nicefrac\eps2, 2, Q_1)$ \qquad (\Cref{alg:sparse vector technique}) \label{match:initialize subsample svt} \\
  Initialize class $\textsc{EstimateSVT}\gets \SVT(\nicefrac\eps2, 2, Q_2)$ \label{match:initialize estimate svt} \\
  \BlankLine
  \Fn{\FnProcessUpdate{$e_t$}} {
    \If{$e_t \not= \bot$}{\label{match:sample-not-bot}
        With probability $p_t$ add $e_t$ to $S$ and initialize counters $c_e^u \gets 0$ and $c_e^v \gets 0$. \label{match:sample-edge} \\
        \For{each edge $e' \not= e_t \in S$}{\label{match:other-edge}
        \Comment{even if $e_t$ is not sampled, it affects the counters of other edges} \\
        \If{$e' = \set{w, w'}$ shares endpoint $w$ with $e_t$}{\label{match:share-endpoint}
            Increment $c_{e'}^w$\label{match:increment-counter}\;
                \If{$c_{e'}^w > \tilde{\alpha}$}{\label{match:exceed-alpha}
                    Remove $e'$ from $S$ and delete $c_{e'}^w, c_{e'}^{w'}$ \label{match:remove-e}\;
                }
            }
        }
    }
    \BlankLine
    \While{\textsc{EstimateSVT.ProcessQuery}$(|S|, p_t\cdot (1+\eta)^{j_t})$ is ``above''} {\label{match:svt-while}
        $j_t \leftarrow j_t + 1$\\
        estimate $\gets (1+\eta)^{j_t}$\label{match:new-estimate}\;
    } 
    Output estimate. \\
    \BlankLine
    $p_{t+1} \gets p_t$, $j_{t+1} \gets j_t$ \\
    \If{\textsc{SubsampleSVT.ProcessQuery}$(|S|, \frac{a_3 \log^2(n)}{\eps\eta^2})$ is ``above''}{\label{match:svt-s} 
      $p_{t+1} \gets p_t/2$\label{match:half-prob}\;
      \For{$e=\set{u, v}\in S$} {
        With probability $\nicefrac12$ remove $e$ from $S$ and delete $c_e^u, c_e^v$ \label{match:resample} \\
      }
    }
  }
}
\end{algorithm}

\subsection{Privacy Guarantee}\label{apx:matching-privacy}

We prove the privacy guarantees of our algorithm in this section. 

\begin{lemma}
    \cref{alg:matching-sampling} is $\eps$-edge differentially private.
\end{lemma}
\begin{proof}
By the definition of edge-neighboring streams (\cref{def:neighboring-streams}), 
 let $\cS =(e_1,\ldots, e_{{T}}), \cS'=(e'_1, \ldots, e'_{{T}})$ be neighboring streams of edges, i.e., 
 there exists $t^* \in [{T}]$ such that $e_{t^*} \neq  e'_{t^*}$ and $e'_{t^*} = \bot$.
 Note that $e_t$ for any $t \in [T]$ can either be an edge $\{u, v\}$ or $\bot$.
 We define the event $\cP$ where the probability for sampling $e_t$ 
 from the stream $\cS$ is set to $p(e_t)=p_t$ for every $t \in [{T}]$. 
 Similarly define $\cP'$ where the probability for sampling $e'_t$ from the stream $\cS'$ is set to $p(e'_t)=p_t$ for 
 every $t \in [{T}]$. %

Define $I_{-t^*}^{(t)} =(I^{(t)}_1, \ldots, I^{(t)}_{t})$ for stream $\cS$ where for $t^* \neq j\leq t$,
\begin{align}
    I_j^{(t)} &= \begin{cases}
        1 & e_j \text{ is in the sample }S_t \\
        0 & \text{otherwise}.
    \end{cases}
\end{align}  
Define ${I'}_{-t^*}^{(t)}=({I'}^{(t)}_1, \ldots, {I'}^{(t)}_t)$ for stream $\cS'$ analogously. Let $\mathbf{b}^{(t)}$ be a $\{0,1\}$-vector with the same dimension as $I_{-t^*}^{(t)}$ and $ {I'}_{-t^*}^{(t)}$. 
Also let $\cI=(I_{-t^*}^{(1)}, \ldots, I_{-t^*}^{({T})})$ and $\cI'=(I_{-t^*}'^{(1)}, \ldots, I_{-t^*}'^{({T})})$.

\begin{adjustwidth}{1em}{1em}
 \begin{claim}\label{clm:cgs-st}
     For any $t\in [T]$, let $S_t$ be the set of sampled edges at timestep $t$. Conditioned on the events $\cI=(\mathbf{b}^{(1)}, \ldots, \mathbf{b}^{({T})}), \cI'=(\mathbf{b}^{(1)}, \ldots, \mathbf{b}^{({T})})$,  $\cP,\cP'$, the sensitivity of $\vert S_t \vert$ is at most 2. 
 \end{claim}
 
 \begin{proof}
    Let $S_t$ and $S'_t$ denote the sampled edges of the neighboring streams at time $t$. By definition $\cP$ and $\cP'$ define the sampling probabilities at every timestep $t \in [T]$ to be $p_t$. We consider the case when the differing edge at timestep $t^*$ is such that $e_{t^*}=\{u,v\}$ and $e'_{t^*}=\perp$.
    If $t<t^*$,
    then the input stream so far is identical.
    Hence the sampled edges are the same, conditioned on the events in the lemma,
    and there is nothing to prove.
    Consider now the case that $t\geq t^*$.
    Processing $e_{t^*}$ may initially increase the sample size by $1$ but may also cause edges to be removed. 
    We show that 
    $e_{t^*} = \set{u, v}$ can at most cause $||S_t| - |S'_t|| \leq 2$
    for all $t\geq t^*$. 
    For every edge that is incident to $u$ or $v$ and is still 
    in $S$, there is at most one edge $e'$ incident to each $w \in \{u, v\}$ that has counter value $c_{e'}^w =\tilde{\alpha}$. 
    This is due to the 
    fact that no two edges are inserted in the same timestep and the counters are incremented when a later edge adjacent to the vertex
    arrives. 
    In other words, 
    for each vertex $u$,
    the set of edges adjacent to $u$ always has distinct counts. 
    Suppose that $e'$ (adjacent to $u$ or $v$) is the edge removed from $G'$ but not removed from $G$ due to the counter $c_{e'}^w$ exceeding 
    $\tilde{\alpha}$ ($w \in \{u, v\}$). Then, any future update that removes another edge from $G'$ incident to $w$ must first cause $e'$ to be removed.
    Hence, at any point in time $S'_t$ contains at most two fewer edges incident to $w \in \{u, v\}$ compared to $S$, at most one per endpoint. $e_{t^*}$
    does not affect the counters of any edges inserted after it. Hence, $||S_t| - |S'_t|| \leq 2$ for all $t\geq t^*$. 
 \end{proof}
\end{adjustwidth}

 \begin{adjustwidth}{1em}{1em}
 \begin{claim}\label{clm:match-prob-dp}
    \[\frac{\prob[\cP]}{\prob[\cP']} \leq e^{\eps/2}\]
\end{claim}
\begin{proof}
     By \cref{clm:cgs-st}, since the sensitivity of the sampled set of edges $S$ is upper bounded by 2. The subsample SVT output is $\eps/2$-DP by~\cref{thm:sparse vector technique}. Now, observe that the probability $p$ of sampling only changes whenever the subsample SVT (\Cref{match:svt-s}) accepts. Since privacy is preserved after post-processing, this implies the statement in the claim. 
\end{proof}
\end{adjustwidth}

\begin{adjustwidth}{1em}{1em}
\begin{claim}
    \begin{align}
        \frac{\prob[\cI=(\mathbf{b}^{(1)}, \ldots, \mathbf{b}^{({T})})\vert \cP]}{\prob[\cI'=(\mathbf{b}^{(1)}, \ldots, \mathbf{b}^{({T})}) \vert \cP']}\cdot \frac{\prob[\cP]}{\prob[\cP']} = \frac{\prob[\cI=(\mathbf{b}^{(1)}, \ldots, \mathbf{b}^{({T})}), \cP]}{\prob[\cI'=(\mathbf{b}^{(1)}, \ldots, \mathbf{b}^{({T})}), \cP']}\leq e^{\eps/2}
    \end{align}
\end{claim}

\begin{proof}
    By definition, $\cP$ and $\cP'$ defines the sampling probability for all elements in $\cS$ and $\cS'$ to be $p_t$. 
    Hence, 
    it holds that $\frac{\prob[\cI=(\mathbf{b}^{(1)}, \ldots, 
    \mathbf{b}^{({T})})\vert \cP]}{\prob[\cI'=(\mathbf{b}^{(1)}, \ldots, \mathbf{b}^{({T})}) \vert \cP']} = 1$. By~\cref{clm:match-prob-dp},
    $\frac{\prob[\cP]}{\prob[\cP']} \leq e^{\eps/2}$ and our claim follows.
\end{proof}
\end{adjustwidth}
Let $\cQ$ be the event that $(\cI=(\mathbf{b}^{(1)}, \ldots, \mathbf{b}^{(T)}), \cP)$ and define $\cQ'$ analogously.
\begin{adjustwidth}{1em}{1em}
\begin{claim}
Let $\mathbf{a}\in \set{\text{``above''}, \text{``below''}, \text{``abort''}}^T$ be a sequence of the estimate SVT answers in \cref{match:svt-while}.
    We have
    \begin{align}
        \frac{\prob[SVT(\vert S \vert)= \mathbf{a} \vert \cQ]}{\prob[SVT(\vert S' \vert)= \mathbf{a} \vert \cQ']} &\leq e^{\eps/2} 
    \end{align}
\end{claim}
\begin{proof}
    By \cref{clm:cgs-st}, since the sensitivity of the sampled set of edges $S$ is upper bounded by 2, the estimate SVT output is $\eps/2$-DP by~\cref{thm:sparse vector technique}.
\end{proof}
\end{adjustwidth}
By the chain rule, putting the above claims together gives us 
\begin{align}\label{eq:match-comb}
    \frac{\prob[SVT(\vert S \vert )= \mathbf{a} \vert \cQ]}{\prob[SVT(\vert S' \vert )= \mathbf{a} \vert \cQ']}\cdot   \frac{\prob[\cI=(\mathbf{b}^{(1)}, \ldots, \mathbf{b}^{({T})})\vert \cP]}{\prob[\cI'=(\mathbf{b}^{(1)}, \ldots, \mathbf{b}^{({T})}) \vert \cP']}\cdot \frac{\prob[\cP]}{\prob[\cP']} &\leq e^{2\eps/2} = e^{\eps}.
\end{align}
\end{proof}

\subsection{Approximation and Space Guarantees}\label{apx:matching-utility}
We now show the approximation guarantees of our algorithm. The below proofs modify the proofs of~\cite{mcgregor2018simple} in the 
non-private setting to the private setting. 
In particular, we modify their proof to hold for any $\tilde \alpha \geq \alpha$.

Let $M(G_t)$ be the maximum size of a matching in input graph $G_t$ (consisting of all updates
up to and including update $e_t$).
We wish to compute an approximation of $M(G_t)$. 
Let $B_u^t$ be the \emph{last} $\tilde{\alpha} + 1$ edges incident to $u \in V$ that appears in $G_t$. 
Let $E_{\tilde\alpha}^t$ be the set of edges $\{u, v\} \in G_t$ 
where $|B_u^t| \leq \tilde\alpha$ and $|B_v^t| \leq \tilde\alpha$. 
That is, $E_{\tilde\alpha}^t$ consists of the 
set of edges $\{u, v\}$ where the number of edges incident to $u$ and $v$ that appear in the stream after 
$\{u, v\}$ is at most $\tilde{\alpha}$. We say an edge $\{u, v\}$ is \defn{good} if $\{u, v\} \in B_u^t \cap B_v^t$, and 
an edge is \defn{wasted} if $\{u, v\} \in B_u^t \oplus B_v^t = (B_u^t \cup B_v^t) \setminus (B_u^t \cap B_v^t)$.
Then, $E_{\tilde\alpha}^t$ is precisely the set of good edges in $G^t$.
In other words, $E_{\tilde\alpha}^t = \bigcup_{u \neq v\in V} \left(B_u^t \cap B_v^t\right)$. 

Our algorithm estimates the cardinality of $E^t_{\tilde\alpha}$ as an approximation of $M(G_t)$. 
So, we first relate $|E^t_{\tilde\alpha}|$ to $M(G_t)$.

\begin{lemma}\label{lem:e-alpha-mg}
    $M(G^t) \leq |E_{\tilde\alpha}^t| \leq (\tilde{\alpha} + 2) \cdot M(G_t)$.
\end{lemma}

\begin{pf}[\Cref{lem:e-alpha-mg}]
    We first prove the right-hand side of this expression. To do this, we define a fractional matching using $E_{\tilde\alpha}^t$. Let $Y_e = 
    \frac{1}{\tilde \alpha + 1}$ if $e \in E_{\tilde \alpha}^t$ and $Y_e = 0$ otherwise. Then, $\{Y_e\}_{e \in E}$ is a fractional matching with maximum
    weight $\frac{1}{\tilde\alpha+1}$.\footnote{A fractional matching is defined to be a set of weights $f(e)\geq 0$ on every edge $e$ where
    $\forall v \in V: \sum_{e \ni v} f(e) \leq 1$.} We now show a corollary of Edmonds' matching polytope theorem \cite{edmonds1965maximum}. Edmonds' matching 
    polytope theorem implies that if the weight of a fractional matching on any induced subgraph $S \subseteq G$ is at most $\frac{|S| -1}{2}$,
    then the weight on the entire graph is at most $M(G)$. Now, we show the following proposition:

    \begin{adjustwidth}{1em}{1em}
    \begin{proposition}\label{cor:edmonds matching polytope}
        Let $\{Y_e\}_{e \in E}$ be a fractional matching where the maximum weight on any edge is $\psi$. Then, $\sum_{e \in E} Y_e \leq (1+\psi)
        \cdot M(G)$. 
    \end{proposition}

    \begin{pf}[\Cref{cor:edmonds matching polytope}]
        Let $S$ be an arbitrary subset of vertices, and let $E(S)$ be the edges in the induced subgraph of $S$. We know that $|E(S)| \leq 
        \frac{|S|(|S|-1)}{2}$ and by the definition of fractional matching, $\sum_{e \in E(S)} Y_e = \frac{\sum_{u \in S}\sum_{v \in N(u)} Y_{\set{u, 
        v}}}{2} \leq \frac{\sum_{u \in S} 1}{2} = \frac{|S|}{2}$ where $\sum_{v \in N(u)} Y_{\set{u, v}} \leq 1$ by the constraints of the fractional
        matching. Thus, we know that (given the maximum weight on any edge is $\psi$)
        \begin{align}
            \sum_{e \in E(S)} Y_e &\leq \min\left(\frac{|S|}{2}, \frac{\psi |S|(|S|-1)}{2}\right) \\
            &\leq \frac{|S|-1}{2} \cdot \min\left(\frac{|S|}{|S|-1}, \psi |S|\right)\\
            &\leq \frac{|S|-1}{2} \cdot (1+\psi).
        \end{align}

        We can show the last inequality by considering two cases:
        \begin{itemize}
            \item If $\frac{|S|}{|S|-1} \leq \psi |S|$, then
            \begin{align}
                \frac{1}{|S|-1} &\leq \psi \\
                1+ \frac{1}{|S|-1} &\leq 1+ \psi \\
                \frac{|S|}{|S|-1} &\leq 1+\psi.
            \end{align}
            \item If $\psi |S| \leq \frac{|S|}{|S|-1}$, then
            \begin{align}
                \psi &\leq \frac{1}{|S|-1}\\
                \psi |S| - \psi &\leq 1\\
                \psi |S| &\leq 1 + \psi.
            \end{align}
        \end{itemize}

        Finally, let $Z_e = \frac{Y_e}{1+\psi}$. 
        We can use Edmonds' polytope theorem to show that $\sum_{e \in E}Z_e \leq \frac{|S| -1}{2} \leq M(G)$ and 
        so $\sum_{e \in E}Y_e \leq (1+\psi)\sum_{e \in E} Z_e \leq (1+\psi) \cdot M(G)$. We have proven our corollary.
    \end{pf}
    \end{adjustwidth}
    Using the proposition above with maximum weight $\frac{1}{\tilde\alpha + 1}$ implies that $\sum_{e \in E} Y_e  = \frac{|E_{\tilde\alpha}^t|}{\tilde\alpha+1} \leq (1+\frac{1}
    {\tilde\alpha+1}) M(G) = \frac{\tilde\alpha +2}{\tilde\alpha + 1} \cdot M(G)$. Hence, we have that $|E_{\tilde\alpha}^t| \leq (\tilde\alpha +2) \cdot M(G)$.

    Now, we prove the left inequality. To prove the left inequality, let $H$ be the set of vertices in $G_t$ with degree at
    least $\tilde\alpha + 1$. These are the \defn{heavy} vertices. We also define the following variables:
    \begin{itemize}
        \item $w := $ the number of good edges with \emph{no} endpoints in $H$,
        \item $x := $ the number of good edges with \emph{exactly one} endpoint in $H$, 
        \item $y := $ the number of good edges with \emph{two} endpoints in $H$,
        \item $z := $ the number of \emph{wasted} edges with \emph{two} endpoints in $H$.
    \end{itemize}

 First, $|E^t_{\tilde\alpha}| = w + x + y$. 
 Below, we omit the superscript $t$ for clarity but everything holds with respect
 to $t$.
 Now, we show the following additional equalities. 
    We will first calculate the number of edges in the $B_u$ of every $u \in H$ in terms of the variables
    we defined above. Since every 
    vertex $u \in H$ is heavy, $|B_u| = \tilde\alpha + 1$. Hence, $\sum_{u \in H}|B_u| = (\tilde\alpha + 1)|H|$. 
    Every edge $\{u, v\}$ in each of these $B_u$ must either be a good edge or a wasted edge by definition since 
    $\set{u, v}\in B_u\cup B_v$.
    For each 
    edge counted in $x$, it has one endpoint in $H$ so it is counted in exactly one $B_u$ in $H$. Furthermore, 
    for each edge counted in $z$, it is also counted in the $B_u$ of exactly one of its two endpoints since
    it is wasted (i.e.\ the other endpoint doesn't have at most $\tilde\alpha$ edges that come after it). This 
    leaves every good edge counted in $y$ which is counted in exactly two $B_u$'s. Hence, $\sum_{u \in H} |B_u| = 
    x + 2y + z = (\tilde\alpha + 1)|H|$. 

    Now, we can also compute $z + y \leq \alpha |H|$  since $z+y$ is a subset of the total number of edges in the
    induced subgraph consisting of $H$. Since we know that the graph has arboricity $\alpha$, we also know that (by the 
    properties of arboricity) that the induced subgraph consisting of $H$ has at most $\alpha |H|$ edges. 
    Hence, we can sum our inequalities: $x + 2y + z = (\tilde\alpha +1)|H|$ and $-z - y \geq -\alpha|H|$ to obtain
    $x + y \geq (\tilde\alpha - \alpha + 1)|H|$. Finally, let $E_L$ be the set of edges with no endpoints in $H$. 
    Every edge $\set{u, v}\in E_L$ is good since $B_u = N(u)$ and $B_v = N(v)$ by assumption.
    Note that $w = |E_L|$. 
    Therefore, $|E^t_{\tilde\alpha}| = w + x + y \geq |H| + |E_L|$ since $\tilde\alpha \geq \alpha$. 
    We can show that $|H| + |E_L| \geq M(G)$ since
    we can partition the set of edges in the maximum matching to edges in $E_L$ and edges incident to $H$. All of the edges
    in $E_L$ can be in the matching and at most one edge incident to each vertex in $H$ is in the matching. We have 
    successfully proven our lower bound: $|E^t_{\tilde \alpha}| \geq M(G)$. 
\end{pf}

Now we have the following algorithm for estimating $|E_{\tilde\alpha}^t|$ for every $t\in [T]$. For every sampled edge, 
$e = \{u, v\}$, the algorithm also stores counters $c_{e}^u$ and $c_e^v$ for the degrees of $u$ and $v$
in the rest of the stream. This requires an additional factor of $O(\log(\tilde\alpha))$ space. 
Thus, the algorithm maintains the invariant that each edge stored in the sample is a good edge with respect to the current $G_t$.

Let  $E^*_t = \max_{t' \leq t}(|E_{\tilde\alpha}^{t'}|)$, we now show that we get a $\left( 1+\eta, \frac{\log^2(n)}{\eps \eta} \right)$-approximation of $E_t^*$ with high probability. 
First, we note that $E^*_t$
satisfies $M(G^t) \leq E^*_t \leq (2+\tilde\alpha)M(G_t)$ since $E^*_t \geq |E_{\tilde\alpha}^t|$, $M(G_{t'}) \leq M(G_t)$, and~\cref{lem:e-alpha-mg}.

We now show our main lemma for our algorithm.

\begin{lemma}\label{lem:sampling-approx}
    \cref{alg:matching-sampling} returns a $\left( 1+\eta, \frac{\log^2(n)}{\eps \eta } \right)$-approximation of $E_t^*$ for every $t\in [T]$ with probability at least $1 - 1/\poly(n)$.
\end{lemma}

\begin{proof}
    Fix some $t\in [T]$.
    We need to show at time $t$,
    we do not exceed the SVT ``above'' budgets
    as well as that sub-sampling yields sufficiently concentrated estimates
    with probability $1-1/\poly(n)$.
    Then, taking a union bound over $T=\poly(n)$ timestamps yields the desired result.

    First, let $\tau = \frac{d\log^2(n)}{\eps\eta^2}$ for some large enough constant $d > 0$ 
    and define level $i_t$ (starting with $i_t = 2$) to be the level that contains $E_t^*$ if
    $2^{i_t-1} \cdot \tau < E^*_t \leq 2^{i_t} \cdot \tau$. Level $i_t = 1$ is defined as $0 \leq E^*_t \leq 2 \cdot \tau$. 
    Let $p_t$ denote the marginal probability that an edge is sampled from $E_{\alpha}^t$.
    Note that this probability is prior to sub-sampling on \Cref{match:svt-s}.
    In a perfect situation, 
    at time $t$,
    edge $e$ is sampled in level $i$ with probability $p_t' := 2^{-i_t}$ when $i_t \geq 2$ and with probability
    $p_t' = 1$ when $i_t = 1$, but this is not the case as our algorithm determines $p_t$ adaptively depending on the size of the sampled edges at time $t$. 
    
    \myparagraph{SVT Budget}
    We show that with probability at least $1-1/\poly(n)$,
    neither of the SVT ``above'' budgets $Q_1$, $Q_2$ in \Cref{alg:matching-sampling} are exceeded. 
    First, for the sub-sampling SVT,
    we take $T = \poly(n)$ noisy comparisons with $\Lap(O(\eps/\log(n)))$ noise.
    Hence with probability $1-1/\poly(n)$,
    all Laplace realizations are of order $O(\eps^{-1}\log^2(n))$.
    
    $Q_1 = a_1\log(n)$ bounds the number of times the subsampling SVT in~\cref{match:svt-s} is satisfied. 
    Each time this SVT is satisfied, 
    we must have
    \begin{align}
        \card{S_t} 
        \geq a_3\eps^{-1}\eta^{-2} \log^2(n) - O(\eps^{-1}\log^2(n)) 
        \geq \frac{a_3}2\eps^{-1}\log^2(n) \label{match:svt-s necessary}
    \end{align}
    for sufficiently large constant $a_3 > 0$,
    leading to the probability of sampling edges being halved. The expected number of edges that are sampled at any point is $p_t \cdot \card{E_{\tilde \alpha}^t} \leq p_t\cdot n^2$. 
    Suppose $a_1$ is sufficiently large and we halved the sampling probability $\nicefrac{Q_1}2 \geq \ceil{\log(n^2)}$ times
    so that $p_t\leq \frac1{n^2}$. 
    We can upper bound the probability that the SVT answers ``above''
    by considering the experiment 
    where we obtain a sample $\tilde S_t$
    by sampling each edge with probability $\tilde p = \frac1{\card{E_{\tilde \alpha}^t}} \geq p_t$.
    By a multiplicative Chernoff bound (\Cref{thm:multiplicative-chernoff}), 
    for $\psi = \frac{a_3}2\eps^{-1}\log^2(n) - 1$,
    \begin{align}
        \prob\left[ \card{S_t} \geq \frac{a_3}2 \eps^{-1} \log^2(n) \right]
        &\leq \prob\left[ \card{\tilde S_t} \geq \frac{a_3}2 \eps^{-1} \log^2(n) \right] \\
        &= \prob[\card{\tilde S_t} \geq (1+\psi) \cdot 1] \\
        &= \prob[\card{\tilde S_t} \geq (1+\psi) \tilde p\cdot \card{E_{\tilde \alpha}^t}] \\
        &\leq \exp\left( -\frac{\psi^2\cdot 1}{2+\psi}  \right) \\
        &\leq \exp\left(- \Omega(\log^2(n)) \right) &&\text{$n$ sufficiently large} \\
        &\leq \frac1{\poly(n)}. 
    \end{align}    
    Hence, for large enough constants $a_1, a_3$, 
    we do not exceed the subsampling SVT ``above'' budget.
    
    A similar argument shows that we do not exceed the estimate SVT ``above'' budget of $Q_2$ with probability at least $1-1/\poly(n)$. 
    In the worst case, the number of times $p$ is halved is $Q_1$ times. Thus, the maximum value of $|S|/p$ is $2^{Q_1} \cdot n^2$
    since the maximum number of edges in the graph is upper bounded by $n^2$.
    Since we are increasing our threshold by a factor of $(1+\eta)$ each time, our threshold exceeds 
    $2^{Q_1}n^2$ after $\log_{1+\eta}(2^{Q_1} n^2) = O(\eta^{-1}Q_1 + \eta^{-1}\log(n)) = O(\eta^{-1} Q_1)$ times. If we set $a_2$ to be a large enough constant, 
    we do not exceed the ``above'' budget for the estimate SVT.

    \myparagraph{Concentration}
    In a perfect situation, edge $e$ is sampled at time $t$ with probability $p'_t= 2^{-i_t}$ for $i_t \geq 2$ and with probability
    $p'_t = 1$ for $i_t = 1$. 
    In that case, we can use the multiplicative Chernoff bound to bound the probability that our estimate concentrates as our sampling probability would be sufficiently high enough to do so. However, it is not the
    case that $p_t$ is guaranteed to be $2^{-i_t}$ since it is determined adaptively by~\cref{alg:matching-sampling}. 
    Hence, we need a slightly more
    sophisticated analysis than simply using the multiplicative Chernoff bound with $p'_t$. 
    
    We consider two cases. If $p_t \geq p'_{t}$, then we can lower bound the probability of success by what we would obtain using $p_t'$ and we can use the multiplicative Chernoff bound in this case. 
    Now, suppose that
    $p_t < p'_t$; then, we claim that this case never occurs in~\cref{alg:matching-sampling} with probability at least 
    $1-1/\poly(n)$. 
    We now formally present these arguments.

    If $p_t \geq p_t'$, then we use a multiplicative Chernoff bound (\Cref{thm:multiplicative-chernoff}) to bound our probability of success. 
    Let $s_t$ be the number of edges we sampled for $E_{\tilde \alpha}^t$ prior to any sub-sampling.
    That is,
    we obtain $s_t$ by sampling with probability $p_t$ before \Cref{match:svt-s}.
    For $\psi := \eta \cdot E_t^* / \card{E_{\tilde\alpha}^t}$,
    \begin{align}
        \prob\left[s_t - p_t \cdot |E_{\tilde\alpha}^t| \geq \eta \cdot p_t \cdot E_t^* \right]
        &= \prob\left[s_t - p_t \cdot |E_{\tilde\alpha}^t| \geq \psi \cdot p_t\cdot |E_{\tilde\alpha}^t| \right] \\
        &\leq \exp\left( - \frac{\psi^2}{2+\psi}\cdot p_t\cdot |E_{\tilde\alpha}^t| \right) \\
        &\leq \exp\left( - \frac{\psi^2}{2+\psi}\cdot p_t'\cdot |E_{\tilde\alpha}^t| \right). &&p_t\geq p_t'
    \end{align}
    When $\psi \geq 1$,
    we have $\frac{\psi}{2+\psi}\geq \frac13$ and this is at most
    \begin{align}
        \exp\left( - \frac{\psi}{3}\cdot p_t'\cdot |E_{\tilde\alpha}^t| \right)
        &= \exp\left( - \frac{\eta E_t^*}3 p_t' \right) \\
        &\leq \exp\left( -\Omega(\eps^{-1}\eta^{-1} \log^2(n)) \right) &&\text{definition of $p_t'$, sufficiently large $a_3$} \\
        &= \frac1{\poly(n)}.
    \end{align}
    When $\psi \in (0, 1)$,
    this is at most
    \begin{align}
        \exp\left( - \frac{\psi^2}{3}\cdot p_t'\cdot E_t^* \right)
        &\leq \exp\left( - \eta^2 E_t^* p_t' \right) \\
        &\leq \exp\left( -\Omega(\eps^{-1}\log^2(n)) \right). &&\text{definition of $p_t'$, sufficiently large $a_3$} \\
        &= \frac1{\poly(n)}
    \end{align}
    Since we need only consider the case of $\psi\in (0, 1)$ for the lower bound,
    an identical calculation yields the same concentration bound above but for the lower bound.
    All in all,
    $\frac{s_t}{p_t}\in [|E_{\tilde\alpha}^t| \pm \eta E_t^*]$ with high probability.

    Now, we consider the case when $p_t < p_t'$. We show that with 
    probability at least $1 - 1/\poly(n)$, 
    this case \emph{does not occur}. We prove this via 
    induction on $t$. For $t=1$, $p_t = p_t'$ trivially since both
    equal $1$. Now, we assume our claim holds for $t$ and show it holds
    for $t+1$. 
    First, note
    that $E^*_t$ cannot decrease so $p_t'$ cannot increase as $t$ increases.
    Then, by \Cref{match:svt-s necessary},
    $p_{t+1} < p_{t+1}'$ only if we sample more than $\frac{a_3'\log^2(n)}{2\eps \eta^2}$
    elements of $E_{\tilde\alpha}^{t}$ at probability $p_{t} = p_{t}'$ before \Cref{match:svt-s}.
    However,
    we have just shown that with high probability,
    \begin{align}
        s_t\leq (1+\eta) p_t E_t^*
        \leq 2p_t' E_t^*
        \leq 2\tau
        < \frac{a_3\log^2(n)}{2\eps \eta^2}
    \end{align}
    by the definition of $p_t'$, $\tau$, and assuming a sufficiently large $a_3$.
    Hence this event does not occur with high probability.

    Finally,
    the estimate SVT ensures that with high probability,
    \begin{align}
        s_t &\leq p_t (1+\eta)^{j_t} + O(\eps^{-1}\eta^{-1}\log^2(n)) \\
        s_t &> p_t(1+\eta)^{j_t-1} - O(\eps^{-1}\eta^{-1}\log^2(n)).
    \end{align}
    But we wish to approximate $s_t/p_t$,
    and dividing the inequalities by $p_t$ yield
    \begin{align}
        \frac{s_t}{p_t} &\leq (1+\eta)^{j_t} + \frac{O(\eps^{-1}\eta^{-1}\log^2(n))}{p_t} \\
        \frac{s_t}{p_t} &> (1+\eta)^{j_t-1} - \frac{O(\eps^{-1}\eta^{-1}\log^2(n))}{p_t}.
    \end{align}
    If $i_t = 1$ then $p_t = 1$ and we are done.
    Thus consider the case where $i_t\geq 2$.
    We have shown above that $p_t\geq p_t' = 2^{-i_t}$ and $2^{i_t-1}\tau \leq E_t^*$ for $i_t\geq 2$.
    Thus in this case,
    \begin{align}
        \frac{O(\eps^{-1}\eta^{-1}\log^2(n))}{p_t} 
        &\leq \frac{O(\eps^{-1}\eta^{-1}\log^2(n)) E_t^*}{\tau} \\
        &= O(\eps^{-1}\eta^{-1}\log^2(n)) E_t^*\cdot \frac{\eps \eta^2}{d\log^2(n)} \\
        &\leq \eta E_t^*
    \end{align}
    assuming a sufficiently large constant $d$.
    
    Thus, combining the case work above with our concentration bounds on $s_t/p_t$, 
    we see that $(1+\eta)^{j_t}$ is a $\left( 1+ \eta, O\left( \eps^{-1}\eta^{-1}\log^2 n \right) \right)$-approximation of $E_t^*$.

\end{proof}

\begin{lemma}\label{lem:space-mm}
    \cref{alg:matching-sampling} requires $O\left( \frac{\log^2(n)\log(\tilde{\alpha})}{\epsilon \eta^2} \right)$ space, with high probability.
\end{lemma}

\begin{proof}
In terms of space complexity, we need $O\left( \frac{\log^2(n)\log(\tilde{\alpha})}{\epsilon \eta^2} \right)$ space to store the sampled edges and the two counters per edge. 
We check when the number of sampled edges exceeds $\frac{a_3 \log^2(n)}{\eps\eta^2}$ in~\cref{match:half-prob} using an SVT check. 
By the guarantees of SVT, the number of sampled edges is at most $\frac{a_3 \log^2(n)}{\eps\eta^2} + \frac{c'\log(n)}{\eps} = O\left(\frac{\log^2(n)}{\eps\eta^2}\right)$ for some large enough constant $c' > 0$,
with high probability, before it exceeds the threshold. Thus, with high probability, we store $O\left(\frac{\log^2(n)}{\eps\eta^2}\right)$ sampled edges and each 
edge uses $O(\log(\tilde{\alpha}))$ bits to store the counters, resulting in our stated space bounds.
\end{proof}

Now, combining~\cref{lem:e-alpha-mg,lem:sampling-approx,lem:space-mm} gives us our final accuracy guarantee in~\cref{thm:edge-dp-matching}.

\section{Proofs of \texorpdfstring{\Cref{thm:matching lower bound small epsilon,thm:connected components lower bound small epsilon}}{Theorems}}\label{apx:lower-bounds}
We now fill in the omitted details from \Cref{sec:lower-bounds}.

\subsection{Inner Product Queries}\label{apx:inner product queries}
Given a private database $y\in \set{0, 1}^n$
and a set of $k$ linear queries $q^{(j)}\in \set{0, 1}^n, j\in [d]$,
the \emph{inner product problem} asks for the values
\[
    \iprod{y, q^{(j)}}\in \Z_{\geq 0}\,, \qquad \forall j\in [k]\,.
\]
We use the following lower bound stated in \cite{jain2023counting}
that is based on \cite{Dinur2003Revealing,Dwork2007Price,MMNW11,De2012Lower}.
\begin{theorem}[Theorem 4.5 in \cite{jain2023counting}]\label{thm:inner product lower bound}
    There are constants $c_1, c_2 > 0$ such that, 
    for sufficiently large $n > 0$:
    if an algorithm $\mcal A$ answers $c_1n$ inner product queries within additive error $c_2\sqrt{n}$ with probability at least $0.99$,
    then $\mcal A$ is not $(1,\nicefrac13)$-DP.
\end{theorem}

\subsection{Maximum Cardinality Matching}\label{apx:matching lower bound}
We now restate and prove the lower bound for computing the size of a maximum matching in the fully-dynamic continual release model.
\matchingLowerBoundSmallEpsilon*

\subsubsection{Description of Reduction}
We reduce to the inner product query problem similar to \cite{jain2023counting}.
Let $n$ be the dimension of a secret dataset $y\in \set{0, 1}^n$.
Consider $3n$ vertices $V = \set{u_i, v_i, w_i: i\in [n]}$.
We use the first $\Theta(n)$ updates to encode the non-zero bits of $y$ as a matching between the vertices $u_i, v_i$,
i.e. we add the edge $\set{u_i, v_i}$ if and only if $y_i = 1$.
Note that the size of the maximum matching is precisely $\norm{y}_0$ by construction. 

Given linear queries $q^{(1)}\in \set{0, 1}^n$,
we make the following updates:
For each $i$ such that $q_i^{(1)} = 1$,
add the edge $\set{v_i, w_i}$.
Then the size of the maximum matching is $\norm{y\mid q^{(1)}}_0$,
where $a\mid b$ denotes the bitwise OR of $a, b\in \set{0, 1}^n$.
The inner product can then be recovered using the inclusion-exclusion principle:
\[
    \iprod{y, q^{(1)}}
    = \norm{y}_0 + \norm{q^{(1)}}_0 - \norm{y\mid q^{(1)}}_0.
\]
We can then delete the added edges and process the next query $q^{(2)}$
and so on.

Since emulating each query requires $\Theta(n)$ updates,
we can answer $\Theta(n)$ inner product queries after $T = \Theta(n^2)$ updates.
At a high level,
since the error lower bound for $O(n)$ inner product queries is $\Omega(\sqrt{n})$,
we have a lower bound of $\Omega(\min(n^{1/2}, T^{1/4}))$.

\begin{algorithm}[htp!]
\caption{Reduction from Inner Product Queries to Maximum Matching}\label{alg:inner product to matching}
\SetKwFunction{FAlg}{Alg}
\SetKwProg{Fn}{Function}{}{}
\Fn{\FAlg{private dataset $y\in \set{0, 1}^n$, public queries $q^{(1)}, \dots, q^{(k)}\in \set{0, 1}^n$, DP mechanism for matchings $\mcal M$}}{
  $V \gets \set{u_i, v_i, w_i: i\in [n]}$ ($3n$ vertices) \\
  $E_0 \gets \varnothing$ empty edge set \\
  $S^{(0)} \gets$ empty update stream of length $n$ \\
  $S^{(1)}, \dots, S^{(k)} \gets$ $k$ empty update streams each of length $2n$ \\
  \For{$i=1, 2, \dots, n$}{
    \If {$y_i = 1$} {
        $S_i^{(0)} \gets +\set{u_i, v_i}$ (insert edge $\set{u_i, v_i}$) \\
    }
  }
  \For{$j=1, 2, \dots, k$} {
    \For{$i=1, 2, \dots, n$} {
        \If{$q^{(j)}_i = 1$} {
            $S_{i}^{(j)}\gets +\set{v_i, w_i}$ \\
            $S_{n+i}^{(j)}\gets -\set{v_i, w_i}$ (delete edge $\set{v_i, w_i}$) \\
        }
    }
  }
  $S \gets S^{(0)} + S^{(1)} + \dots + S^{(k)}$ concatenation of all streams \\
  $r^{(0)}, r^{(1)}, \dots, r^{(k)} \gets \mcal A(S)$ answers to queries \\
  \For{$j=1, 2, \dots, k$} {
    Output $\norm{q^{(j)}}_0 + r_n^{(0)} - r_n^{(j)}$ as approximate answer to $\iprod{q^{(j)}, y}$ \\
    \Comment{$\iprod{q^{(j)}, y} = \norm{q}_0 + \norm{y}_0 - \norm{q^{(j)}\mid y}_0$} \\
  }
}
\end{algorithm}

\subsubsection{Proof of Lower Bound}
\begin{lemma}\label{lem:inner product to matching}
    Given an $\varepsilon$-DP mechanism for maximum matching in the continual release setting
    that outputs estimates within additive error at most $\zeta$ with probability $0.99$
    for a fully-dynamic stream of length $T\geq n+2nk$,
    \Cref{alg:inner product to matching} is an $\varepsilon$-DP mechanism for answering $k$ inner product queries within additive error $2\zeta$ with probability $0.99$.
\end{lemma}

\begin{pf}[\Cref{lem:inner product to matching}]
    Let $\mcal M$ be such a $\varepsilon$-DP mechanism
    and suppose we are given a private dataset $y\in \set{0, 1}^n$
    as well as public linear queries $q^{(1)}, \dots, q^{(k)}$.
    Let $S^{(0)}, S^{(1)}, \dots, S^{(k)}$ be the partial update streams as in \Cref{alg:inner product to matching}
    and
    \[
        (r^{(0)}, r^{(1)}, \dots, r^{(k)})
        = \mcal M(S^{(0)} + S^{(1)} + \dots + S^{(k)})
    \]
    be the corresponding stream of outputs when we run $\mcal M$ on the concatenation of partial streams.
    By assumption,
    $r^{{(0)}}_n$ is an estimate of $\norm{y}_0$ within additive error $\zeta$
    and $r^{(j)}_n$ is an estimate of $\norm{y\mid q^{(j)}}_0$ within additive error $\zeta$.
    Hence the value of
    \[
        \norm{q^{(j)}}_0 + r_n^{(0)} - r_n^{(j)}
    \]
    is an estimate of $\iprod{y, q^{(j)}}$ with additive error within $2\zeta$.
\end{pf}

\begin{theorem}\label{thm:matching lower bound}
    If $\mcal A$ is a $1$-DP mechanism that answers maximum matching queries on graphs with $n$ vertices in the continual release model within additive error $\zeta$
    with probability at least $0.99$
    for fully dynamic streams of length $T$,
    then
    \[
        \zeta
        = \Omega(\min(\sqrt{n}, T^{1/4})).
    \]
    Moreover,
    we may assume the graph is bipartite,
    has maximum degree 2,
    and arboricity 1.
\end{theorem}

\begin{pf}[\Cref{thm:matching lower bound}]
    Let $c_1, c_2$ be the constants from \Cref{thm:inner product lower bound}.
    Fix $n\in \N$ and suppose $T\geq n+2c_1n^2$.
    By \Cref{lem:inner product to matching},
    $\mcal A$ implies a $1$-DP mechanism for answering $k=c_1 n$ inner product queries within additive error $2\zeta$ with probability $0.99$.
    But then by \Cref{thm:inner product lower bound},
    we must have $\zeta \geq c_2\sqrt{n}/2$.

    Now consider the regime of $T < n+2c_1n^2$.
    The same argument holds for $\tilde n = \Theta(\sqrt{T})$ such that $\tilde n + 2c_1 \tilde n^2\leq T$
    to yield a lower bound of $\zeta \geq c_2\sqrt{\tilde n} = \Omega(T^{1/4})$.
\end{pf}

\begin{lemma}\label{lem:small epsilon matching to matching}
    Let $\varepsilon \in (0, 1)$,
    $\ell := \floor{\nicefrac1\varepsilon}$,
    $\tilde T\geq \ell$,
    $\zeta: \R\times \R\to \R$ be an increasing error function.
    Suppose there is an $\varepsilon$-DP mechanism $\mcal A$ for answering maximum matching queries on graphs of $\tilde n$ vertices in the continual release model
    that outputs estimates within additive error at most
    \[
        \ell\cdot \zeta(\nicefrac{\tilde n}\ell, \nicefrac{\tilde T}\ell)
        = O\left( \frac{\zeta(\varepsilon \tilde n, \varepsilon \tilde T)}\varepsilon \right)
    \]
    with probability $0.99$ for fully-dynamic streams of length $\tilde T$.

    Let $T = \nicefrac{\tilde T}\ell$,
    and $n = \nicefrac{\tilde n}\ell$.
    There is a $1$-DP mechanism for answering maximum matching queries on graphs of $n$ vertices in the continual release model
    that outputs estimates within additive error at most $\zeta(n, T)$ with probability $0.99$
    for fully-dynamic streams of length $T$.
\end{lemma}

\begin{pf}[\Cref{lem:small epsilon matching to matching}]
    Let $S$ be an edge update stream of length $T$ from a graph on $n$ vertices.
    
    We construct an edge update stream of length $\tilde T = \ell T$ from a graph on $\tilde n=\ell n$ vertices as follows:
    Duplicate each vertex $v$ in the graph $\ell$ times into $v_1, \dots, v_\ell$. 
    For each update $S_i\in \set{+\set{u, v}, -\set{u, v}, \perp}$ from $S$,
    define the length $\ell$ partial stream
    \[
        S^{(i)} :=
        \begin{cases}
            +\set{u_1, v_1}, +\set{u_2, v_2}, \dots, +\set{u_\ell, v_\ell}, &S_i = +\set{u, v} \\
            -\set{u_1, v_1}, -\set{u_2, v_2}, \dots, -\set{u_\ell, v_\ell}, &S_i = -\set{u, v} \\
            \perp, \perp, \dots, \perp, &S_i = \perp.
        \end{cases}
    \]
    Here $u_i, v_i$ are the copies of the original vertices $u, v$.
    Then we take the new stream $\tilde S := S^{(1)} + \dots + S^{(T)}$ to be the concatenation of all partial streams.

    Let $\tilde r = \mcal A(\tilde S)$ be the result of the algorithm on $\tilde S$
    and output $r = (\nicefrac{\tilde r_\ell}\ell, \nicefrac{\tilde r_{2\ell}}\ell, \dots, \nicefrac{\tilde r_{T\ell}}\ell)$ as the query answers to the original stream $T$.

    By construction,
    the size of the maximum matching $\mu(\tilde S_{i\ell})$ in the new stream at time $i\ell, i\in [T]$ is exactly $\ell\cdot \mu(S_i)$ 
    in the original stream at time $i$.
    Thus by assumption,
    our output has additive error within
    \[
        \frac{\zeta(\varepsilon\cdot \ell n, \varepsilon\cdot \ell T)}{\varepsilon \ell}
        = O( \zeta(n, T) )
    \]
    with probability $0.99$.
    Any neighboring streams from $S$ becomes $\ell$-neighboring streams from $\tilde S$
    and thus by group privacy,
    our output is $1$-DP.
\end{pf}

We are now ready to prove \Cref{thm:matching lower bound small epsilon}.
\begin{pf}[\Cref{thm:matching lower bound small epsilon}]
    In the regime $T\geq \nicefrac1\varepsilon$,
    combining \Cref{lem:small epsilon matching to matching} and \Cref{thm:matching lower bound} implies a lower bound of
    \[
        \frac1\varepsilon \Omega(\min(\sqrt{\varepsilon n}, (\varepsilon T)^{1/4}))
        = \Omega\left( \min\left( \sqrt{\frac{n}\varepsilon}, \frac{T^{1/4}}{\varepsilon^{3/4}} \right) \right).
    \]

    In the second regime $T<\nicefrac1{\varepsilon}, \nicefrac{n}2$,
    we prove a lower bound of $\Omega(T)$.
    Suppose towards a contradiction that there is an $\varepsilon$-DP mechanism in this regime with additive error within $\zeta = \nicefrac{T}4$.
    Let $r, \tilde r$ be the outputs of $\mcal A$ on the empty update stream of length $T$
    and the stream that adds an edge of a fixed perfect matching at every time step,
    respectively.
    By the accuracy of $\mcal A$,
    we have $\prob[r_T\leq \nicefrac{T}4] \geq 0.99$.
    By $\varepsilon$-DP and group privacy,
    we have
    \[
        \prob[\tilde r_T > \nicefrac{T}4]
        \leq e^{\varepsilon T} \prob[r_T > \nicefrac{T}4]
        \leq 0.01e
        < 0.99.
    \]
    But then since the true matching size in the second stream at time $T$ is $T$,
    $\mcal A$ does not have additive error at most $\zeta = \nicefrac{T}4$ with probability at least $0.99$.
    By contradiction.
    it follows that we must have $\zeta = \Omega(T)$.

    In the final regime $\nicefrac{n}2 \leq T < \nicefrac1{\varepsilon}$,
    we show a lower bound of $\Omega(n)$.
    Simply apply our second argument with $\tilde T = \nicefrac{n}2$ to arrive at a lower bound of $\Omega(\tilde T) = \Omega(n)$.
\end{pf}

\subsection{Connected Components}\label{apx:connected components lower bound}
We now prove \Cref{thm:connected components lower bound small epsilon},
whose statement we repeat below for convenience.
\connectedComponentsLowerBoundSmallEpsilon*

\subsubsection{Description of Reduction}
Using a similar technique,
we show a lower bound for the problem of privately estimating the number of connected components.
In particular,
we encode a single bit of a private database $y\in \set{0, 1}^n$ using a subgraph on 4 vertices $u, v, a, b$.
The edges $\set{u, a}, \set{v, b}$ always exist but $\set{a, b}$ exists if and only if the bit is non-zero.
Then notice that adding the edge $\set{u, v}$ decreases the number of connected components
if and only if the bit is zero.

\begin{algorithm}[htp!]
\caption{Reduction from Inner Product Queries to Connected Components}\label{alg:inner product to connected components}
\SetKwFunction{FAlg}{Alg}
\SetKwProg{Fn}{Function}{}{}
\Fn{\FAlg{private dataset $y\in \set{0, 1}^n$, public queries $q^{(1)}, \dots, q^{(k)}\in \set{0, 1}^n$, DP mechanism for connected components $\mcal M$}}{
  $V \gets \set{u_i, v_i, a_i, b_i: i\in [n]}$ ($4n$ vertices) \\
  $E_0 \gets \varnothing$ empty edge set \\
  $S^{(0)} \gets$ empty update stream of length $3n$ \\
  $S^{(1)}, \dots, S^{(k)} \gets$ $k$ empty update streams each of length $2n$ \\
  \For{$i=1, 2, \dots, n$}{
  $S_i^{(0)} \gets +\set{u_i, a_i}$ (insert edge $\set{u_i, a_i}$) \\
  $S_{n+i}^{(0)} \gets +\set{v_i, b_i}$ \\
    \If {$y_i = 1$} {
        $S_{2n+i}^{(0)} \gets +\set{a_i, b_i}$ \\
    }
  }
  \For{$j=1, 2, \dots, k$} {
    \For{$i=1, 2, \dots, n$} {
        \If{$q^{(j)}_i = 1$} {
            $S_{i}^{(j)}\gets +\set{u_i, v_i}$ \\
            $S_{n+i}^{(j)}\gets -\set{u_i, v_i}$ (delete edge $\set{u_i, v_i}$) \\
        }
    }
  }
  $S \gets S^{(0)} + S^{(1)} + \dots + S^{(k)}$ concatenation of all streams \\
  $r^{(0)}, r^{(1)}, \dots, r^{(k)} \gets \mcal A(S)$ answers to queries \\
  \For{$j=1, 2, \dots, k$} {
    Output $\norm{q^{(j)}}_0 + r^{{(j)}}_{n} - r^{(0)}_{3n}$ \\
  }
}
\end{algorithm}

\subsubsection{Proof of Lower Bound}
\begin{lemma}\label{lem:inner product to connected components}
    Given an $\varepsilon$-DP mechanism for connected components in the continual release setting
    that outputs estimates within additive error at most $\zeta$ with probability $0.99$
    for a fully-dynamic stream of length $T\geq 3n+2nk$,
    \Cref{alg:inner product to connected components} is an $\varepsilon$-DP mechanism for answering $k$ inner product queries within additive error $2\zeta$ with probability $0.99$.
\end{lemma}

\begin{pf}[\Cref{lem:inner product to connected components}]
    Let $\mcal M$ be such a $\varepsilon$-DP mechanism
    and suppose we are given a private dataset $y\in \set{0, 1}^n$
    as well as public linear queries $q^{(1)}, \dots, q^{(k)}$.
    Let $S^{(0)}, S^{(1)}, \dots, S^{(k)}$ be the partial update streams as in \Cref{alg:inner product to connected components}
    and
    \[
        (r^{(0)}, r^{(1)}, \dots, r^{(k)})
        = \mcal M(S^{(0)} + S^{(1)} + \dots + S^{(k)})
    \]
    be the corresponding stream of outputs when we run $\mcal M$ on the concatenation of partial streams.
    By assumption,
    $(2n-r^{{(0)}}_{3n})$ is an estimate of $\norm{y}_0$ within additive error $\zeta$
    and $(2n-r^{(j)}_n)$ is an estimate of $\norm{y\mid q^{(j)}}_0$ within additive error $\zeta$.
    Hence the value of
    \[
        \norm{q^{(j)}}_0 + (2n-r^{{(0)}}_{3n}) - (2n-r^{(j)}_n)
        = \norm{q^{(j)}}_0 + r^{{(j)}}_{n} - r^{(0)}_{3n}
    \]
    is an estimate of $\iprod{y, q^{(j)}}$ with additive error within $2\zeta$.
\end{pf}

\begin{theorem}\label{thm:connected components lower bound}
    If $\mcal A$ is a $1$-DP mechanism that answers connected component queries on graphs with $n$ vertices in the continual release model within additive error $\zeta$
    with probability at least $0.99$
    for fully dynamic streams of length $T$,
    then
    \[
        \zeta
        = \Omega(\min(\sqrt{n}, T^{1/4})).
    \]
    Moreover,
    we may assume the graph is bipartite,
    has maximum degree 2,
    and arboricity 2.
\end{theorem}

\begin{pf}[\Cref{thm:connected components lower bound}]
    Let $c_1, c_2$ be the constants from \Cref{thm:inner product lower bound}.
    Fix $n\in \N$ and suppose $T\geq 3n+2c_1n^2$.
    By \Cref{lem:inner product to connected components},
    $\mcal A$ implies a $1$-DP mechanism for answering $k=c_1 n$ inner product queries within additive error $2\zeta$ with probability $0.99$.
    But then by \Cref{thm:inner product lower bound},
    we must have $\zeta \geq c_2\sqrt{n}/2$.

    Now consider the regime of $T < 3n+2c_1n^2$.
    The same argument holds for $\tilde n = \Theta(\sqrt{T})$ such that $3\tilde n + 2c_1 \tilde n^2\leq T$
    to yield a lower bound of $\zeta \geq c_2\sqrt{\tilde n} = \Omega(T^{1/4})$.
\end{pf}

By following an identical argument to \Cref{lem:small epsilon matching to matching},
i.e. creating duplicate instances of the lower bound construction within a longer stream,
we derive a similar lower bound to \Cref{thm:matching lower bound small epsilon} 
that accounts for the privacy parameter $\varepsilon$.

\section{Private Static Densest Subgraph}
\begin{theorem}[Theorem 6.1 in \cite{DLL23}]\label{thm:improved private static densest subgraph}
    Fix $\eta \in (0, 1]$.
    There is an $\varepsilon$-edge DP densest subgraph algorithm
    that runs in polynomial time and $O(m+n)$ space
    and returns a subset $V'\sset V$ of vertices
    that induces a $O\left( 2, O(\varepsilon^{-1} \log n) \right)$-approximate densest subgraph
    with probability at least $1-O(1/\poly(n))$.
\end{theorem}

\begin{theorem}[Theorem 5.1 in \cite{DLRSSY22}]\label{thm:private static densest subgraph}
    Fix $\eta \in (0, 1]$.
    There is an $\varepsilon$-edge DP densest subgraph algorithm
    that runs in $O((m+n)\log^3 n)$ time and $O(m+n)$ space
    and returns a subset $V'\sset V$ of vertices
    that induces a $O\left( 1+\eta, O(\varepsilon^{-1} {\eta^{-3}}\log^4 n) \right)$-approximate densest subgraph
    with probability at least $1-O(1/\poly(n))$
    for any constant $\eta > 0$.
\end{theorem}

\section{The Density of the Densest Subgraph}\label{apx:densest value}
We now recall the following theorem about computing the exact density of a densest subgraph in the static setting.

\begin{restatable}{theorem}{densestValue}\label{thm:densest value}
    Given an undirected graph $G=(V, E)$,
    there is an algorithm \textsc{ExactDensity} that computes the exact density of the densest subgraph.
    Moreover,
    the algorithm terminates in $O(\card{E}^3\log(\card V))$ time
    and uses $O(\card{E})$ space.
\end{restatable}

We first describe the auxiliary flow graph describe in \cite{boob2020flowless, chekuri2022densest}.
Given an undirected graph $G = (V, E)$ and a value $\lambda > 0$,
we construct the following bipartite flow network $D = (N, A)$.
We use the terminology of nodes and arcs to distinguish from the original graph $G$.
$N$ consists of a source node $s$,
a node $a_e$ for every edge $e\in E$,
a node $a_v$ for every vertex $v\in V$,
and a sink node $t$.
There is a directed arc from $s$ to $a_e$ for every $e\in E$
with capacity 1,
directed arcs from $a_e$ to $a_v$
and from $a_e$ to $a_u$ for every $e = \set{u, v}\in E$,
both with infinite capacity,
and a directed arc from $a_v$ to $t$ for every $v\in V$
with capacity $\lambda$.

For every $U\sset V$ in the original graph,
the set of nodes
\[
    N_U :=
    \set{s}
    \cup \set{a_{uv}: u, v\in U}
    \cup \set{a_v: v\in U}
\]
is an $(s, t)$-cut with capacity
\[
    \card{E} - \card{E[U]} + \lambda\card{U}.
\]
Any minimum $(s, t)$-cut must be of this form (allowing for $U=\varnothing$).
This leads to the following observation.

\begin{proposition}[\cite{chekuri2022densest}]\label{prop:densest value maximum flow}
    Let $G=(V, E)$ be an undirected graph,
    $\lambda \geq \nicefrac{\card{E}}{\card V}$,
    and $D=(N, A)$ be the directed flow graph obtained from $G$.
    Let $N_U$ induce a minimum $(s, t)$-cut in $D$.
    Either the densest subgraph in $G$ has density at most $\lambda$
    and $\delta(N_U)$ has capacity $\card E$,
    or $G[U]$ has density strictly greater than $\lambda$
    and $\delta(N_U)$ has capacity strictly less than $\card E$.
\end{proposition}

\begin{pf}[\Cref{prop:densest value maximum flow}]
    For $U=\varnothing$,
    the cut induced by $N_\varnothing = \set{s}$ has capacity $\card{E}\leq \lambda\card{V}$.

    Let $U\sset V$.
    If the density of $G[U]$ is at most $\lambda$,
    then $\card{E[U]}\leq \lambda\card{U}$ so that the cut induced by $N_U$ has capacity at least $\card E$.
    If this holds for all $\varnothing\neq U\sset V$,
    then the cut induced by $N_\varnothing$
    with capacity $\card E$
    is a minimum cut
    and by the max-flow min-cut theorem,
    the maximum flow has value $\card E$.

    Otherwise,
    assume that the density of $G[U^*]$ is strictly greater than $\lambda$.
    Then $\card{E[U^*]} > \lambda\card{U^*}$ so that the cut induced by $N_{U^*}$ has capacity strictly less than $\card E$.
    Another application of the max-flow min-cut theorem concludes the proof.
\end{pf}
Now we recall the classical result of \textcite{dinitz2006dinitz}.
\begin{theorem}[\cite{dinitz2006dinitz}]\label{thm:dinic algorithm}
    Given directed graph $D=(N, A)$ 
    with a source-sink pair $s, t\in N$
    and non-negative arc capacities,
    there is an algorithm for computing an exact maximum $(s, t)$-flow.
    Moreover,
    the algorithm terminates in $O(\card{N}^2\card{A})$ time
    while using $O(\card N + \card A)$ space.
\end{theorem}
We are now ready to show \Cref{thm:densest value}.

\begin{pf}[\Cref{thm:densest value}]
    Without loss of generality,
    assume that $G$ is connected so that $\card E = \Omega(\card V)$.
    
    Construct the auxiliary flow network $D = (N, A)$ above for some guess $\lambda\geq \nicefrac{\card E}{\card V}$
    of the value of the densest subgraph.
    There are $2 + \card{V} + \card{E} = O(\card E)$ nodes
    and $\card E + 2\card E + \card V = O(\card E)$ arcs.
    We can check if the densest subgraph has value at most $\lambda$
    by computing the maximum flow in the $D$ using Dinic's algorithm \Cref{thm:dinic algorithm}.
    This is justified by \Cref{prop:densest value maximum flow}.
    Moreover,
    each flow computation also retrieves the vertices inducing a minimum $(s, t)$-cut by computing the vertices reachable from $s$ in the residual graph.
    The time and space complexity of computing the residual flow is dominated by the time and space complexity of the flow algorithm.
    Thus each flow computation either declares that the densest subgraph has density ``at most $\lambda$'',
    or produces some $U\sset V$ with density strictly more than $\lambda$.

    The value of the densest subgraph is guaranteed to lie in the interval $[\nicefrac{\card E}{\card V}, \nicefrac{(\card V-1)}2]\sset [\nicefrac12, \nicefrac{\card V}2]$,
    so we can binary search over $\lambda$ in $O(\log \nicefrac{\card V}{\alpha})$ time
    and constant space
    to produce an estimate of the optimal density within an additive error of $\alpha$.
    Note that the densest subgraph has density at least $\nicefrac{\card E}{\card V}$.
    
    We first use the flow gadget to check if the densest subgraph has density strictly greater than $\nicefrac{\card E}{\card V}$.
    Thus at every iteration,
    we maintain a vertex set $U$ and and upper bound $h\geq 0$
    such that the density $\rho(U)$ of $G[U]$ and the optimal density $\rho^*$ satisfies
    \[
        \rho(U)\leq \rho^* \leq h.
    \]
    Each iteration of the binary search takes time
    \[
        O(\card{N}^2\card A)
        = O(\card{E}^3)
    \]
    Thus the total running time is
    \[
        O\left( \card{E}^3\log(\nicefrac{\card V}\alpha) \right).
    \]
    Here $\alpha > 0$ is the accepted level of additive error.
    
    We remark that the density of a subgraph is a rational number whose numerator lies in $[\card E]$ and denominator in $[\card V]$.
    This means that the gap between the optimal density and the density of suboptimal subgraphs is at least $\nicefrac1{\card V(\card V-1)} \geq \nicefrac1{\card{V}^2}$.
    We can thus set $\alpha = \nicefrac1{\card{V}^2}$ so that the density of the maintained subgraph $U$ is optimal after termination.
    This brings the total running time to
    \[
        O\left( \card{E}^3\log(\card{V}^3) \right)
        = O\left( \card{E}^3\log(\card V) \right)
    \]
    Faster max-flow or min-cut algorithms lead to better running time.
    
    The auxiliary space is proportional to the size of the digraph and is thus
    \[
        O(\card{N} + \card{A})
        = O(\card E).
    \]
\end{pf}

\section{Public Bound on Arboricity}\label{sec:arboricity assumption}
Our edge-DP matching algorithm (\Cref{thm:edge-dp-matching}),
 requires a public data-oblivious
estimate $\tilde \alpha$ of the maximum arboricity $\alpha$ over the stream.
While our algorithms \emph{always} guarantee privacy,
their utility is only guaranteed when $\tilde \alpha\geq \alpha$.
One way to remove this assumption is to first execute an insertion-only continual release algorithm for estimating the arboricity
and then
run the desired algorithm for a second pass given the estimated arboricity.

\subsection{Notification of Failure}
In the one-pass setting,
the simple approach above fails since we require a bound on the maximum arboricity throughout the entire stream.
However,
it may still be desirable to notify the observer that the current arboricity of the graph exceeds the public bound
and thus the utility is no longer guaranteed.
For this purpose,
we can run a continual release algorithm for estimating the arboricity alongside the desired algorithm.
For the sake of simplicity,
we explain how we can use our $k$-core algorithm (\Cref{thm:kcore-formal}) to accomplish such a task.

Recall that the \emph{degeneracy} of a graph
is defined to be the least number $k$
such that every induced subgraph has a vertex of degree at most $k$.
It is known that this quantity is equal to the value of the largest $k$-core,
say denoted by $k_{\max}$.
Using the definition of the arboricity $\alpha$ as the minimum number of forests that partition the edge set,
we can show that
\[
    \frac12 k_{\max}
    \leq \alpha
    \leq k_{\max}.
\]
Hence our $k$-core algorithm immediately yields a $(5, O(\eps^{-1} \log^3(n)))$-approximation of the arboricity.
When this estimated arboricity (approximately) exceeds the public bound,
we can notify the observer that utility is no longer guaranteed
{by outputting ``FAIL''}.
Note that algorithms with better approximation and space guarantees that directly estimate $\alpha$ most likely exist but we describe the approach here using \Cref{thm:kcore-formal} for the sake of simplicity.

\subsection{Guessing the Arboricity}\label{apx:guess arboricity}
In the case of our matching algorithms (\Cref{sec:matching}),
we briefly sketch how to completely remove the dependence on a prior public $\tilde \alpha$ bound at the cost of additional space.
Suppose we are given an $\eps$-edge DP $(\beta, \zeta)$-approximation continual release algorithm $\mcal A_e$ for computing the arboricity that uses $O(S)$ space (e.g. see previous section).
We can guess $\tilde \alpha = 2^p$ for each $p=1, 2, \dots, O(\log n)$
and execute $O(\log n)$ instances of our edge-DP continual release matching algorithm (\Cref{thm:edge-dp-matching}) alongside $\mcal A$.
Thus we can identify the ``correct'' output from the parallel instances
corresponding to the current arboricity $\alpha_t$ up to a $O\left( 2\beta, O(\zeta) \right)$-approximation.

Plugging in our edge-DP matching algorithm (\Cref{thm:edge-dp-matching}) and $\mcal A_e$ yields an $\eps$-edge DP continual release algorithm
that returns a $\left( O(\beta \alpha_t + \zeta) , O(\eps^{-1}\poly\log(n)) \right)$-approximate estimate of the maximum matching size
using $O(S + \eps^{-1} \poly\log(n))$ space.

\end{document}